\documentclass[11pt,a4paper]{article}
\pdfoutput=1
\usepackage[left=2.3cm,right=2.3cm,top=2.3cm,bottom=2.3cm]{geometry}
\usepackage[ascii]{inputenc}
\usepackage[dvipsnames,svgnames]{xcolor}
\usepackage[bookmarksnumbered,linktocpage,hypertexnames=false,colorlinks=true,linkcolor=NavyBlue,urlcolor=NavyBlue,citecolor=ForestGreen,anchorcolor=green,breaklinks=true,pagebackref=true,pdfusetitle]{hyperref}
\usepackage{bbm,braket,microtype,mathrsfs,amsmath,amssymb,xcolor,amsthm,amsfonts,latexsym,mathtools,graphicx,enumitem,booktabs,bm,xspace,float,mathdots,caption,subcaption,ellipsis,mleftright}
\usepackage[T1]{fontenc}
\usepackage{textcomp}
\usepackage[english]{babel}
\usepackage[capitalize,nameinlink]{cleveref}
\usepackage{mathpazo}
\usepackage{dsfont}
\usepackage{scalerel}
\usepackage{comment}
\usepackage{algorithm}
\usepackage{algpseudocode}

\usepackage{makecell,multirow,tabularx,tabulary}

\usepackage[normalem]{ulem}
\newtheorem{thm}{Theorem}[section]\crefname{thm}{Theorem}{Theorems}
\newtheorem{lem}[thm]{Lemma}\crefname{lem}{Lemma}{Lemmas}\AddToHook{env/lem/begin}{\crefalias{thm}{lem}}
\newtheorem{prp}[thm]{Proposition}\crefname{prp}{Proposition}{Propositions}\AddToHook{env/prp/begin}{\crefalias{thm}{prp}}
\newtheorem{cor}[thm]{Corollary}\crefname{cor}{Corollary}{Corollaries}\AddToHook{env/cor/begin}{\crefalias{thm}{cor}}
\crefname{cnj}{Conjecture}{Conjectures}\AddToHook{env/cnj/begin}{\crefalias{thm}{cnj}}
\newtheorem*{thm*}{Theorem}
\theoremstyle{definition}
\crefname{dfn}{Definition}{Definitions}\AddToHook{env/dfn/begin}{\crefalias{thm}{dfn}}
\newtheorem{rem}[thm]{Remark}\crefname{rem}{Remark}{Remarks}\AddToHook{env/rem/begin}{\crefalias{thm}{rem}}
\crefname{exa}{Example}{Examples}\AddToHook{env/exa/begin}{\crefalias{thm}{exa}}
\numberwithin{equation}{section}
\newcommand{\Ex}{\mathop{\bf E\/}}
\DeclareMathOperator{\tr}{tr}
\DeclareMathOperator{\rk}{rk}

\DeclareMathOperator{\polylog}{polylog}
\DeclareMathOperator{\Sym}{Sym}

\DeclareMathOperator{\diag}{diag}

\DeclareMathOperator{\Lin}{L}

\DeclareMathOperator{\U}{U}
\DeclareMathOperator{\GL}{GL}
\DeclareMathOperator{\Sp}{Sp}
\DeclareMathOperator{\Mp}{Mp}
\DeclareMathOperator{\ASp}{ASp}

\DeclareMathOperator{\SO}{SO}
\DeclareMathOperator{\Spin}{Spin}
\DeclareMathOperator{\Pin}{Pin}
\DeclareMathOperator{\PSD}{P}
\DeclareMathOperator{\trcl}{\operatorname{T}}
\DeclareMathOperator{\bounded}{\operatorname{B}}

\DeclareMathOperator{\Span}{span}
\DeclareMathAlphabet{\pazocal}{OMS}{zplm}{m}{n}
\DeclareMathOperator{\Tr}{Tr}

\newcommand*{\coloneqq}{\mathrel{\vcenter{\baselineskip0.5ex \lineskiplimit0pt \hbox{\scriptsize.}\hbox{\scriptsize.}}} =}

\renewcommand{\O}{\operatorname{O}}
\renewcommand{\S}{\operatorname{S}}
\renewcommand{\L}{\operatorname{L}}
\DeclarePairedDelimiter\abs{\lvert}{\rvert}
\DeclarePairedDelimiter\norm{\lVert}{\rVert}
\DeclarePairedDelimiter\floor{\lfloor}{\rfloor}
\DeclarePairedDelimiter\ceil{\lceil}{\rceil}
\DeclarePairedDelimiter\parens{\lparen}{\rparen}

\allowdisplaybreaks[4]
\newcommand{\R}{\mathbbm R}
\newcommand{\RR}{\mathbbm R}
\newcommand{\C}{\mathbbm C}
\newcommand{\N}{\mathbbm N}

\newcommand{\E}{\mathcal E}
\newcommand{\cA}{\mathcal A}
\newcommand{\cB}{\mathcal B}

\newcommand{\cP}{\mathcal P}
\newcommand{\cQ}{\mathcal Q}
\newcommand{\cH}{\mathcal H}
\newcommand{\cK}{\mathcal K}
\newcommand{\cL}{\mathcal L}
\newcommand{\cR}{\mathcal R}
\newcommand{\cV}{\mathcal V}
\newcommand{\cW}{\mathcal W}
\newcommand{\cG}{\mathcal G}
\newcommand{\cT}{\mathcal T}

\newcommand{\so}{\mathfrak{so}}
\renewcommand{\sp}{\mathfrak{sp}}
\newcommand{\ot}{\otimes}

\newcommand{\eps}{\varepsilon}
\newcommand{\bit}{\{0,1\}}
\newcommand{\id}{I}
\newcommand{\idCh}{\mathcal I}
\newcommand*{\tran}{^{\mkern-1.5mu\mathsf{T}}}

\newcommand{\dd}{\mathrm{d}}
\newcommand{\bigO}{\mathcal O}
\newcommand{\ketbra}[2]{\mathinner{\lvert#1\rangle\!\langle#2\rvert}}
\newcommand{\proj}[1]{\mathinner{\lvert#1\rangle\!\langle#1\rvert}}

\begin{document}

\title{Optimal tomography\\of bosonic and fermionic Gaussian states}
\author{Senrui Chen\thanks{Caltech, Pasadena, USA, \href{mailto:csenrui@gmail.com}{csenrui@gmail.com}} \and Marco Fanizza\thanks{Inria, T\'el\'ecom Paris -- LTCI, Institut Polytechnique de Paris, Palaiseau, France, \href{mailto:marco.fanizza@inria.fr}{marco.fanizza@inria.fr}}\and Filippo Girardi\thanks{Scuola Normale Superiore, Pisa, \href{mailto:filippo.girardi@sns.it}{filippo.girardi@sns.it}} \and Ludovico Lami\thanks{Scuola Normale Superiore, Pisa, \href{mailto:ludovico.lami@gmail.com}{ludovico.lami@gmail.com}} \and Francesco Anna Mele\thanks{Scuola Normale Superiore, Pisa, \href{mailto:francesco.mele@sns.it}{francesco.mele@sns.it}} \and Michael Walter\thanks{LMU Munich \& Munich Center for Quantum Science and Technology (MCQST), University of Amsterdam \& QuSoft, \href{mailto:michael.walter@lmu.de}{michael.walter@lmu.de}} \and Freek Witteveen\thanks{CWI \& QuSoft, Amsterdam, \href{mailto:f.witteveen@cwi.nl}{f.witteveen@cwi.nl}}}
\date{}
\maketitle
\begin{abstract}
The sample complexity is the minimum number of copies required to learn an accurate classical description of a quantum state.
Bosonic and fermionic Gaussian quantum states are families of quantum states that play a key role in quantum science and technology, from quantum optics and many-body physics to quantum chemistry, quantum computing, and quantum information theory.
Despite their importance, their sample complexity had not been fully determined.
We settle this open problem and show that both bosonic and fermionic Gaussian states can be learned using a number of copies that scales quadratically in the number of modes, regardless of whether the state is pure or mixed, and independently of any energy bound on the state. We derive these results by using the representation theory of Gaussian unitaries and by putting forth a generalization of the random purification channel to this setting and beyond.
\end{abstract}

\begin{center}
    \textit{This work merges and extends \cite{mergedwork1} and \cite{mergedwork2}.}
\end{center}

\tableofcontents

\section{Introduction}
Consider the fundamental task of quantum state tomography~\cite{AnshuArunachalam2024}: given access to independent copies of an unknown quantum state $\rho$, determine an accurate estimate $\hat \rho$ of $\rho$.
One of the key questions to ask in this context is that of \emph{sample complexity}:
what is the minimal number of copies or samples~$n$ needed to achieve a desired accuracy~$\eps$ with probability at least $1-\delta$?
For general states on a $d$-dimensional Hilbert space, this question, and variants of it, have been extensively studied and are now well-understood.
The answer is that~\cite{ODonnell2016-1,Haah2016,pelecanos2025, scharnhorst2025optimal}
\begin{align}\label{eq:sample complexity tomography}
    n = \Theta\mleft(\frac{d^2 + \log(\delta^{-1})}{\eps^2} \mright)
\end{align}
copies of $\rho$ are necessary and sufficient to obtain an estimate~$\hat \rho$ with fidelity to the true state~$F(\rho, \hat\rho)^2 \geq 1 - \eps^2$ with probability at least $1 - \delta$.
By the Fuchs--van de Graaf inequalities~\cite{FuchsVanDeGraaf1999}, the latter fidelity guarantee also implies the more commonly used trace-distance guarantee \(\frac12\|\rho-\hat\rho\|_1\le \eps\), which has a strong operational meaning due to the Holevo--Helstrom theorem~\cite{Holevo1973Decision,Helstrom1976}.

When the unknown state $\rho$ is promised to have some structure, the task of tomography may become easier, i.e., may have lower sample complexity.
This is highly relevant, since
(1)~the dependence on~$d$ in \cref{eq:sample complexity tomography} quickly becomes intractable, as $d = 2^m$ scales exponentially with the number of qubits~$m$, while
(2)~most states that are relevant for practical applications in quantum science and technology are indeed far from generic.

Among the most fundamental classes of quantum states are the \emph{Gaussian states}, which exist in bosonic~\cite{BUCCO, BraunsteinVanLoock2005, weedbrook2012gaussian} and fermionic~\cite{Bravyi2005} systems.
These states are ground states or Gibbs states of quadratic Hamiltonians in the (fermionic or bosonic) creation and annihilation operators, and play an important role in physics, chemistry, quantum computing and quantum information.
Because of their central role in these fields, the notion of tomography of Gaussian states has also been studied extensively.
In fact, quantum state tomography first appeared in physics in the~1990s in the context of bosonic systems, with quantum optics experiments aimed at estimating states of light~\cite{SmitheyBeckRaymerFaridani1993, DAranio1995, WallentowitzVogel1995, LeonhardtPaul1995, BabichevAppelLvovsky2004, LvovskyRaymer2009, BUCCO}.
About a decade ago, the optimal performance of tomography began to be analyzed systematically within the rigorous framework of quantum learning theory~\cite{ODonnell2016-1, Haah2016, AnshuArunachalam2024}.
By contrast, the theory of quantum learning for Gaussian states, both for bosons and fermions, has only recently begun to be developed~\cite{aaronson2021efficient,ogorman2022fermionic,bittel2025optimalfermion,mele2025learning,bittel2025optimalboson,fanizza2024efficient, bittel2025energy,chen2026towards,mele2026advances}.

Gaussian states are completely characterized by their covariance matrix and mean~\cite{BUCCO,Bravyi2005}, which can easily be estimated by single-copy Gaussian measurements.
This suggests a natural class of strategies for learning Gaussian states~\cite{aaronson2021efficient, ogorman2022fermionic, bittel2025optimalfermion, mele2025learning, bittel2025optimalboson, fanizza2024efficient, bittel2025energy, chen2026towards}.
However, neither these natural strategies nor more involved algorithms~\cite{bittel2025energy, chen2026towards} have led to a conclusive answer to the question of the sample complexity of learning Gaussian states.
The main result of this work is to provide this answer.

\begin{thm*}[Main result]
    The sample complexity of learning a bosonic or fermionic Gaussian state on $m$ modes to accuracy $\eps$ in purified distance or trace distance with probability of error at most $\delta$ is
    \begin{align}\label{eq_main_res}
        n = \Theta\mleft(\frac{m^2 + \log(\delta^{-1})}{\eps^2}\mright).
    \end{align}
    This holds regardless of whether the state is promised to be pure or not.
\end{thm*}

\noindent
Remarkably, the sample complexity we find for Gaussian states has exactly the same form as the sample complexity in \cref{eq:sample complexity tomography} for learning arbitrary $d$-dimensional states, but with the Hilbert-space
dimension~$d$ replaced by the number of modes~$m$.
While for arbitrary pure states this scaling can be quadratically improved, this is \emph{not} the case for Gaussian states.
This is intuitive because an $m$-mode Gaussian pure state is still described by $\Theta(m^2)$ parameters.

In the above theorem, just as in the earlier discussion and throughout the paper, it is sufficient to prove the sample complexity upper bound for the purified distance guarantee
\begin{align*}
    P(\rho,\hat \rho):=\sqrt{1-F(\rho,\hat \rho)^2}\leq \eps,
\end{align*}
in other words, we require $F(\rho, \hat \rho)^2 \geq 1 - \eps^2$. 
By the Fuchs--van de Graaf inequality $T(\rho,\hat \rho) \leq P(\rho,\hat \rho)$, the result for the purified distance directly implies the sample complexity upper bound in trace distance. By the fact that purified distance and trace distance coincide for pure states, and the fact that for the sample complexity lower bound it suffices to restrict to pure states, the result for the purified distance implies the sample complexity bound with respect to trace distance.
The tomography schemes we design are not directly easy to implement. This is in contrast to schemes based on measuring the covariance matrix and mean, which can be done with low-complexity single-copy measurements, and are practical to use in experiments. Our result fully determines the information-theoretic sample complexity of the problem, and may also be a starting point to design schemes that combine optimal sample complexity with efficient implementations.
Our main result is the combination of \cref{thm:tomo fermion}, \cref{thm:tomo fermion lower bound}, \cref{thm:tomo boson} and \cref{thm:tomo boson lower bound}, which prove the fermionic and bosonic cases of the upper and lower bounds on the sample complexity.

\subsection{High-level overview and challenges}\label{sec:high level}


The idea behind the tomography protocol we consider, and more sophisticated versions of it, has been analyzed in detail for tomography of general finite-dimensional states in~\cite{pelecanos2025}.
Remarkably, while the analysis is \emph{easier} than for previous protocols, the resulting protocol is \emph{optimal} in terms of sample complexity and has an efficient implementation. It has two key ingredients:
\begin{enumerate}
    \item Pure states on $\C^d$ can be learned optimally by the \emph{Hayashi measurement}~\cite{hayashi1998asymptotic,hayashi2017group, harrow2013church}, which is a maximally symmetric POVM on $n$ copies of the state, where the POVM element corresponding to an estimate $\hat{\psi}$ is proportional to $\proj{\hat{\psi}}^{\ot n}$.
    \item The \emph{random purification channel}~\cite{tang2025conjugate,pelecanos2025,girardi2025random} is a quantum channel that on input $\rho^{\ot n}$, prepares a quantum state that can be interpreted as a mixture of states $\psi^{\ot n}$ where $\psi$ is a \emph{random} purification of $\rho$. Applying a pure state tomography protocol to this output yields an estimate of a random purification of $\rho$, and upon taking a partial trace, an estimate of $\rho$ itself. In other words, this reduces mixed-state tomography to pure-state tomography \cite{pelecanos2025}.
\end{enumerate}

Our main result, the optimal sample complexity for tomography of Gaussian states, requires three innovations.
\begin{enumerate}
    \item[1.] We analyze a \emph{Gaussian analogue of the Hayashi measurement} for pure state tomography, based on the fact that these states are the orbit of the group of Gaussian unitaries. This proves \cref{eq_main_res} for pure states.
    \item[2a.] We provide a generalization of the random purification channel to \emph{arbitrary symmetry} groups or algebras. This result is of independent interest beyond the case of Gaussian states. It is more general, but it also gives a simple and transparent proof of the standard random purification channel.
    As a direct application, our theorem allows us to reduce tomography of mixed fermionic Gaussian states to the case of pure states.
    \item[2b.] For the bosonic case, there is a fundamental technical challenge that has limited progress previous to this work, and that we overcome. The issue is that the infinite-dimensional nature of the underlying Hilbert spaces forbids the existence of a random purification channel. Indeed, the notion of random purification is itself problematic, because it is impossible to construct a Haar probability measure on the non-compact group of symplectic unitaries. One of our main conceptual contributions is the construction of a \emph{quasi-purification channel} that applies to general bosonic Gaussian states. While it does not have the same interpretation as preparing a random purification, it does allow reduction of mixed-state tomography to pure-state tomography.
\end{enumerate}
The high-level idea behind these innovations is that optimal quantum information processing protocols can often be derived from exploiting the symmetries of the problem. If the setting is that of $n$ identical copies of arbitrary $d$-dimensional quantum states, the relevant symmetries are the action of the symmetric group $S_n$, permuting copies, and the unitary group $\U(d)$, acting identically on each copy. Together, these symmetries decompose the $n$-copy Hilbert space according to \emph{Schur-Weyl duality}, and this decomposition guides the construction of sample-optimal protocols, in particular for tomography. When considering Gaussian states, the role of arbitrary unitaries is replaced by Gaussian unitaries, and instead of Schur-Weyl duality there exist \emph{Howe duality} \cite{howe1989remarks} and \emph{Kashiwara-Vergne duality} \cite{kashiwara1978segal} in the fermionic and bosonic settings, respectively. Here, the permutation symmetry is enhanced to an $\O(n)$-symmetry. We expect that this perspective will find further applications in the study of Gaussian quantum information theory.

\subsection{Related work}

\begin{table*}[!tp]\label{tab:summary}
    \centering
    \renewcommand{\arraystretch}{1.5}
    \begin{tabulary}{\textwidth}{@{} L C C @{}}
        \toprule
        Class of states & Prior bounds & This work \\
        \midrule
        Fermionic Gaussian
          & $\widetilde{\bigO}(m^4/\varepsilon^2)$~\cite{bittel2025optimalfermion} &  $\boldsymbol{\Theta(m^2/\eps^2)}$\\ 
          \midrule
         Pure fermionic Gaussian & ${\widetilde{\bigO}}(m^2/\varepsilon^2)$~\cite{zhao2024learning} & $\boldsymbol{\Theta(m^2/\varepsilon^2)}$\\        
        \bottomrule
        
        \multirow{2}{*}{Bosonic Gaussian}
          & $\widetilde{\bigO}(m^3/\varepsilon^2)$~\cite{bittel2025energy} &  $\boldsymbol{\bigO(m^2/\eps^2)}$\\ 
          & ${\Omega}(m^2/\varepsilon^2)$~\cite{chen2026towards} & $\Omega(m^2/\varepsilon^2)$\\
        \midrule
        \multirow{2}{*}{Pure bosonic Gaussian}
          &  $\widetilde{\bigO}(m^2/\varepsilon^2)$~\cite{chen2026towards} 
          & $\boldsymbol{\bigO(m^2/\eps^2)}$\\
          & ${\Omega}\left(\tfrac{m^2}{\eps^2\log\tfrac m\eps}\right)$~\cite{chen2026towards} 
          & $\boldsymbol{\Omega(m^2/\eps^2)}$\\
        \bottomrule
    \end{tabulary}

    \caption{ Comparison with previous best sample complexity bounds. Improvements are in bold font.
    Here $\widetilde O$, $\widetilde \Omega$ hide polynomial factors in $\log m$, $\log\log {E}$. }
    \label{tab:bounds_summary}
\end{table*}

Tomography of fermionic Gaussian states has been studied in previous work \cite{aaronson2021efficient,ogorman2022fermionic,bittel2025optimalfermion}. The previous best sample complexities were $\widetilde \bigO(m^2/\eps^2)$ for pure states\footnote{Since \cite{zhao2024learning} shows a sample complexity linear in $L$ for states prepared by $L$ two-qubit gates and all Gaussian fermionic states can be prepared using $\bigO(m^2)$ two-qubit gates.} \cite{zhao2024learning} and $\widetilde{\bigO}(m^4/\eps^2)$ for mixed states \cite{bittel2025optimalfermion}. Here $\widetilde{O}$ hides $\polylog$ dependence on $m$.
Our work improves this to~$\bigO(m^2)$ for both pure and mixed states, and determines the optimal dependence on the accuracy and probability of success.

Tomography of bosonic Gaussian states has also been studied extensively in previous work \cite{mele2025learning,bittel2025optimalboson,fanizza2024efficient,bittel2025energy,chen2026towards}, see \cite{mele2026advances} for a recent overview of quantum learning theory for bosonic states. We now describe the previous best sample complexity bounds.
For pure states, \cite{chen2026towards} derives an upper bound on the sample complexity of $\widetilde\bigO(m^2 / \eps^2)$, with a very mild dependence on an energy cut-off, and the same scaling for gauge-invariant (or passive) mixed Gaussian states.
For arbitrary mixed Gaussian states, the best previous sample complexity was found in \cite{bittel2025energy} as $\widetilde\bigO(m^3/\eps^2)$. Our work improves this to $\bigO(m^2/\eps^2)$.
Furthermore, \cite{chen2026towards} derives a $\Omega(m^2 / \eps^2)$ sample complexity lower bound for learning mixed Gaussian states and $\widetilde \Omega(m^2 / \eps^2)$ for pure Gaussian states, together with a $\widetilde{\Omega}(m^3/\eps^2)$ lower bound for adaptive Gaussian measurements; \cref{thm:tomo boson lower bound} improves by logarithmic factors for pure states (and uses a completely different argument). Here $\widetilde{\bigO}$ hides a $\polylog$ dependence on $m$ and on $\log E$, with $E$ being the average energy, while $\widetilde{\Omega}$ hides a $\polylog$ dependence on $m$.

The random purification channel was developed in~\cite{tang2025conjugate,pelecanos2025,girardi2025random,soleimanifar2022testing,chen2024local}, and used to reduce mixed-state tomography to pure-state tomography in~\cite{pelecanos2025}.
For mixed states, tomography protocols scaling optimally in $d$ and $\eps$ were first derived from representation-theoretic considerations in \cite{Haah2016,ODonnell2016-1}, and lower bounds can be found in \cite{Haah2016, scharnhorst2025optimal}.
The protocol based on the random purification channel is not only easier to analyze, but also has the advantage that the random purification channel can be implemented efficiently~\cite{bacon2006efficient, harrow2005applicationscoherentclassicalcommunication, burchardt2025high, girardi2026randomstinespringsuperchannelconverting, yoshida2026randomdilationsuperchannel}, leading to the first tomography scheme that is both optimal in sample complexity and computationally efficient.
The concept of random purifications has also been used for quantum channel tomography, where one can apply random purification to the Choi state of a quantum channel, with an improved sample complexity compared to previous channel tomography results~\cite{mele2025optimal,chen2025quantum}. An analogue of the purification channel exists for the Stinespring dilation of a quantum channel~\cite{chen2025quantum, girardi2026randomstinespringsuperchannelconverting, yoshida2026randomdilationsuperchannel}.
A version of the random purification channel for general symmetries was used to study cloning of structured quantum states \cite{bansal2026cloning}.

\subsection{Organization of the paper}
In the following, we start by presenting the random purification channel for general symmetries (\cref{sec:intro_general_pur}).
We then review the mathematical structure and symmetries of bosonic quantum systems (\cref{sec:church rotation invariant subspace}).
These ideas are then applied to our main results for fermionic and bosonic Gaussian states (\cref{sec:result_sample_optimal_tomography_of_fermionic_gaussian_states,sec:result_sample_optimal_tomography_of_bosonic_gaussian_states}).
We conclude with a discussion and open problems (\cref{sec:discussion_and_open_questions}).
The detailed technical proofs are delegated to the Supplementary Material.

\section{Result: Random purification channel for general symmetries}\label{sec:intro_general_pur}
This section describes an abstract random purification channel, that applies to general symmetries, and states 

Let $\cH$, $\tilde{\cH}$ be Hilbert spaces of the same finite dimension, with basis vectors~$\ket x$,~$\ket{\tilde{x}}$, respectively.
For an operator~$A$ on the first Hilbert space~$\cH$, denote by~$A\tran$ its transpose on the other Hilbert space~$\tilde{\cH}$, where we identify the operators through the choices of bases.
Also denote by~$\ket\Gamma = \sum_x \ket{x\tilde{x}} \in \cH \ot \tilde{\cH}$ the corresponding unnormalized maximally entangled state.
For any quantum state $\rho \in \S(\cH)$, we have a purification $\ket{\psi^\text{std}_\rho} \coloneqq (\sqrt\rho \ot I) \ket\Gamma$, called the \emph{standard purification}.%
\footnote{To obtain a truly canonical purification one can choose $\tilde{\cH}$ to be the dual (or the conjugate) Hilbert space of~$\cH$.}

\begin{thm}[Random purification for general symmetries, simplified]\label{thm:main simplified}
For any closed subgroup~$G \subseteq \U(\cH)$, there is a channel $\cP_G \colon \Lin(\cH) \to \Lin(\cH \ot \tilde{\cH})$ such that the following holds:
For any state~$\rho \in \S(\cH)$ such that~$\rho \in \C G$ (the span of~$G$), we have
\begin{align}\label{eq:twirl intro}
    \cP_G[\rho] = \int_G (I \ot g\tran) \psi^\text{std}_\rho (I \ot \bar g) \, \dd g.
\end{align}
If the Fourier transform for the action of~$G$ can be implemented efficiently, then so can the channel.
\end{thm}

\noindent
\Cref{thm:main simplified} is a simplified version of \cref{thm:main technical}, which gives a precise description of the channel and is naturally phrased in the language of $*$-algebras.

Here we sketch the construction; details and a full proof are deferred to \cref{sec:random purification general}.
Let
\begin{align*}
    \cH \cong \bigoplus_{\lambda \in \Lambda} \cL_\lambda \ot \cR_\lambda,
\end{align*}
where $\cL_\lambda $ is an irreducible representation of $G$ and $\cR_\lambda$ is its multiplicity space. If $\rho \in \C G$, then it has the form
\begin{align*}
    \rho \cong \bigoplus_{\lambda \in \Lambda} \rho_{\cL_{\lambda}}\ot\frac{\id_{\cR_\lambda}}{\dim\cR_{\lambda}} .
\end{align*}
If $\tilde{\cH}$ is a copy of $\cH$, and we rearrange
\begin{align*}
    \cH \otimes \tilde{\cH} &\cong \left(\bigoplus_{\lambda \in \Lambda} \cL_\lambda \ot \cR_\lambda\right) \otimes \left(\bigoplus_{\mu \in \Lambda} \tilde{\cL}_{\mu} \ot \tilde{\cR}_{\mu}\right)\\
    &\cong \bigoplus_{\lambda,\mu \in \Lambda}\left(\cL_\lambda \otimes  \tilde{\cL}_{\mu} \right)\ot \left( \cR_\lambda\ot \tilde{\cR}_{\mu}\right),
\end{align*}
then the standard purification of $\rho$ will be supported on $\bigoplus_{\lambda \in \Lambda}\left(\cL_\lambda \otimes  \tilde{\cL}_{\lambda} \right)\ot \left( \cR_\lambda\ot \tilde{\cR}_{{\lambda}}\right)\subseteq \cH\otimes \tilde\cH$ (so $\lambda = \mu$) and will take the form
\begin{align*}
    \psi_{\rho} = \bigoplus_{\lambda \in \Lambda} (\psi_{\rho_{\lambda}})_{\cL_{\lambda}\tilde{\cL}_{\lambda}}\ot\frac{\ketbra{\Gamma}{\Gamma}_{\cR_\lambda\tilde{\cR}_{\lambda}} }{\dim\cR_{\lambda}},
\end{align*}
where $\frac{\ketbra{\Gamma}{\Gamma}_{\cR_\lambda\tilde{\cR}_{\lambda}} }{\dim\cR_{\lambda}}$ is a maximally entangled state on $\left(\cR_\lambda \ot \tilde{\cR}_\lambda\right)$ and $(\psi_{\rho_{\lambda}})_{\cL_{\lambda}\tilde{\cL}_{\lambda}}$ is an unnormalized purification of $\rho_{\cL_{\lambda}}=\Tr_{\tilde{\cL}_{\lambda}}[(\psi_{\rho_{\lambda}})_{\cL_{\lambda}\tilde{\cL}_{\lambda}}]$.
The twirling map acting on $\tilde{\cH}$ will correspond to depolarizing channels on the $\tilde \cL_\lambda$, and hence
\begin{align*}
    \int_G (I \ot g\tran) \psi^\text{std}_\rho (I \ot \bar g) \, \dd g= \bigoplus_{\lambda \in \Lambda} \rho_{\cL_\lambda}\otimes \frac{I_{\tilde{\cL}_{\lambda}}}{\dim \cL_{\lambda}}\ot \frac{\ketbra{\Gamma}{\Gamma}_{\cR_\lambda\tilde{\cR}_{\lambda}} }{\dim\cR_{\lambda}}.
\end{align*}
It is now clear that the right-hand side can also be implemented as a quantum channel $\cP_G$ acting on $\rho$. 

Notably, the channel $\cP_G$ maps arbitrary states into the symmetric subspace for the commutant algebra~$G'$.
It can be interpreted as the Petz recovery map for the partial trace and the maximally mixed state on this symmetric subspace.

It is easy to see that the random purification channel of~\cite{tang2025conjugate,pelecanos2025,girardi2025random} in \cref{eq:random purification intro} is a special case of the above theorem.
Indeed, let $\cH = \tilde{\cH} = (\C^d)^{\ot n}$, with the computational basis. 
If~$\rho = \sigma^{\ot n}$ is an IID state, then~$\rho \in \C \{ U^{\ot n} : U \in \U(d) \}$.
This follows from Schur--Weyl duality, but also from the following easier argument (see, e.g., \cite{symqi}). As a vector space, $\C \{ U^{\ot n} : U \in \U(d) \}$ contains arbitrary limits, hence in particular the Lie algebra action, therefore also its complexification, and thus also $g^{\ot n}$ for $g \in \overline{\GL(d)} = \C^{d \times d}$.
Because the unitary group is closed under transposition, the standard purification of an $n$-th tensor power is the $n$-th tensor power of the standard purification, and purifications are unique up to unitaries, we obtain the following as a direct consequence of \cref{thm:main simplified}:

\begin{cor}[\cite{tang2025conjugate,pelecanos2025,girardi2025random}]\label{cor:og channel}
There is a channel~$\cP^\text{Purify}_{n,d}$ from~$(\C^d)^{\ot n}$ to the symmetric subspace of~$(\C^d \ot \C^d)^{\ot n}$ such~that
\begin{align*}
    \cP^\text{Purify}_{n,d}[\sigma^{\ot n}]
= \int_{\U(d)} (I \ot U^{\ot n}) \psi^{\text{std}}_{\sigma^{\ot n}} (I \ot U^{\dagger,\ot n}) \ \dd U
= \Ex_{\psi_\sigma \text{ purification of } \sigma } \psi_\sigma^{\ot n}
\end{align*}
The first equation holds not just IID states~$\sigma^{\ot n}$, but for arbitrary permutation-invariant states~$\rho$.
\end{cor}

The fact that this holds for arbitrary permutation-invariant states is due to Schur--Weyl duality, which asserts that $\C \{ U^{\ot n} : U \in \U(d) \}$ is the commutant of the permutation operators. Moreover, the channel can be efficiently implemented by means of the quantum Schur transform~\cite{bacon2006efficient,burchardt2025high}.

The usefulness of the random purification channel in quantum tomography is that, by turning IID copies of a mixed state into IID copies of one of its purifications (albeit randomly chosen), it allows us to reduce the problem of mixed state learning to that of pure state learning.
This fact does not have a general analogue for arbitrary symmetries, but it is adequately reproduced in the specific cases of fermionic and bosonic Gaussian states. The reason for this is to be found in dualities that characterize the commutant of Gaussian unitaries, together with the fact that IID pure Gaussian states are invariant under the action of the commutant.

\section{Gaussian states and the church of the rotation-invariant subspace}\label{sec:church rotation invariant subspace}
Bosonic and fermionic quantum systems are described in terms of creation and annihilation operators $a_i^\dagger$ and $a_i$, for $i = 1,\dots,m$, where $m$ is the number of modes.
The associated Hilbert spaces are the symmetric and antisymmetric Fock spaces, which are isomorphic to $L^2(\R^m)$ and $(\C^2)^{\ot m}$, respectively.
A crucial role is played by Hamiltonians that are quadratic in the creation and annihilation operators.
Pure Gaussian states are ground states of such Hamiltonians, mixed Gaussian states are thermal states of quadratic Hamiltonians (and, in the bosonic case, limits thereof), and Gaussian unitaries are generated by quadratic Hamiltonians.
Gaussian states are fully determined by their mean vector and covariance matrix; detailed definitions are given in \cref{sec:fermions,sec:bosons}.

The random purification channel combines two important principles: the ``church of the larger Hilbert space''~\cite{Gottesman2000}, which replaces mixed states by purifications, and the ``church of the symmetric subspace''~\cite{harrow2013church}, which exploits the fact that $n$ identical copies of a pure state live in a small, highly symmetric subspace.
For arbitrary pure states this subspace is
\begin{align*}
    \Sym^n(\C^d)=\Span\{\ket{\psi}^{\ot n}:\ket{\psi}\in\C^d\},
\end{align*}
and it equals the subspace invariant under permutations of the $n$ copies.
It yields the identity
\begin{align}\label{eq:hayashi intro}
    \Pi_{n,d} = {n + d - 1 \choose n} \int \proj{\psi}^{\ot n} \dd \psi
\end{align}
where $\Pi_{n,d}$ is the projection onto the symmetric subspace, and the integral is the Haar integral over states in $\C^d$. This leads to the definition of the Hayashi measurement.
More generally, Schur-Weyl duality states that the actions of $\U(d)$ through $U^{\ot n}$, and $S_n$ through permutation of tensor factors, span each others commutants.
For this work, what is relevant is the tensor product action of \emph{Gaussian unitaries}.
Since this is strict subgroup of unitaries, we should find a \emph{larger} commutant, and indeed it turns out that it is spanned by an action that rotates modes between the copies.
The role of the symmetric subspace is played by a \emph{rotation-invariant subspace}: the span of tensor powers of pure Gaussian states.

The relevant dualities are analogues of Schur--Weyl duality.
The Gaussian unitaries are projective representations of $\SO(2m,\R)$ in the fermionic case and of $\Sp(2m,\R)$ in the bosonic case; equivalently, one can work with the double covers $\Spin(2m,\R)$ and $\Mp(2m,\R)$.
In the following, we will write $\Sp(2m)$, $\Mp(2m)$, $\SO(2m)$, and $\O(n)$ for these real Lie groups.
In the bosonic case, including displacements leads to the affine symplectic group $\ASp(2m)$, but the rotation-invariant subspace below concerns first the mean-zero orbit.
The associated dualities are Howe duality~\cite{howe1989remarks,howe1995perspectives} and Kashiwara--Vergne duality~\cite{saito1972representations,gelbart1973holomorphic,kashiwara1978segal}, see also~\cite{girardi2025gaussian} for a previous discussion in the context of quantum information theory.
For our purposes, the full decompositions are not needed, and we mostly use the summands  corresponding to the span of tensor products of Gaussian states (i.e. the rotation-invariant subspaces). We nevertheless explain these decompositions for context. 

\subsection{Fermions: Howe duality}
For $n$ copies of an $m$-mode fermionic system, write $c_{a,r}$ for the Majorana operator with mode index $a\in[2m]$ and copy index $r\in[n]$.
The diagonal action of $R\in\SO(2m)$ is
\begin{align*}
    c_{a,r}\mapsto \sum_{b=1}^{2m} R_{ba} c_{b,r},
\end{align*}
whereas the copy rotations $O\in\O(n)$ act as
\begin{align*}
    c_{a,r}\mapsto \sum_{s=1}^{n} O_{sr} c_{a,s}.
\end{align*}
These two actions commute (up to projective phases, or exactly after passing to the corresponding double covers).
Howe duality for this representation says that the two actions span each other's commutants.
Equivalently, there is a decomposition
\begin{align}\label{eq:howe intro}
    \cH^{\ot n}
    \cong
    \bigoplus_{\lambda\in\Lambda_{m,n}} \cV_\lambda\ot\cW_\lambda,
\end{align}
where the $\cV_\lambda$ are irreducible representations of $\Spin(2m)$ and the $\cW_\lambda$ are irreducible representations of $\Pin(n)$.
Recent quantum-information discussions of this fermionic commutant can be found in~\cite{sierant2026theory,braccia2026commutant}. Note that when identifying $\cH^{\ot n}$ with $n \times m$ modes, one has to take into account an isomorphism between standard and fermionic tensor products.

Let $\cG_m^+$ and $\cG_m^-$ denote the even- and odd-parity pure Gaussian states, respectively.
The two one-dimensional characters of $\O(n)$ relevant here are the trivial character and the determinant character.
With the convention $\chi_+(O)=1$ and $\chi_-(O)=\det O$, the two rotation-isotypic subspaces are
\begin{align}\label{eq:fermion rotation invariant subspace}
    \cV_{n,m}^{\pm}
    \coloneqq
    \{\ket v\in \cH_m^{\ot n}: R_O\ket v=\chi_\pm(O)\ket v\text{ for all }O\in\O(n)\}.
\end{align}
Equivalently, if one restricts from $\O(n)$ to $\SO(n)$, both spaces are invariant subspaces.
Howe duality identifies these spaces with the spans of tensor powers of fixed-parity pure Gaussian states:
\begin{align}\label{eq:fermion gaussian span}
    \cV_{n,m}^{\pm}
    =
    \Span\{\ket{\psi}^{\ot n}:\ket{\psi}\in\cG_m^{\pm}\}
\end{align}
where the sign $\pm$ refers to the parity of the single-copy Gaussian state.

Let $\Pi_{n,m}^{\pm}$ be the orthogonal projection onto $\cV_{n,m}^{\pm}$ and let $d_{n,m}=\dim\cV_{n,m}^{\pm}$, the two dimensions being equal.
Since each $\cV_{n,m}^{\pm}$ is an irreducible $\Spin(2m)$-module, Schur's lemma gives the Gaussian analogue of \cref{eq:hayashi intro}:
\begin{align}\label{eq:rot inv subspace integral fermions}
    \Pi_{n,m}^{\pm}
    =
    d_{n,m}\int_{\cG_m^{\pm}} \proj{\psi}^{\ot n}\,\dd\psi,
\end{align}
where $\dd\psi$ is the Gaussian-unitarily invariant probability measure on the fixed-parity orbit.
This identity is the starting point for the fermionic Hayashi measurement used in \cref{sec:result_sample_optimal_tomography_of_fermionic_gaussian_states}.

\subsection{Bosons: Kashiwara--Vergne duality}\label{sec:intro_KV}
For bosons, let $\cH=L^2(\R^m)$ and first restrict attention to mean-zero pure Gaussian states, i.e. the orbit of the vacuum under $\Mp(2m)$.
On $n$ copies we identify
\begin{align*}
    \cH^{\ot n}\cong L^2(\R^{m\times n}),
\end{align*}
whose phase space is $\R^{2m\times n}\cong \R^{2m\times}\otimes \R^{n}$.
The tensor-power action $U_S^{\ot n}$ of $S\in\Sp(2m)$ corresponds to left multiplication $x\mapsto (S\otimes I_n)x$ on this phase space, while $O\in\O(n)$ acts by right multiplication $x\mapsto (I_{2m}\otimes O)x$; denote the corresponding unitary by $R_O$.
These actions commute. Kashiwara--Vergne duality 
gives the multiplicity-free decomposition
\begin{align}\label{eq:kv intro}
    \cH^{\ot n}
    \cong
    \bigoplus_{\lambda\in\Sigma_{m,n}} \cV_\lambda\ot\cW_\lambda,
\end{align}
where the $\cV_\lambda$ are irreducible representations of $\Mp(2m)$ and the $\cW_\lambda$ are irreducible representations of $\O(n)$.
Thus the bosonic analogue of the symmetric subspace is the $\O(n)$-invariant summand
\begin{align}\label{eq:boson rotation invariant subspace}
    \cV_{n,m}
    \coloneqq
    \{\ket v\in\cH^{\ot n}:R_O\ket v=\ket v\text{ for all }O\in\O(n)\}.
\end{align}
Equivalently,
\begin{align}\label{eq:boson gaussian span}
    \cV_{n,m}
    =
    \Span\{\ket{\psi}^{\ot n}:\ket{\psi}\in\cG_m^0\},
\end{align}
where $\cG_m^0$ is the set of mean-zero pure bosonic Gaussian states.
This is an irreducible representation of $\Mp(2m)$.

There are two differences from the finite-dimensional symmetric subspace.
First, $\cV_{n,m}$ is infinite-dimensional.
Second, $\Sp(2m)$ is non-compact, so the Gaussian orbit carries an invariant measure but not a finite Haar probability measure.
Nevertheless, when $n\geq 2m+1$, the representation in \cref{eq:boson rotation invariant subspace} lies in the discrete series and has a formal degree $d_{n,m}$.
With the invariant measure normalized by this formal degree, one has the resolution of the identity
\begin{align}\label{eq:rot inv subspace integral bosons}
    \Pi_{n,m}
    =
    d_{n,m}\int_{\cG_m^0}\proj{\psi}^{\ot n}\,\dd\psi,
\end{align}
where $\Pi_{n,m}$ projects onto $\cV_{n,m}$.
This is the bosonic replacement for the symmetric-subspace identity and underlies the covariant measurement in \cref{sec:result_sample_optimal_tomography_of_bosonic_gaussian_states}.
The non-compactness is also why the random purification channel for bosonic Gaussian states needs a modified construction described in \cref{sec:result_sample_optimal_tomography_of_bosonic_gaussian_states}.

\subsection{Bosons: \texorpdfstring{$\U(n)\times\U(m)$}{U(n) x U(m)} duality}
There is a different duality for bosons where we do get compact groups.
Passive Gaussian unitaries preserve total photon number and form a compact group $\U(m)$ on $m$ modes.
For $n$ copies, the Hilbert space can be viewed as the bosonic Fock space over $\C^m\ot\C^n$, and passive transformations on all $mn$ modes give a representation of $\U(mn)$.
The subgroups $\U(m)$ and $\U(n)$ act on the mode and copy indices, respectively, and commute.
On the $N$-photon sector, compact Howe duality gives a decomposition
\begin{align}\label{eq:unitary intro}
    \Sym^N(\C^m\ot\C^n)
    \cong
    \bigoplus_{\substack{\lambda\vdash N\\ \ell(\lambda)\leq \min\{m,n\}}}
    \cV_\lambda^{(m)}\ot\cW_\lambda^{(n)},
\end{align}
where $\cV_\lambda^{(m)}$ and $\cW_\lambda^{(n)}$ are irreducible polynomial representations of $\U(m)$ and $\U(n)$, respectively~\cite{howe1989remarks,RevModPhys.84.711,rowe_simple_2011}.
Equivalently,
\begin{align*}
    L^2(\R^{mn})
    \cong
    \bigoplus_{N=0}^{\infty}
    \bigoplus_{\substack{\lambda\vdash N\\ \ell(\lambda)\leq \min\{m,n\}}}
    \cV_\lambda^{(m)}\ot\cW_\lambda^{(n)}.
\end{align*}
This is used to derive the bosonic random purification channel for gauge-invariant, or number preserving, Gaussian states, stated in \cref{cor:bosons}.

\section{Result: Sample-optimal tomography of fermionic Gaussian states}\label{sec:result_sample_optimal_tomography_of_fermionic_gaussian_states}
We use the tomography scheme described in \cref{sec:high level} as a template to derive optimal tomography for fermionic Gaussian states.
We explain the tomography procedure here, and leave detailed proofs and calculations to \cref{sec:fermions}.
First of all, using the rotation-invariant subspaces, from \cref{eq:rot inv subspace integral fermions} we derive the Hayashi measurement for Gaussian pure state tomography.
Its performance is determined by ratios of dimensions of the spaces $\cV_{m,n}^{\pm}$.
These spaces are irreducible representations of $\so(2m)$ and their dimensions $d_{n,m}$ are computed using Lie theory.
We show that if $\hat{\psi}$ is the outcome of the Hayashi measurement applied to $n$ copies of a pure Gaussian state $\psi$, the expected overlap satisfies
\begin{align*}
    \Ex \abs{\braket{\psi | \hat{\psi}}}^{2(k - n)} = \frac{d_{n,m}}{d_{k,m}} = \prod_{1 \leq i < j \leq m} \frac{2m + n - (i+j)}{2m + k - (i+j)}.
\end{align*}
From this, by an application of Markov's inequality, we find that $n = \bigO\mleft((m^2 + \log(\delta^{-1}))/ \eps^2 \mright)$ copies suffice for accuracy $\eps$ and probability of failure at most $\delta$.

To deal with mixed states, we specialize the random purification channel of \cref{thm:main simplified} to this setting, to obtain a quantum channel that sends $n$ copies of a fermionic Gaussian state to a random fermionic Gaussian purification:\footnote{Just like \cref{cor:og channel}, which applies to arbitrary permutation-invariant states, \cref{cor:fermions} generalizes to states invariant under an action of the orthogonal group~$\O(n)$, as a consequence of Howe duality.}

\begin{cor}[Fermionic Gaussian random purification]\label{cor:fermions}
For all $m $ and $n$, there is a channel $\cP_{n,m}^\text{Fermi}$ from $(\bigwedge \C^m)^{\ot n}$ to $(\bigwedge \C^{2m})^{\ot n}$ with the following property:
for any fermionic Gaussian state~$\sigma$ on $\bigwedge \C^m$,
\begin{align*}
    \cP_{n,m}^\text{Fermi}[\sigma^{\ot n}]
= \Ex_{\psi_\sigma \text{ Gaussian purification of } \sigma } \psi_\sigma^{\ot n}
= \int_{\SO(2m)} (I \ot U_R^{\ot n}) \psi^{\text{std},\ot n}_\sigma (I \ot U_R^{\dagger,\ot n}) \ \dd R,
\end{align*}
where the expectation is taken over the Gaussian purifications obtained by applying a uniformly random Gaussian unitary~$U_R$ to the standard purification, as in the right-hand side formula.
\end{cor}

We use \cref{cor:fermions} to construct a tomography protocol for fermionic Gaussian states, by composing the Hayashi measurement for pure state tomography with the random purification channel.

\begin{algorithm}
\caption{Learning fermionic Gaussian states}\label{alg:learn fermionic gaussian}
\begin{algorithmic}[1]
\Require $n$ copies of unknown Gaussian state $\sigma$.
\Ensure Estimate $\hat{\sigma}$ of $\sigma$.

\State Apply the fermionic Gaussian random purification channel to
$\sigma^{\otimes n}\rightarrow \cP_{n,m}^\text{Fermi}[\sigma^{\ot n}]$
\State Perform the $\SO(2m)$-Hayashi measurement on the output, $\rightarrow$ estimate $\hat{\psi}_{AA'}$
\State \Return $\hat \sigma = \Tr_{A'}[\hat{\psi}_{AA'}]$
\end{algorithmic}
\end{algorithm}

This yields the following result.

\begin{thm}[Fermionic Gaussian tomography, upper bound]\label{thm:tomo fermion}
There exists a tomography protocol for $m$-mode fermionic Gaussian states~$\sigma$ that outputs an estimate~$\hat\sigma$ with purified distance at most~$P(\sigma,\hat \sigma)\leq \eps$ with probability at least $1-\delta$, using $\bigO((m^2 + \log(\delta^{-1}))/\eps^2)$ copies of~$\sigma$.
\end{thm}

This sample complexity holds for both pure and mixed states. We briefly comment on the reason for this. For general states, pure state tomography scales with the dimension of the Hilbert space $d$. Taking a purification increases the dimension to $d^2$. In the Gaussian case, pure state tomography scales with $m^2$, where $m$ is the number of modes. The number of modes of the purification is $2m$, so this only gives a multiplicative constant overhead.

We also show that this sample complexity is optimal, already in the case of pure states.

\begin{thm}[Fermionic Gaussian tomography, lower bound]\label{thm:tomo fermion lower bound}
Suppose there exists a tomography protocol for $m$-mode pure fermionic Gaussian states~$\psi$ that outputs an estimate~$\hat\psi$ with purified distance at most~$P(\psi,\hat \psi)\leq \eps$  with probability at least $1-\delta$ for any pure Gaussian state. Then this protocol needs to use at least $\Omega((m^2 + \log(\delta^{-1}))/\eps^2)$ copies of~$\psi$.
\end{thm}
This result follows in the same way as for tomography of general pure states: by averaging over random Gaussian states, we can relate to a ratio of dimensions of the spaces $\cV_{m,n}^{\pm}$.

\section{Result: Sample-optimal tomography of bosonic Gaussian states}
\label{sec:result_sample_optimal_tomography_of_bosonic_gaussian_states}
We also derive optimal sample complexity results for bosonic Gaussian states, which we report here, with the full derivations in \cref{sec:bosons}.
The Hilbert space of $m$ bosonic modes is the infinite dimensional space $\cH = \L^2(\R^m)$, and the pure mean zero Gaussian states are the orbit of the Gaussian unitaries, a projective representation of the symplectic group $\Sp(2m)$.
One important challenge is that the symplectic group is non-compact.
This means that there is no invariant probability measure---in particular, the notion of a uniformly random Gaussian unitary and twirl as in \cref{eq:twirl intro} cannot be interpreted as a convex combination of multi-copy states.
In particular, the random purification channel is not well-defined.

However, an invariant infinite-volume measure exists. Fortunately, if the number of copies $n$ is large enough, $n \geq 2m+1$, the product action of Gaussian unitaries is well-behaved, and we show that a variant of the same template as for general tomography and fermionic Gaussian tomography can be made to work.
The first step is to understand that a version of the Hayashi measurement is well defined when $n \geq 2m+1$: there exist a function $d_{n,m}$ such that $d_{n,m}\cdot d\mu(S) (U_{S}\ketbra{0}{0}U_S^{\dagger})^{\otimes n}$ is a POVM on the rotation-invariant space
\begin{align*}
    \mathcal V_{n,m} = \text{span} \{ \ket{\psi}^{\ot n} \, : \, \ket{\psi}  \text{ mean zero Gaussian state on } \cH \}.
\end{align*}
implementing a pretty good measurement on pure, zero-mean Gaussian states. The parameter $d_{n,m}$ is an example \textit{formal degree} of an irreducible representation of $\mathrm{Mp}(2m)$, which generalizes the role of the dimension of an irreducible representation of a compact group~\cite{harish1956representations,harish-chandra1966II,knapp2001representation}. While it is not a dimension of a Hilbert space (for non-compact groups, there are no non-trivial finite-dimensional unitary irreducible representations), it nevertheless determines the performance of the POVM for pure state learning, as there is still a meaningful `ratio of dimensions.' The key result is the computation of the following identity for the expected overlap of the estimate $\hat{\psi}$ of the Hayashi measurement with the input $\psi$:
\begin{align*}
    \Ex \abs{\braket{\psi | \hat{\psi}}}^{2(k - n)} = \frac{d_{n,m}}{d_{k,m}} = \prod_{1 \leq i \leq j \leq m} \frac{n - (i+j)}{k - (i+j)},
\end{align*}
which implies, by Markov's inequality applied to the moments of the fidelity, that $n
        =
        \bigO\mleft(
            \frac{m^2+\log(\delta^{-1})}{\eps^2}
        \mright)$    copies of a pure bosonic Gaussian state suffice to obtain that
    $ \abs{\braket{\psi|\hat\psi}}^2
        \geq
        1-\eps^2$ with probability at least \(1-\delta\).
The existence of covariant POVMs for generalized coherent states in discrete series representations was observed already by Hayashi himself~\cite{hayashi2017group}, but we are not aware of an explicit evaluation of the estimation performance for the case at hand.

The second key point we need to resolve is how to adapt the construction in~\eqref{thm:main simplified}. In particular, it turns out that directly following that construction would require one to prepare maximally mixed states on the irreducible representations $\cV_\lambda$ of the metaplectic group, but infinite-dimensional maximally mixed states are not well-defined.
We overcome this challenge by arguing that for the reduction from mixed-state tomography to pure-state tomography one can equally well use a channel that prepares some arbitrary, but fixed, state instead of a maximally mixed state on the $\cV_\lambda$.

We now explain this in more detail.
The Kashiwara--Vergne duality in~\cref{eq:kv intro} implies, as desired, that a standard, Gaussian purification $\psi_{\rho}$ of a Gaussian state $\rho$ is of the form
\begin{align}
   \rho^{\otimes n} \cong \bigoplus_{\lambda \in \Lambda}\rho_{\cV_\lambda}\ot \frac{I_{\cW_\lambda}}{\mathrm{dim}\mathcal{W}_{\lambda}} \,,\qquad    \psi_{\rho}^{\text{std},\otimes n} \cong \bigoplus_{\lambda \in \Lambda} (\psi_{\rho})_{\cV_\lambda\tilde{\cV}_\lambda}\ot\frac{\proj{\Gamma}_{\cW_\lambda \tilde \cW_\lambda}}{\dim \cW_{\lambda}}\,,
\end{align}
with the same notation as in~\cref{sec:intro_general_pur} and \cref{sec:intro_KV}.

The next step is to understand that performing the Hayashi measurement on $\psi_{\rho}^{\otimes n}$, obtaining an estimate $\hat{\psi}_{\rho}$ and outputting $\hat{V}$ as the covariance matrix of $\Tr_{\tilde{\mathcal{H}}}[\hat{\psi}_{\rho}]$ gives a density on $\hat{V}$ that does not depend on the specific Gaussian purification $\psi_{\rho}^{\otimes n}$. In fact, by decomposing the integral on the symplectic group as
\begin{align*}\int_{\mathrm{Sp}(4m)} f(S) \dd S = \int_{\mathrm{Sp}(4m)/\mathrm{Sp}(2m)} \left( \int_{\mathrm{Sp}(2m)} f(ST) \dd T \right) \dd \mu(SH),
\end{align*}
where the subgroup $\mathrm{Sp}(2m)$ is the one acting on the auxiliary modes, we obtain that the measure of the random variable $\hat{V}$ has the form
\begin{align*}
\dd P(\hat V \vert \psi_{\rho}^{\otimes n})=\Tr[\dd E({\hat{V})}[(\mathcal{T}_{\mathcal{H}}\otimes \mathrm{id}_{\tilde{\mathcal{H}}})(\psi_{\rho}^{\otimes n})]],
\end{align*}
where $\dd E({\hat{V}})$ is an unnormalized POVM and the twirling map
 \begin{align*}
    \mathcal{T}[X] = \int_{\Sp(2m)} U_S^{\ot n} X U_S^{\dagger, \ot n} \dd S
    \end{align*}
is not a channel, but it is a valid CP map. More in general, by applying the procedure to a general input state $\Omega$, we obtain
\begin{align*}
\dd P(\hat V\vert\Omega)=\Tr[\dd E(\hat{V})[(\mathcal{T}_{\mathcal{H}}\otimes \mathrm{id}_{\tilde{\mathcal{H}}})(\Omega)]].
\end{align*}
Via representation theoretic analysis, it can be shown that

\begin{align*}
   (\mathcal{T}_{\mathcal{H}}\otimes \mathrm{id}_{\tilde{\mathcal{H}}})\left[\bigoplus_{\lambda \in \Lambda}\psi_{\cV_\lambda\tilde{\cV}_\lambda}\otimes \frac{\proj{\Gamma}_{\cW_\lambda \tilde \cW_\lambda}}{\dim \cW_{\lambda}}\right] = \bigoplus_{\lambda \in \Lambda}c_{\lambda}\Tr_{\tilde{\cV}_{\lambda}}[\psi_{\cV_\lambda\tilde{\cV}_\lambda}]\ot I_{\tilde{\cV}_\lambda}\otimes \frac{\proj{\Gamma}_{\cW_\lambda \tilde \cW_\lambda}}{\dim \cW_{\lambda}},
    \end{align*}
for certain positive scalars $c_{\lambda}$ independent of $\rho$. This means that sampling from $\dd P(\hat\sigma \vert \rho)$ can be simulated by replacing access to
$\psi_{\rho}^{\otimes n}$ with access to $ \cQ_{n,m}^{\text{Bose}}[\rho^{\otimes n}]$, where $\cQ_{n,m}^{\text{Bose}}$ is the channel such that
\begin{align*}
  \cQ_{n,m}^{\text{Bose}}[\rho^{\otimes n}] = \bigoplus_\lambda \rho_{\cV_\lambda} \ot \proj{v_\lambda}_{\tilde{\cV}_{\lambda}} \ot \frac{\proj{\Gamma}_{\cW_\lambda \tilde \cW_\lambda}}{\dim \cW_{\lambda}}\,,
\end{align*}
where $\proj{v_\lambda}$ is any pure state in $\tilde{\cV}_{\lambda}$.
Indeed we have $(\mathcal{T}_{\mathcal{H}}\otimes \mathrm{id}_{\tilde{\mathcal{H}}})[\cQ_{n,m}^{\text{Bose}}[\rho^{\otimes n}]]=(\mathcal{T}_{\mathcal{H}}\otimes \mathrm{id}_{\tilde{\mathcal{H}}})[\psi_{\rho}^{\otimes n}]$. It follows that to analyze the performance of the protocol, it suffices to analyze the concentration of $\dd P(\hat V \vert \psi_{\rho}^{\otimes n})$. For zero-mean states, a learning guarantee follows directly from data-processing. For non-zero mean states, we need to extract a relative error guarantee on the covariance matrix, and use the estimate of $\hat{V}$ to act on the remaining copies $\rho_{\sqrt{2}\mu,V}^{\otimes n}$ and analyze the result using the concentration properties of generaldyne measurements.

\begin{algorithm}
\caption{Learning bosonic Gaussian states}\label{alg:learn bosonic gaussian}
\begin{algorithmic}[1] 
\Require $2n>4m+2$ copies of unknown Gaussian state $\rho$ with first and second moments $\mu,\,V$.
\Ensure Estimate $\hat{\rho}$ of $\rho$. 

\State Prepare the state $\rho_{\sqrt{2}\mu,V}^{\otimes n}\otimes \rho_{0,V}^{\otimes n} $ via $50\%$ beam splitters on $2n$ pairs $\rho^{\otimes 2}$,
\State Apply the bosonic quasi-purification channel to
$\rho_{0,V}^{\otimes n}\rightarrow \cQ_{n,m}^{\text{Bose}}[\rho_{0,V}^{\otimes n}]$
\State Perform $\mathrm{Sp}(2m)$-Hayashi measurement on $\cQ_{n,m}^{\text{Bose}}[\rho_{0,V}^{\otimes n}]$, $\rightarrow$ estimate $\hat{\psi}_{AA'}$
\State Set $\hat V$ equal to the covariance matrix of $\Tr_{A'}[\hat{\psi}_{AA'}]$
\State Measure $\rho_{\sqrt{2}\mu,V}^{\otimes n}$ with generaldyne with seed $\hat{V}$, i.e. sample $\sqrt{2}\hat{\mu}\sim \mathcal{N}\left(\sqrt{2}\mu,\frac{V+\hat{V}}{2n}\right)$.
\State Set $\hat{\rho}$ as the Gaussian state with first moment $\hat{\mu}$ and second moment $\hat{V}$.
\State \Return $\hat{\rho}$
\end{algorithmic}
\end{algorithm}

\begin{figure}
\centering
\def\svgwidth{0.85\linewidth}
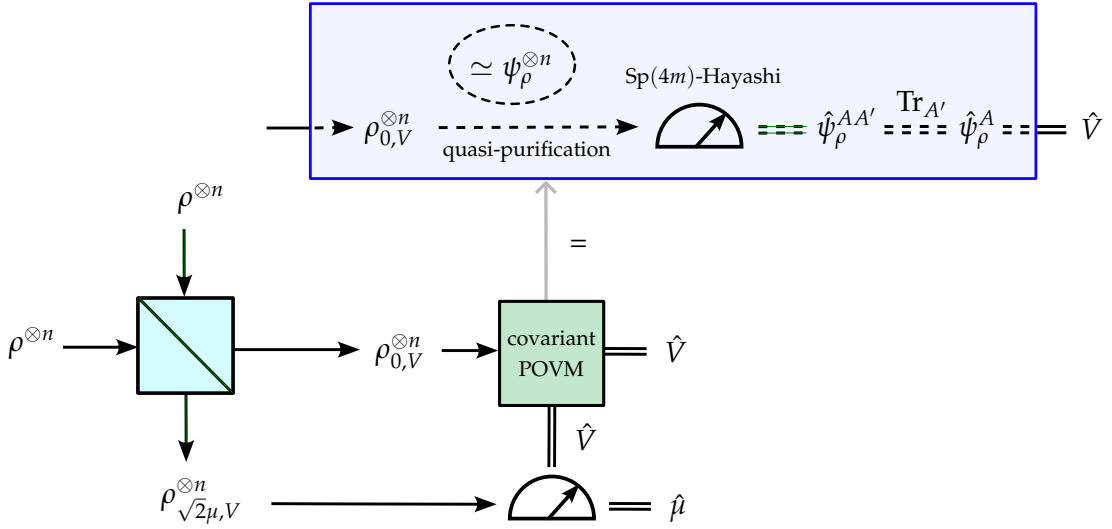
\caption{Scheme of the learning protocol for bosonic Gaussian states. After obtaining zero-mean copies through a beam-splitter, the operations in the green box generate a distribution over covariance-matrix estimates, equivalent to a) measuring a purification with Hayashi measurement b) outputting the covariance matrix of the partial trace of the Hayashi measurement outcome. The estimated covariance matrix can then be used to efficiently estimate the first moment on the remaining copies, via a generaldyne measurement.}
\end{figure}

The above considerations are summarized in the following result.

\begin{thm}[Bosonic Gaussian tomography, upper bound]\label{thm:tomo boson}
There exists a tomography protocol for $m$-mode bosonic Gaussian states~$\sigma$ that outputs an estimate~$\hat\sigma$ with purified distance at most~$P(\sigma,\hat \sigma) \leq \eps$ with probability at least $1-\delta$, using $\bigO((m^2 + \log(\delta^{-1}))/\eps^2)$ copies of~$\sigma$.
\end{thm}

Finally, we also show a matching lower bound in sample complexity that already holds for pure states. It makes use of the same kind of averaging argument as in \cref{thm:tomo fermion lower bound}, except that we average over an appropriate weighted probability distribution instead of a uniform distribution.

\begin{thm}[Bosonic Gaussian tomography, lower bound]\label{thm:tomo boson lower bound}
Suppose there exists a tomography protocol for $m$-mode pure bosonic Gaussian states~$\psi$ that outputs an estimate~$\hat\psi$ with purified distance at most~$P(\sigma,\hat \sigma) \leq \eps$ with probability at least $1 - \delta$ for any pure Gaussian state. Then this protocol needs to use at least $\Omega((m^2 + \log(\delta^{-1}))/\eps^2)$ copies of~$\psi$.
\end{thm}

Together, \cref{thm:tomo fermion}, \cref{thm:tomo fermion lower bound}, \cref{thm:tomo boson}, and \cref{thm:tomo boson lower bound} constitute our main result, which fully determines the sample complexity of learning Gaussian quantum states.

\subsection{Random purification of bosonic gauge-invariant Gaussian states}
As discussed above, the random purification channel for Gaussian bosonic states is not well-defined since the group of Gaussian unitaries is non-compact.
However, we can also consider the maximally compact subgroup~$\Sp(2m) \cap \O(2m) \cong \U(m)$, corresponding to the so-called \emph{passive} Gaussian unitaries, which has an invariant probability measure.
The corresponding quantum states are known as the \emph{gauge-invariant} (or number-preserving) Gaussian states.
Similarly as before, if $\sigma$ is a gauge-invariant Gaussian state, then $\sigma^{\ot n}$ is in the weak operator closure of the span of the $n$-th tensor powers of passive Gaussian unitaries.
This yields the~following~result, proven in \cref{sec:passive bosons}.

\begin{cor}[Bosonic gauge-invariant Gaussian random purification]\label{cor:bosons}
For any $m,n$, there is a channel $\cP_{n,m}^\text{Bose}$ from $\L^2(\R^m)^{\ot n}$ to $\L^2(\R^{2m})^{\ot n}$ such that the following holds:
for any gauge-invariant bosonic Gaussian state~$\sigma$ on $\L^2(\R^m)$,
\begin{align*}
    \cP_{n,m}^\text{Bose}[\sigma^{\ot n}] =
\Ex_{\substack{\psi_\sigma \text{ Gaussian purification of } }\sigma} \psi_\sigma^{\ot n} =
\int_{\U(m)} (I \ot U_O^{\ot n}) \psi^{\text{std},\ot n}_\sigma (I \ot U_O^{\dagger,\ot n}) \ \dd O,
\end{align*}
where the expectation is over Gaussian purifications obtained by applying a uniformly random passive Gaussian unitary~$U_O$ to the standard purification, as in the right-hand side formula.
\end{cor}

\section{Discussion and open questions}
\label{sec:discussion_and_open_questions}

In this work, we solved the open problem of finding the optimal sample complexity of learning fermionic and bosonic Gaussian states on $m$ modes. Specifically, we proved that $\Theta(m^2)$ copies are necessary and sufficient for this task. We did this by introducing two main tools of independent interest: a generalization of the random purification channel~\cite{tang2025conjugate,pelecanos2025,girardi2025random} and a generalization of the Hayashi measurement~\cite{hayashi1998asymptotic,hayashi2017group,harrow2013church}.

The previous state-of-the-art algorithm~\cite{bittel2025energy} for learning $m$-mode bosonic Gaussian states has sample complexity that scales as $O(m^3+(m+\log\log\log E)\log\log E)$, where $E$ is the energy of the unknown Gaussian state. Our algorithm not only improves the dependence on the number of modes, reducing it from $m^3$ to $m^2$, but also completely removes the mild energy dependence $\log\log E$. Consequently, even if an unknown Gaussian state has arbitrarily large energy, a finite number of copies is sufficient to learn it. Regarding the dependence on the number of modes, we can ultimately attribute the improvement of our algorithm to the fact that it employs non-Gaussian operations: random purification and the bosonic version of the Hayashi measurement are not Gaussian operations. In fact, Ref.~\cite{chen2026towards} proved that $\Omega(m^3)$ copies are necessary to perform tomography of Gaussian states via Gaussian operations, matching the performance of the algorithm introduced in~\cite{bittel2025energy} up to the mild $\log\log E$ dependence, which indeed employed only Gaussian operations. Thus, our result shows that non-Gaussian operations provably outperform Gaussian operations for learning Gaussian states. Regarding the energy dependence, it remains open whether tomography algorithms based on Gaussian operations can achieve a fully energy-independent sample complexity, thus removing the mild $\log\log E$ dependence of~\cite{bittel2025energy}. Finally, a second difference between our optimal algorithm and that of~\cite{bittel2025energy} is that the former requires fully entangled measurements across the copies, whereas the latter employs only single-copy adaptive operations. We leave as an open question whether entangled measurements are fundamentally required to perform optimal tomography of bosonic and fermionic Gaussian states.

Another important open question is how to explicitly and efficiently implement our tomography algorithm via a quantum circuit with commonly used Gaussian and non-Gaussian operations. The optimal tomography scheme for qudits in~\cite{pelecanos2025}, based on the random purification channel and the Hayashi measurement, can be implemented explicitly: the random purification channel can be constructed by implementing the Schur transform, while the Hayashi measurement can be replaced by an explicit single-copy tomography algorithm. In our setting, it is not known whether the decompositions in Howe duality and KV duality can be implemented efficiently by a quantum circuit. In the case of KV duality, it is not even completely clear what the right computational model is, given the infinite-dimensional nature of the involved Hilbert spaces. Regarding the Gaussian variant of the Hayashi measurement, it is not yet clear how it can be replaced by an efficiently implementable single-copy protocol.

 Moreover, the dualities used in this paper open the way towards an in-depth study of entangled strategies for estimating parameters of Gaussian states, such as the eigenvalues or the symplectic eigenvalues of the covariance matrix, and the entropy of the state, which is a function of these eigenvalues. Lower bounds for estimating these quantities are also missing.

Another natural open question is the optimal query complexity of learning bosonic and fermionic Gaussian channels. At present, partial results are known only for Gaussian unitaries~\cite{oszmaniec2022fermionsampling,christensen2026learning,fanizza2026efficientlearningbosonicgaussian}. A promising route is to develop a Gaussian analogue of the random Stinespring superchannel constructions~\cite{chen2025quantum,girardi2026randomstinespringsuperchannelconverting,yoshida2026randomdilationsuperchannel}. We expect that the tools introduced here for Gaussian random purification can be adapted to this setting.

\section*{Acknowledgments}
We would like to thank Ben Lovitz, Sepehr Nezami, Mich\`ele Vergne, and John Wright for insightful discussions on highest weight orbit learning, geometric quantization, random purifications, and tomography.
We are grateful to Lennart Bittel, Yanbei Chen, Jens Eisert, Hsin-Yuan Huang, Lorenzo Leone, Alfred Li, Zachary Mann, Antonio Anna Mele, and John Preskill for inspiring discussions.
SC acknowledges funding provided by the Institute for Quantum Information and Matter, an NSF Physics Frontiers Center (NSF Grant PHY-2317110).
FG and LL acknowledge financial support from the European Union (ERC StG ETQO, Grant Agreement no.\ 101165230).
FW acknowledges support by the European Union (ERC Grant ASC-Q, Grant Agreement no.\ 101040624).
MW acknowledges support by the European Union (ERC Grant SYMOPTIC, Grant Agreement no.\ 101040907), by the Deutsche Forschungsgemeinschaft (DFG, German Research Foundation) under Germany's Excellence Strategy~--~EXC-2111~--~390814868 and grant~556164098 as well as the German Federal Ministry of Research, Technology and Space (QuSol, 13N17173).
MF thanks Giacomo De Palma for his kind hospitality at the University of Bologna, where part of this work was done.
MW thanks Q-FARM and the Leinweber Institute for Theoretical Physics at Stanford University and the Simons Institute for the Theory of Computing at UC Berkeley for their hospitality.

\bibliographystyle{alpha}
\bibliography{library}

\begin{appendix}
\clearpage

\section{Review of sample-optimal tomography via random purification}\label{sec:review general tomo}

We give a concise review of an analysis of the optimal sample complexity for tomography of arbitrary $d$-dimensional quantum states.
This section is not necessary to understand our results and proofs, but it outlines in a simpler setting the line of thought that we generalize to prove our main results on sample-optimal tomography of fermionic and bosonic Gaussian states.
There are two key ingredients:
\begin{enumerate}
\item the \emph{symmetric subspace}, which gives an elegant way to analyze the sample complexity of pure state tomography,
\item the \emph{random purification channel}, which lifts the results for pure states to mixed states.
\end{enumerate}
These two ingredients align with two important general principles in quantum information theory, often referred to colloquially as the church of the symmetric subspace and the church of the larger Hilbert space, respectively.

\subsection{The symmetric subspace and pure state tomography}
We consider $n$ copies of a Hilbert space $\cH = \C^d$.
Then the \emph{symmetric subspace} is defined as
\begin{align*}
    \Sym^n(\C^d) \coloneqq \{ \ket{v} \in (\C^d)^{\ot n} \, : \, \sigma \ket{v} = \ket{v} \text{ for all } \sigma \in S_n\}
\end{align*}
where the symmetric group $S_n$ acts by permuting the tensor factors.
Equivalently, this is the subspace spanned by the tensor powers of arbitrary pure states:
\begin{align*}
    \Sym^n(\C^d) = \text{span} \{ \ket{\psi}^{\ot n} \, : \, \ket{\psi} \in \C^d\}.
\end{align*}
The dimension of the symmetric subspace can be computed as
\begin{align}\label{eq:dim symmetric subspace}
    D_{n,d} \coloneqq \dim \Sym^n(\C^d) = \binom{n+d-1}{n} = \binom{n+d-1}{d-1}.
\end{align}
It is an irreducible representation of $\U(d)$ under the action of $U^{\ot n}$ for $U \in \U(d)$.
By Schur's lemma, this implies the following crucial fact:
\begin{align}\label{eq:symmetric subspace integral}
    \Pi_{n,d} = D_{n,d} \int \proj{\psi}^{\ot n} \dd \psi \, ,
\end{align}
where $\Pi_{n,d}$ denotes the orthogonal projection onto~$\Sym^n(\C^d)$, and the~$\dd\psi$ is the unique unitarily-invariant probability measure on the space of $d$-dimensional pure quantum states, often called the ``uniform'' or the ``Haar'' probability measure.

The integral formula in \cref{eq:symmetric subspace integral} implies that $\dd\mu(\hat\psi) = D_{n,d} \proj{\hat\psi}^{\ot n} d\hat\psi$ defines a POVM on the symmetric subspace with outcomes in the space of pure states (by adding $I - \Pi_{n,d}$, we can extend this to a POVM on all of $(\C^d)^{\ot n}$).
That is, the POVM has density $D_{n,d} \proj{\hat\psi}^{\ot n}$ with respect to the uniform measure~$\dd\hat\psi$.
We refer to this as the \emph{Hayashi measurement}.
If we apply this measurement to a state $\ket{\psi}^{\ot n}$, we interpret the outcome $\hat \psi$ as an estimate of $\ket{\psi}$. This turns out to be a good measurement for learning $\ket{\psi}$.
One can analyze its performance by computing the expected fidelity of the estimate with the target state.
Indeed, using \cref{eq:symmetric subspace integral} twice we find
\begin{align}\label{eq:second moment calculation}
\begin{split}
    \Ex_{\hat\psi} \abs{\braket{\psi | \hat \psi}}^2 &= \int \tr[\mu(\hat\psi) \proj{\psi}^{\ot n}] \abs{\braket{\psi | \hat \psi}}^2 \dd \hat\psi\\
    &= D_{n,d} \int \tr[\proj{\hat \psi}^{\ot n}  \proj{\psi}^{\ot n}] \tr[ \proj{\hat \psi} \proj{\psi}] \dd \hat\psi\\
    &= D_{n,d} \int \tr[\proj{\hat \psi}^{\ot (n+1)}  \proj{\psi}^{\ot (n+1)}] \dd \hat\psi \\
    &= \frac{D_{n,d}}{D_{n+1,d}} \tr[\Pi_{n+1} \proj{\psi}^{\ot(n+1)}] = \frac{D_{n,d}}{D_{n+1,d}} = \binom{n+d-1}{d-1} / \binom{n+d}{d-1}.
\end{split}
\end{align}
An easy calculation shows this is lower bounded by $1 - d/n$, so by applying Markov's inequality, we find that $n = \bigO(d/\eps^2)$ copies suffice to get an estimate for which $\abs{\braket{\psi | \hat \psi}}^2 \geq 1 - \eps^2$ with constant probability of error~$\delta$.

If one wants to get the dependence on~$\delta$ right, one can compute higher moments.
An argument to this end is given in \cite{pelecanos2025}, which is based on explicitly determining the distribution of the fidelity from its moments.
Here, we give an alternative argument that has the advantage that it can be easily be adapted to the case of Gaussian states.
An essentially identical calculation to \cref{eq:second moment calculation} shows that, for any integer~$k \geq 0$,
\begin{align}\label{eq:k-n}
    \Ex_{\hat\psi} \abs{\braket{\psi | \hat \psi}}^{2(k-n)} = \frac{D_{n,d}} {D_{k,d}}.
\end{align}
Taking $n$ even and $k = n/2$, we can bound the probability of failure as
\begin{align*}
    \Pr_{\hat\psi}\mleft(\abs{\braket{\psi | \hat \psi}}^2 \leq 1 - \eps^2 \mright) &= \Pr\mleft(\abs{\braket{\psi | \hat \psi}}^{-n} \geq (1 - \eps^2)^{-n/2}\mright) \\
    &\leq (1 - \eps^2)^{n/2} \Ex \abs{\braket{\psi | \hat \psi}}^{-n} \\
    &= (1 - \eps^2)^{n/2} \frac{D_{n,d}} {D_{n/2,d}} = (1 - \eps^2)^{n/2} \prod_{j=1}^{d-1} \underbrace{\frac{n+j}{n/2 + j}}_{\leq 2}\\
    &\leq (1 - \eps^2)^{n/2} 2^{d-1} \leq \exp\left(- \frac{n\eps^2}{2} + d\right),
\end{align*}
where we used Markov's inequality in the first inequality, next \cref{eq:k-n}, and lastly \cref{eq:dim symmetric subspace}.
We conclude that $n = \bigO\mleft(\frac{d + \log(\delta^{-1})}{\eps^2}\mright)$ copies suffice to obtain an $\eps$-accurate estimate with failure probability at most $\delta$.

The symmetric subspace, and in particular the integral formula \cref{eq:symmetric subspace integral}, can also be used to give a \emph{lower bound} on the sample complexity of pure state tomography.
The idea is as follows.
Suppose we are given an arbitrary POVM with outcomes in the space of pure states, and let~$P_{\psi}$ denote the distribution of measurement outcomes when applied to~$\psi^{\ot n}$.
By a calculation similar to \cref{eq:second moment calculation,eq:k-n} one can show that for a uniformly random pure state~$\psi$,
\begin{align*}
    \Ex_{\psi} \Ex_{\hat\psi \sim P_\psi} \abs{\braket{\psi | \hat \psi}}^{2k} \leq \frac{D_{n,d}} {D_{n+k,d}}.
\end{align*}
Looking at the $k$-th moment for $k = \Theta(\eps^{-2})$ then shows that at least $\Omega(d/\eps^2)$ copies are needed for constant probability of success.

Finally, to also get a matching lower bound in the regime of small $\delta$, we need to show that we need at least $\Omega(\log(\delta^{-1})/\eps^2)$ copies.
This follows from a reduction: tomography to precision~$\eps/2$ is at least as hard as distinguishing two specific quantum states $\ket{\psi}, \ket{\phi}$ with $\abs{\braket{\psi | \phi}}^2 = 1 - \eps^2$, given access to either $\ket{\psi}^{\ot n}$ or $\ket{\phi}^{\ot n}$. Indeed, $T(\psi,\phi) = \eps$, so it suffices to get an estimate that has accuracy $\eps/2$ in trace distance by the triangle inequality.
The optimal probability of failure for this task is given by the trace distance: it equals
\begin{align}\label{eq:distinguising two states}
    \begin{split}
        \delta = \frac12(1 - T(\psi^{\ot n},\phi^{\ot n})) = \frac12\left(1 - \sqrt{1 - (1 - \eps^2)^n})\right)
        \geq (1 - \eps^2)^n / 4 = \Omega(\exp(-\eps^2 n)).
    \end{split}
\end{align}
We conclude that pure state tomography has sample complexity $n = \Theta\mleft(\frac{d + \log(\delta^{-1})}{\eps^2}\mright)$.
This approach to pure state tomography was proposed by Hayashi~\cite{hayashi1998asymptotic}. See \cite{harrow2013church,brandao2016mathematics} for a detailed discussion of the symmetric subspace, tomography of pure states, and further applications. The lower bound argument we use is given by \cite{scharnhorst2025optimal}.

\subsection{Random purification and mixed state tomography}
Next, we consider the task of learning a state~$\rho$ on $\cH = \C^d$ that is not necessarily pure.

For any mixed state~$\rho$, there exists a pure state~$\ket{\psi_\rho}$ on $\cH \ot \tilde\cH$, where $\tilde\cH = \C^d$ is an auxiliary system of the same dimension, such that the reduced state on the first system equals~$\rho$.
This purification is not unique: applying an arbitrary unitary~$U \in \U(d)$ to the auxiliary system yields another purification; in fact, any two purifications of~$\rho$ differ by such a unitary.
This is a particularly useful principle in quantum information theory, where it is often much easier to reason about pure states rather than mixed states.
For example, the derivation of the optimal sample complexity for pure state tomography we just saw is rather elegant and elementary.
Recently, it has been shown that purifications not only exist, but can in fact be constructed in a meaningful way by a \emph{quantum channel} \cite{tang2025conjugate,pelecanos2025,girardi2025random}, see also \cite{soleimanifar2022testing,chen2024local}.
While no quantum channel can assign a fixed purification to every quantum state, there does exist a \emph{random purification channel} that prepares a purification with a random unitary applied to the auxiliary system.
Crucially, it can be consistently extended to multiple copies of the state. Namely, the random purification channel $\cP_{n,d}$ acts on $n$ copies of a quantum system with Hilbert space $\cH = \C^d$, and it has the property that for every state $\rho$ on $\C^d$,
\begin{align}\label{eq:random purification intro}
    \cP_{n,d}[\rho^{\ot n}] = \int_{\U(d)} \left((I \ot U) \psi^\text{std}_\rho (I \ot U^\dagger)\right)^{\ot n} \, \dd U = \Ex_{\psi_\rho \text{ purification of } \rho } \psi_\rho^{\ot n},
\end{align}
where $\psi^\text{std}_\rho$ is a fixed (``standard'') purification of $\rho$.

This result can be applied to tomography~\cite{pelecanos2025}.
Any tomography protocol for pure states can be extended to a tomography protocol for mixed states as follows:
\begin{enumerate}[noitemsep]
    \item Given $n$ copies of a mixed state~$\sigma$, apply the random purification channel.
    \item Apply the pure state tomography protocol to the output of the random purification channel.
    \item Return the reduced density matrix of the estimated pure state.
\end{enumerate}
To learn the mixed state, indeed, it clearly suffices to learn a purification of it, and it does not matter which one.
Formalizing this argument shows that the above protocol is correct.
Note that it is crucial that the random purification is consistent over multiple copies.
The above protocol, and more sophisticated versions of it, has been analyzed in detail in~\cite{pelecanos2025}.
Remarkably, while the analysis is \emph{easier} than for previous protocols, the resulting protocol is \emph{optimal} in terms of sample complexity.
Indeed, the number of copies needed is that for pure-state estimation on the Hilbert space $\cH \ot \cH$, which has dimension~$d^2$, so this shows achievability in \cref{eq:sample complexity tomography}.
One can show that this is indeed optimal \cite{Haah2016, scharnhorst2025optimal}, but the details of that argument are not relevant to the Gaussian problem, where we will find that a lower bound for pure states is already optimal.
Furthermore, the random purification channel can be efficiently implemented, yielding the first tomography scheme that is both computationally efficient and optimal in sample complexity.

The random purification channel in \cref{eq:random purification intro} can be derived from symmetry considerations.
In the discussion of the symmetric subspace, we already mentioned two group actions on~$\cH^{\ot n}$.
One is given by $\U(d)$, with unitaries acting identically on each copy by $U^{\ot n}$.
The other action is given by the permutation group~$S_n$, which permutes the copies.
These two actions are each others commutants, and this gives the Schur--Weyl decomposition, which decomposes the actions of $\U(d)$ and $S_n$ to a block-diagonal form:
\begin{align}\label{eq:sw appendix}
    \cH^{\ot n} \cong \bigoplus_{\lambda} \mathcal V_\lambda \ot \mathcal W_\lambda,
\end{align}
where $\mathcal V_\lambda$ and $\mathcal W_\lambda$ are irreducible representations of $\U(d)$ and $S_n$ respectively, labeled by Young diagrams~$\lambda$.
In particular, the symmetric subspace corresponds to the trivial representation of $S_n$.

The high-level idea of the random purification channel is that $\rho^{\ot n}$ is in the algebra generated by the action of~$\U(d)$, hence
\begin{align}\label{eq:rho n times}
    \rho^{\ot n} \cong \bigoplus_{\lambda} \rho_{\lambda} \ot \frac{\id_{\mathcal W_{\lambda}}}{\dim \mathcal W_{\lambda}},
\end{align}
so the state is maximally mixed on the irreducible representations of $S_n$.
If we identify
\begin{align*}
    (\cH \ot \tilde{\cH})^{\ot n} \cong \bigoplus_{\lambda, \mu} \mathcal{V}_\lambda \ot \tilde{\mathcal{V}}_{\mu} \ot \mathcal W_\lambda \ot \tilde{\mathcal W}_{\mu}
\end{align*}
then one can show that $n$ copies of the standard purification of~$\rho$ take the form
\begin{align*}
    \ket{\psi_\rho^{\text{std}}}^{\ot n} \cong \bigoplus_\lambda \underbrace{\ket{\psi_\rho}_{\mathcal V_\lambda \tilde{\mathcal V}_{\lambda}}}_{\text{purifies } \rho_\lambda} \ot \underbrace{\ket{\Phi}_{\mathcal W_\lambda \tilde{\mathcal W}_{\lambda}}}_{\substack{\text{purifies} \\ \text{max. mixed state}}},
\end{align*}
where each $\ket{\Phi}_{\mathcal W_\lambda \tilde{\mathcal W}_{\lambda}}$ is a maximally entangled state, and the $\ket{\psi_\rho}_{\mathcal V_\lambda \tilde{\mathcal V}_{\lambda}}$ are purifications of the (subnormalized) states $\rho_{\lambda}$.
Note that if we apply $U^{\ot n}$ to the $(\tilde{\cH})^{\ot n}$ this acts as $U_\lambda$ on $\tilde{\mathcal V}_{\lambda}$, and averaging over random unitaries~$U^{\ot n}$ corresponds to a depolarizing channel on the~$\tilde \cV_\lambda$, so we find that
\begin{align}\label{eq:random purification expression review}
    \int_{\U(d)} \left((I \ot U) \psi^\text{std}_\rho (I \ot U^\dagger)\right)^{\ot n} \, \dd U \cong \bigoplus_\lambda \rho_\lambda \ot \frac{I_{\tilde{\mathcal V}_{\lambda}}}{\dim \mathcal V_\lambda} \ot \Phi_{\mathcal W_\lambda \tilde{\mathcal W}_{\lambda}}.
\end{align}
The right-hand side can be obtained by applying a quantum channel to \cref{eq:rho n times}:
for each $\lambda$, discard the state on~$\cW_{\lambda}$ and prepare a maximally mixed state on~$\tilde\cV_\lambda$ as well as a maximally entangled state between~$\cW_{\lambda}$ and~$\tilde\cW_{\lambda}$.
This defines the random purification channel $\cP_{n,d}$.

In fact, one does not need the full strength of Schur--Weyl duality to derive this result and a simpler derivation is given in \cite{girardi2025random}.
The discussion after~\cref{thm:main simplified} explains how to recover the random purification channel without any use of Schur--Weyl duality as a special case of a more general result.
The advantage of using the explicit decomposition in \cref{eq:sw appendix} then is that it is efficiently implemented by the quantum Schur transform---this gives a computationally efficient implementation of the random purification channel, and the same applies in the more general setting discussed below.

\begin{rem}\label{rem:any staet on V_lambda}
Above we saw that the random purification channel prepares maximally mixed states on the registers~$\tilde \cV_\lambda$.
In the generalization to bosonic systems, which have an infinite-dimensional Hilbert space, we will see that the spaces analogous to $\cV_\lambda, \tilde\cV_\lambda$ will become infinite-dimensional (while the analogue to $\cW_\lambda, \tilde\cW_\lambda$ remains finite-dimensional).
A priori, this poses a problem, as there exists no maximally mixed state on an infinite-dimensional Hilbert space.

To that end, we make the observation that in order to reduce mixed-state tomography to pure-state tomography, it is \emph{not} necessary to prepare the maximally mixed states on the registers~$\tilde\cV_\lambda$ (both in the present setting as well as in the Gaussian settings discussed later).
Indeed, suppose we modify the channel~$\mathcal P_{n,d}$ to a channel $\mathcal Q_{n,d}$ such that
\begin{align}\label{eq:modified random puri}
    \mathcal Q_{n,d}[\rho^{\ot n}] \cong \bigoplus_\lambda \rho_\lambda \ot \nu_\lambda \ot \Phi_{\mathcal W_\lambda \tilde{\mathcal W}_{\lambda}},
\end{align}
using arbitrary fixed states~$\nu_\lambda$ on the registers~$\tilde \cV_\lambda$ in place of the maximally mixed states in \cref{eq:random purification expression review}.
We call this channel a \emph{quasi-purification channel}.
This quasi-purification channel no longer has the interpretation of outputting a random purification, it is not hard to see that we obtain the same distribution of measurement outcomes (i.e., estimates of the true state) when using it in the mixed-state tomography protocol.\footnote{In \cite{pelecanos2025}, a different channel is constructed that is also called a quasi-purification channel; while that channel is different and serves a different purpose, it has the same philosophy: it loses the interpretation of preparing copies of a random purification, but composes well with pure state tomography.}
Because the Hayashi measurement is covariant, the POVM defined by applying the Hayashi measurement and then outputting the reduced state is invariant under applying~$U^{\ot n}$ on the purifying system.
Accordingly, if we insert a Haar twirl over~$U^{\ot n}$ right before the measurement this does not change the statistics -- but such a twirl maps \cref{eq:modified random puri} to \cref{eq:random purification expression review}.

Alternatively, one can also use Schur's lemma to see that the invariance of the POVM implies that it is block-diagonal with respect to $\lambda=\mu$, with each block proportional to the identity operator on the~$\tilde\cV_\lambda$ register, hence it is insensitive to the choice of states~$\nu_\lambda$. While there is not a twirling channel for non-compact groups, an infinite-dimensional analogue of Schur's lemma will still imply that the POVM we use, when applied to $\bigoplus_\lambda \rho_{\cW_\lambda,\tilde{\cW}_{\lambda}}\ot \Phi_{\mathcal W_\lambda \tilde{\mathcal W}_{\lambda}}$, is only sensitive to $\Tr_{\tilde{\cW}_{\lambda}}[\rho_{\cW_\lambda,\tilde{\cW}_{\lambda}}]$.
\end{rem}

\section{A random purification channel for general symmetries}\label{sec:random purification general}
In this section we state and prove a random purification theorem for general symmetries in finite dimensions.
We first set up some notation and recall facts about finite-dimensional $*$-algebras in \cref{sec:algebras}.
We then prove the main result of this section in \cref{sec:random purification theorem}.
\Cref{subsec:examples} gives some direct applications, including the original random purification channel from~\cite{tang2025conjugate,pelecanos2025}.

\subsection{Algebras and commutants}\label{sec:algebras}
In this section, we consider finite-dimensional Hilbert spaces.
We call~$\Lin(\cH)$ the set of linear operators acting on a Hilbert space~$\cH$. We write $\PSD(\cH)$ for the positive semidefinite operators on~$\cH$, and~$\S(\cH)$ for the set of density matrices on~$\cH$.

Given two Hilbert spaces~$\cH$ and~$\cK$, a \emph{quantum channel} from~$\cH$ to~$\cK$ is defined as a completely positive trace-preserving map~$\Phi\colon\Lin(\cH)\to\Lin(\cK)$.
We recall that any closed subgroup~$G \subseteq \U(\cH)$ has a Haar probability measure, which we are going to denote by~$\dd g$.

Let us consider a $*$-algebra of linear operators~$\cA \subseteq \Lin(\cH)$. $\cA$ can be written as a direct sum of matrix algebras. More concretely, there exists a unitary~$\Upsilon$ implementing a decomposition
\begin{align}\label{eq:decomposition}
    \cH \cong \bigoplus_{\lambda \in \Lambda} \cL_\lambda \ot \cR_\lambda,
\end{align}
where~$\Lambda$ is a finite index set and $\cL_\lambda$ and $\cR_\lambda$ are Hilbert spaces such that
\begin{align}\label{eq:decomposition algebras}
    \cA \cong \bigoplus_{\lambda \in \Lambda} I_{\cL_\lambda} \ot \Lin(\cR_\lambda),
    \quad\text{and hence}\quad
    \cA' \cong \bigoplus_{\lambda \in \Lambda} \Lin(\cL_\lambda) \ot I_{\cR_\lambda},
\end{align}
where $S' \subseteq \Lin(\cH)$ is the commutant of a subset~$S = S^\dagger \subseteq \Lin(\cH)$, i.e., the $*$-algebra of operators that commute with all the operators in~$S$.

Let~$\E_\cA \colon \Lin(\cH) \to \cA$ be the orthogonal projection onto~$\cA$ with respect to the Hilbert-Schmidt inner product. In terms of~\eqref{eq:decomposition}, $\E_\cA$ is~given~by
\begin{align}\label{eq:E}
    \E_\cA[M] \cong \bigoplus_{\lambda \in \Lambda}  \frac{I_{\cL_\lambda}}{\dim \cL_\lambda} \ot \tr_{\cL_\lambda}[ P_\lambda M P_\lambda ]
    \qquad (M \in \Lin(\cH)),
\end{align}
where $P_\lambda$ denotes the orthogonal projections onto the summands in \cref{eq:decomposition}.
That is, $\E_\cA$ acts as a measurement of~$\lambda$ composed with completely depolarizing channels on each~$\cL_\lambda$.
Alternatively, the orthogonal projection can be written as follows:
Take any closed subgroup~$G \subseteq \U(\cH)$ that generates (equivalently, spans) the commutant algebra~$\cA'$.%
\footnote{Equivalently, by the double commutant theorem, the closed subgroup should satisfy~$G' = \cA$.}
Then the orthogonal projection is given by the twirl with respect to the Haar probability measure~$\dd g$ on~$G$:
\begin{align}\label{eq:twirl}
    \E_\cA[M] = \int_G g M g^\dagger \,\dd g.
\end{align}
A choice that always works is to take~$G = \U(\cA')$, the group of unitaries in the $*$-algebra, but often~$G$ can be taken to be much smaller.%
\footnote{For example, the Pauli group spans the $*$-algebra of $2\times 2$-matrices. More generally, any ensemble of unitaries satisfying~\eqref{eq:twirl} is called a \emph{unitary 1-design} for~$\cA'$.}
A well-known and motivating example is the following: if~$\cH = \cK^{\ot t}$ and $\cA \subseteq \Lin(\cK^{\ot t})$ is spanned by the permutation operators, then $\cA'$ is spanned by the group~$G = \{ U^{\ot t} : U \in \U(\cK) \}$.
This result is known as Schur--Weyl duality.

\subsection{Random purification theorem}\label{sec:random purification theorem}

While purifications can be canonically defined by choosing the purifying Hilbert space~$\tilde\cH$ as the dual of the original Hilbert space~$\cH$, in applications we often need to choose~$\tilde\cH$ and the purification in a specific way (so that, e.g., it can be interpreted as a Gaussian state).
Accordingly we will work with the following concrete setup.
Let $\cH, \tilde{\cH}$ be Hilbert spaces of the same dimension, with fixed orthonormal bases~$\{\ket x\}_{x \in I}$, $\{\ket{\tilde{x}}\}_{x \in I}$, for some index set~$I$.%
\footnote{This is the same data as a basis of $\cH$ and a unitary $\cH \to \tilde\cH$.}
The transpose and conjugation are defined with respect to these bases, which also allow us to associate to any operator~$c \in \Lin(\cH)$ an operator~$\tilde c \in \Lin(\tilde{\cH})$, so $\braket{\tilde{x} | \tilde c | \tilde{y}} = \braket{x | c | y}$.
We write~$c$ instead of~$\tilde c$ when it is clear on which Hilbert space the operator acts (as will often be the case).
Now let~$\ket\Gamma = \sum_x \ket{x\tilde{x}} \in \cH \ot \tilde{\cH}$ denote the unnormalized maximally entangled state corresponding to these bases.
It satisfies the following version of the transpose trick:
$(c \ot I) \ket\Gamma = (I \ot c\tran) \ket\Gamma$.
In particular, $(u \ot \bar u) \ket\Gamma = \ket\Gamma$ for any unitary~$u \in \U(\cH)$.
We define the \emph{standard purification} of a quantum state~$\rho \in \S(\cH)$ by the formula $\ket{\psi^\text{std}_\rho} \coloneqq (\sqrt\rho \ot I) \ket\Gamma$.
Of course, it is only standard once the two bases have been chosen.

We now discuss the interplay of algebras, transposes, and basis choices.
Let $\cA \subseteq \Lin(\cH)$ be $*$-algebra of linear operators, as in \cref{sec:algebras}.
Let $\Upsilon \colon \mathcal H \to \bigoplus_{\lambda \in \Lambda} \cL_\lambda \ot \cR_\lambda$ be a unitary implementing the decomposition in \cref{eq:decomposition,eq:decomposition algebras}, and fix a product basis $\{\ket{\lambda,lr}\}_{l,r}$ for each direct summand~$\cL_\lambda \ot \cR_\lambda$.
Then, if we expand an arbitrary operator~$c \in \Lin(\cH)$ as
\begin{align}\label{eq:transpose}
    \Upsilon c \Upsilon^\dagger = \bigoplus_{\lambda\in\Lambda} b_\lambda \ot a_\lambda,
\quad\Rightarrow\quad
    \bar\Upsilon c\tran (\bar\Upsilon)^\dagger
= (\Upsilon c \Upsilon^\dagger )\tran
= \bigoplus_{\lambda\in\Lambda} b_\lambda\tran \ot a_\lambda\tran,
\end{align}
where $\bar\Upsilon$ denotes the conjugate with respect to the chosen bases.
Thus, $\bar\Upsilon$ implements decompositions as in \cref{eq:decomposition,eq:decomposition algebras} for the transposed algebra~$\cA\tran$ and its commutant. 
The right-hand side formula holds verbatim if we consider $c\tran$ as an operator in~$\Lin(\tilde\cH)$ and also identify~$\bar\Upsilon$ with a unitary $\tilde\cH \to \bigoplus_{\lambda \in \Lambda} \tilde\cL_\lambda \ot \tilde\cR_\lambda$, where $\tilde{\cL}_\lambda = \cL_\lambda$ and $\tilde{\cR}_\lambda = \cR_\lambda$ (we use tildes to distinguish the purifying space from the original one), by using the chosen bases.
If we use~$\Upsilon$ and~$\bar\Upsilon$ to decompose $\cH$ and~$\tilde\cH$, respectively, the joint Hilbert spaces decompose as
\begin{align*}
    \Upsilon \ot \bar\Upsilon
\colon
    \cH \ot \tilde\cH
\rightarrow \bigoplus_{\lambda\in\Lambda} ( \cL_\lambda \ot \cR_\lambda ) \ot \bigoplus_{\mu\in\Lambda} ( \tilde\cL_\mu \ot \tilde\cR_\mu ) \cong \bigoplus_{\lambda,\mu\in\Lambda} ( \cL_\lambda \ot \tilde\cL_\mu ) \ot ( \cR_\lambda \ot \tilde\cR_\mu ).
\end{align*}
Furthermore, the maximally entangled state~$\ket\Gamma$ identifies with a direct sum of maximally entangled states of the corresponding tensor factors:
\begin{align}\label{eq:max ent decomposition}
    (\Upsilon \ot \bar\Upsilon) \ket{\Gamma}
= \sum_x \Upsilon \ket x \otimes \overline{\Upsilon \ket x}
= \sum_{\lambda, l, r} \ket{\lambda,lr}\otimes \ket{\lambda,lr}
\cong \bigoplus_\lambda \ket{\Gamma}_{\cL_\lambda\tilde{\cL}_\lambda} \ot \ket{\Gamma}_{\cR_\lambda\tilde{\cR}_\lambda},
\end{align}
where~$\ket{\Gamma}_{\cL_\lambda\tilde{\cL}_\lambda}$, $\ket{\Gamma}_{\cR_\lambda\tilde{\cR}_\lambda}$ denote unnormalized maximally entangled states in~$\cL_\lambda \ot \tilde{\cL}_\lambda$ and~$\cR_\lambda \ot \tilde{\cR}_\lambda$ (with respect to the local bases chosen to define the product basis $\ket{\lambda,lr}$), with squared norm equal to $\dim(\cL_\lambda)$ and $\dim(\cR_\lambda)$, respectively.
The second equality in \cref{eq:max ent decomposition} holds due to the well-known $u \otimes \bar u$ invariance of maximally entangled states with respect to the chosen bases (here, the $\ket{\lambda,lr}$ basis).

We are now ready to present our general random purification theorem.
In the statement of the theorem, we use the fixed bases chosen above to identify operators on~$\cH$ and~$\tilde\cH$ and to take conjugates and transposes of such operators, as well as to define the standard purification.
\Cref{thm:main simplified} in the introduction is an immediate corollary (choose the algebra~$\cA$ as the commutant of the group~$G$).

\begin{thm}[Random purification for general symmetries]\label{thm:main technical}
For any $*$-algebra $\cA \subseteq \Lin(\cH)$, there is a quantum channel
$\cP_\cA \colon \Lin(\cH) \to \Lin(\cH \ot \tilde{\cH})$
with the following properties:
\begin{enumerate}
\item \emph{Random purification of invariant input states:}
For any input state~$\rho \in \S(\cH)$ that commutes with~$\cA$, i.e., $\rho \in \cA'$, we have
\begin{align}\label{eq:sym action}
    \cP_\cA[\rho]
= \parens*{ \idCh \ot \E_{\cA\tran} }[ \psi^\text{std}_\rho ]
= \int_G (I \ot g\tran) \psi^\text{std}_\rho (I \ot \bar g) \,\dd g,
\end{align}
where $G \subseteq \U(\cH)$ is any closed subgroup with~$\C G = \cA'$.

\item \emph{Symmetry of output states:}
For any input state~$\rho \in \S(\cH)$, the output state $\cP_\cA[\rho]$ is supported on the $\cA$-symmetric subspace of~$\cH \ot \tilde{\cH}$ (defined as the subspace of vectors invariant by~$u \ot \bar u$ for all unitaries~$u$ in the algebra~$\cA$),
and it is also $G$-invariant on the purifying system (i.e., commutes with~$I \ot g\tran$ for all $g \in G$).

\item \emph{Efficiency:}
If unitaries implementing the decomposition~\eqref{eq:decomposition}, the basis change~$\ket x \mapsto \ket{\tilde{x}}$ can be implemented efficiently, and one can efficiently prepare maximally mixed states on $\cL_\lambda$ and maximally entangled states on two copies of $\cR_\lambda$, then $\cP_\cA$ can be implemented efficiently.
\item \emph{Recovery map:}
The channel~$\cP_\cA$ is the Petz recovery map for the partial trace over the purifying system and the maximally mixed state on the~$\cA$-symmetric subspace.

\item \emph{Explicit formulas:}
For all $\rho \in \S(\cH)$,
\begin{align}\label{eq:explicit}
    \cP_\cA[\rho]
&\cong \bigoplus_{\lambda\in\Lambda} \tr_{\cR_\lambda}[P_\lambda \rho P_\lambda] \ot \frac{I_{\tilde{\cL}_\lambda}}{\dim \cL_\lambda} \ot \tfrac1{\dim\cR_\lambda} \ket{\Gamma}_{\cR_\lambda\tilde{\cR}_\lambda}\bra{\Gamma}_{\cR_\lambda\tilde{\cR}_\lambda}\\
&= \Pi_\cA \parens*{ \rho \ot \Omega } \Pi_\cA
= \sqrt{\cP_\cA[I]} \parens*{ \rho \ot I } \sqrt{\cP_\cA[I]},
\end{align}
where $\Pi_\cA$ denotes the projection onto the~$\cA$-symmetric subspace, and~$\Omega \coloneqq \bigoplus_{\lambda\in\Lambda} \frac {\dim\cR_\lambda}{\dim\cL_\lambda} P_\lambda\tran$, with~$P_\lambda$ the orthogonal projections onto the direct summands in~\eqref{eq:decomposition}.
\end{enumerate}
\end{thm}
\begin{proof}
By the discussion above, we may assume that the decompositions~\eqref{eq:decomposition} and~\eqref{eq:decomposition algebras} hold with equality:
\begin{align*}
    \cH = \bigoplus_{\lambda \in \Lambda} \cL_\lambda \ot \cR_\lambda, \quad
    \cA = \bigoplus_{\lambda \in \Lambda} I_{\cL_\lambda} \ot \Lin(\cR_\lambda),
    \quad
    \cA' = \bigoplus_{\lambda \in \Lambda} \Lin(\cL_\lambda) \ot I_{\cR_\lambda}.
\end{align*}
and similarly, using \cref{eq:transpose},
\begin{align*}
    \tilde{\cH} = \bigoplus_{\lambda \in \Lambda} \tilde{\cL}_\lambda \ot \tilde{\cR}_\lambda, \quad
    \cA\tran = \bigoplus_{\lambda \in \Lambda} I_{\tilde{\cL}_\lambda} \ot \Lin(\tilde{\cR}_\lambda),
    \quad
    (\cA')\tran = \bigoplus_{\lambda \in \Lambda} \Lin(\tilde{\cL}_\lambda) \ot I_{\tilde{\cR}_\lambda},
\end{align*}
where $\tilde{\cL}_\lambda = \cL_\lambda$ and $\tilde{\cR}_\lambda = \cR_\lambda$, and we think of the transposed algebra~$\cA\tran$ as operators on $\tilde{\cH}$ by using the chosen bases.
In particular, if $P_\lambda$ denotes the projections onto the direct summands of~$\cH$, then~$P_\lambda\tran$ are the projections onto the direct summands of~$\tilde{\cH}$.
The joint Hilbert space decomposes as
\begin{align*}
    \cH \ot \tilde\cH
= \bigoplus_{\lambda\in\Lambda} ( \cL_\lambda \ot \cR_\lambda ) \ot \bigoplus_{\mu\in\Lambda} ( \tilde\cL_\mu \ot \tilde\cR_\mu )
= \bigoplus_{\lambda,\mu\in\Lambda} \cL_\lambda \ot \tilde\cL_\mu \ot \cR_\lambda \ot \tilde\cR_\mu,
\end{align*}
where we write the right-hand side identification as an equality since it is canonical.
Then \cref{eq:max ent decomposition} also holds with equality:
\begin{align}\label{eq:EPR}
    \ket\Gamma = \bigoplus_\lambda \ket{\Gamma}_{\cL_\lambda\tilde{\cL}_\lambda} \ot \ket{\Gamma}_{\cR_\lambda\tilde{\cR}_\lambda}
\in \cH \ot \tilde{\cH}.
\end{align}
\begin{enumerate}
\item We will compute $\parens*{ \idCh \ot \E_{\cA\tran} }[ \psi^\text{std}_\rho ]$ and see that it can be realized by a natural quantum channel acting on~$\rho$.
We start by the computation of the standard purification with respect to the above decomposition.
Since $\rho \in \cA'$, we see that
\begin{align}\label{eq:rho}
    \rho = \bigoplus_{\lambda\in\Lambda} \rho_\lambda \ot I_{\cR_\lambda},
\quad\text{where}\quad
    \rho_\lambda = \frac1{\dim\cR_\lambda} \tr_{\cR_\lambda}[P_\lambda \rho P_\lambda] \in \PSD(\cL_\lambda).
\end{align}
Using \cref{eq:EPR}, we see that the standard purification is given by
\begin{align*}
    \ket{\psi^\text{std}_\rho} = (\sqrt\rho \ot I) \ket\Gamma = \bigoplus_{\lambda\in\Lambda} \ket{\psi_\lambda}_{\cL_\lambda\tilde{\cL}_\lambda} \ot \ket{\Gamma}_{\cR_\lambda\tilde{\cR}_\lambda},
\end{align*}
where each $\ket{\psi_\lambda}$ is the standard purification of~$\rho_\lambda$ (with respect to the bases chosen above).
As a density operator,
\begin{align*}
    \proj{\psi^\text{std}_\rho} = \bigoplus_{\lambda,\mu\in\Lambda} \ket{\psi_\lambda}_{\cL_\lambda\tilde{\cL}_\lambda}\bra{\psi_\mu}_{\cL_\mu\tilde{\cL}_\mu} \ot \ket{\Gamma}_{\cR_\lambda\tilde{\cR}_\lambda}\bra{\Gamma}_{\cR_\mu\tilde{\cR}_\mu},
\end{align*}
From \cref{eq:E}, we see that
\begin{align*}
    \parens*{ \idCh \ot \E_{\cA\tran} }\mleft[ \proj{\psi^\text{std}_\rho} \mright]
&= \bigoplus_{\lambda\in\Lambda}  \tr_{\tilde{\cL}_\lambda}[\ket{\psi_\lambda}_{\cL_\lambda\tilde{\cL}_\lambda}\bra{\psi_\lambda}_{\cL_\lambda\tilde{\cL}_\lambda}] \ot \frac{I_{\tilde{\cL}_\lambda}}{\dim \cL_\lambda} \ot \ket{\Gamma}_{\cR_\lambda\tilde{\cR}_\lambda}\bra{\Gamma}_{\cR_\lambda\tilde{\cR}_\lambda}\\
&= \bigoplus_{\lambda\in\Lambda} \rho_{\lambda} \ot \frac{I_{\tilde{\cL}_\lambda}}{\dim \cL_\lambda} \ot \ket{\Gamma}_{\cR_\lambda\tilde{\cR}_\lambda}\bra{\Gamma}_{\cR_\lambda\tilde{\cR}_\lambda}\\
&= \bigoplus_{\lambda\in\Lambda} \tr_{\cR_\lambda}[P_\lambda \rho P_\lambda] \ot \frac{I_{\tilde{\cL}_\lambda}}{\dim \cL_\lambda} \ot \frac{\ket{\Gamma}_{\cR_\lambda\tilde{\cR}_\lambda}\bra{\Gamma}_{\cR_\lambda\tilde{\cR}_\lambda}}{\dim\cR_\lambda}.
\end{align*}
We may take the last line as the \emph{definition} of the desired channel: for all $\rho \in \S(\cH)$, we define
\begin{align}\label{eq:def in proof}
    \cP_\cA[\rho] \coloneqq \bigoplus_{\lambda\in\Lambda} \tr_{\cR_\lambda}[P_\lambda \rho P_\lambda] \ot \frac{I_{\tilde{\cL}_\lambda}}{\dim \cL_\lambda} \ot \frac{\ket{\Gamma}_{\cR_\lambda\tilde{\cR}_\lambda}\bra{\Gamma}_{\cR_\lambda\tilde{\cR}_\lambda}}{\dim\cR_\lambda},
\end{align}
which satisfies the first equation in~\eqref{eq:sym action}; the second equation follows at once from~\eqref{eq:twirl}.

\item
Because $\U(\cA) = \bigoplus_{\lambda\in\Lambda} I_{\cL_\lambda} \ot \U(\cR_\lambda)$, it is clear from \cref{eq:transpose} that
the projection onto the $\cA$-symmetric subspace 
takes the form
\begin{align}\label{eq:pi}
    \Pi_\cA = \bigoplus_{\lambda\in\Lambda}  I_{\cL_\lambda} \ot I_{\tilde{\cL}_\lambda} \ot \frac{\ket{\Gamma}_{\cR_\lambda\tilde{\cR}_\lambda}\bra{\Gamma}_{\cR_\lambda\tilde{\cR}_\lambda}}{\dim\cR_\lambda}.
\end{align}
Thus, \eqref{eq:def in proof} shows that the output states of~$\cP_\cA$ are supported on this subspace.
Moreover, the second equation in~\eqref{eq:sym action} shows that, for all~$g \in G$, the output states commute with~$(I \ot g\tran)$.

\item
Up to the decompositions of~$\cH$ and of~$\tilde{\cH}$ (which can be obtained from the one for~$\cH$ and the unitary~$\ket x \mapsto \ket{\tilde{x}}$,
the channel~\eqref{eq:def in proof} can be implemented by the following efficient steps:
given as input an arbitrary quantum state~$\rho$,
\begin{enumerate}[noitemsep]
    \item Measure~$\lambda$ to obtain the (unnormalized) post-measurement state $P_\lambda \rho P_\lambda$ on~$\cL_\lambda \ot \cR_\lambda$.
    \item Discard the register~$\cR_\lambda$, while keeping the register~$\cL_\lambda$.
    \item Prepare the maximally entangled state $\frac1{\sqrt{\dim\cR_\lambda}} \ket{\Gamma}_{\cR_\lambda\tilde{\cR}_\lambda}$, with $\tilde{\cR}_\lambda$ an additional register.
    \item Prepare the maximally mixed state on an additional register~$\tilde{\cL}_\lambda$.
\end{enumerate}

\item[4.]
We need to verify the following formulas:
\begin{align*}
    \cP_\cA[\rho]
&=   \Pi_\cA \parens*{ \rho \ot \Omega } \Pi_\cA
=   \sqrt{\cP_\cA[I]} \parens*{ \rho \ot I } \sqrt{\cP_\cA[I]} \\
&=   \Pi_\cA \parens*{ \tr_{\tilde{\cH}}(\Pi_\cA)^{-1/2} \rho \tr_{\tilde{\cH}}(\Pi_\cA)^{-1/2} \ot I_{\tilde{\cH}} } \Pi_\cA,
\end{align*}
where the right-hand side is the Petz map described in item~4.
The first equality immediately follows from \cref{eq:rho,eq:pi,eq:def in proof}.
To establish the other equalities, we only need to verify that
\begin{align*}
    \parens*{ I \ot \sqrt\Omega} \Pi_\cA = \sqrt{\cP_\cA[I]} = \parens*{ \tr_{\tilde{\cH}}(\Pi_\cA)^{-1/2} \ot I } \Pi_\cA;
\end{align*}
but these are all equal to
\begin{align*}
    \bigoplus_{\lambda\in\Lambda}  I_{\cL_\lambda} \ot I_{\tilde{\cL}_\lambda} \ot \frac{\ket{\Gamma}_{\cR_\lambda\tilde{\cR}_\lambda}\bra{\Gamma}_{\cR_\lambda\tilde{\cR}_\lambda}}{\sqrt{\dim\cL_\lambda \dim\cR_\lambda}},
\end{align*}
as follows readily from \cref{eq:pi,eq:def in proof}.
\qedhere
\end{enumerate}
\end{proof}

\subsection{Examples}\label{subsec:examples}
We discuss some easy examples of symmetries to illustrate the theorem.
In all examples, we take~$\cH = \tilde{\cH}$ with the computational basis, and we have $\cA = \cA\tran$ and $G = G\tran$ with respect to this choice.

\paragraph{Classical states:}
Choose $\cH = \C^d$, and let~$\cA$ be the algebra given by the diagonal matrices.
This algebra is its own commutant, i.e.\ $\cA' = \cA$, and the terms~$\cL_\lambda \ot \cR_\lambda$ in the decomposition~\eqref{eq:decomposition} are one-dimensional and spanned by~$\proj x$ for $x \in [d]$.
Note that $\rho \in \cA'$ if and only if $\rho$ is a diagonal density matrix, i.e., a classical state.
Furthermore, $\E_{\cA\tran}$ is the standard basis measurement channel.
The random purification channel obtained from \cref{thm:main technical} is given by the classical copying channel:
\[ \cP^\text{Copy}[\rho] = \sum_{x=1}^d \braket{x | \rho | x} \proj{xx}. \]

For the next three examples we let $\cH = (\C^d)^{\ot n}$.
The unitary group $\U(d)$ acts by operators~$U^{\ot n}$, with $U \in \U(d)$, while the permutation group~$S_n$ acts by operators~$R_\pi$, where $\pi \in S_n$ maps $R_\pi \ket{x_1,\dots,x_n} = \ket{x_{\pi^{-1}(1)},\dots,x_{\pi^{-1}(n)}}$ for~$x \in [d]^n$.
These two actions clearly commute. 
By Schur--Weyl duality, they generate each other's commutant:
if we denote by~$\cA_{\U(d)}$ and~$\cA_{S_n}$ the algebras generated by these group actions, then $\cA_{\U(d)}' = \cA_{S_n}$ and vice versa.
We have
\begin{align*}
    (\C^d)^{\ot n} \cong \bigoplus_{\lambda} \cV_{\lambda} \ot \cW_{\lambda},
\end{align*}
where the $\cV_\lambda$ are irreducible representations of~$\U(d)$ and the $\cW_\lambda$ are irreducible representations of~$S_n$, labeled by Young diagrams~$\lambda$ with $n$ boxes and no more than~$d$ rows.
For either algebra, the decomposition in \cref{eq:decomposition} can be obtained from this at once.

\paragraph{Permutation-invariant states:}
Let $\cA = \cA_{S_n}$, so that $\cA' = \cA_{\U(d)}$.
The states  $\rho \in \cA'$ are \emph{permutation-invariant}, i.e., $R_\pi \rho = \rho R_\pi$ for all~$\pi \in S_n$.
For instance, this is the case if $\rho = \sigma^{\ot n}$ consists of $n$ copies of the same state~$\sigma$.
Then the channel in \cref{thm:main technical} is the original random purification channel from~\cite{tang2025conjugate,pelecanos2025,girardi2025random}, which produces a state that is a purification of~$\rho$ with a random unitary tensor power~$U^{\ot n}$ applied, see \cref{cor:og channel}.

\paragraph{Werner states:}
We can also consider $\cA = \cA_{\U(d)}$, so $\cA' = \cA_{S_n}$.
Then the states $\rho \in \cA'$ are invariant under all product unitaries, i.e. $U^{\ot n} \rho = \rho U^{\ot n}$ for all~$U \in \U(d)$.
These states are called \emph{(multipartite) Werner states}.
The channel in \cref{thm:main technical} then produces a state that is a purification of~$\rho$ with a random permutation applied.

\paragraph{Symmetric Werner states:}
Finally, let $\cA$ be the algebra generated by both~$\cA_{\U(d)}$ and~$\cA_{S_n}$.
Then $\rho \in \cA'$ if it is both permutation-invariant and unitarily invariant---such states are also called \emph{symmetric Werner states.}
They can be written as
$\rho = \sum_{\lambda} p_\lambda \tau_{\lambda}$, where $(p_\lambda)$ is some probability distribution over Young diagrams~$\lambda$, and $\tau_{\lambda} = P_\lambda / \!\rk P_\lambda$ is the maximally mixed state on the isotypical component labeled by~$\lambda$.
The channel in \cref{thm:main technical} measures~$\lambda$ and prepares a copy of~$\tau_\lambda$:
\begin{align*}
    \cP[\rho] = \sum_\lambda P_\lambda \rho P_\lambda \ot \tau_{\lambda}.
\end{align*}

\section{Random purification and tomography for fermionic Gaussian states}\label{sec:fermions}
In this section we explain how to use \cref{thm:main technical} to obtain a random purification theorem for fermionic Gaussian states---also known as \emph{(quasi)free fermionic states}---and apply it to design a tomography protocol achieving optimal sample complexity.

\subsection{Gaussian states and unitaries}
Let us start by defining Gaussian states and related concepts, see e.g.\ \cite[App.~B]{helsen2022matchgate}.

\paragraph{Fermionic Hilbert space and operators:}
If we consider $m$ fermionic modes, then the fermionic Hilbert space is the so-called \emph{fermionic Fock space}, defined as
\begin{align*}
    \cH_m = \bigwedge \C^m = \bigoplus_{k=0}^m \bigwedge^k \C^m. 
\end{align*}
A natural algebra on $\cH_m$ is the algebra generated by the fermionic \emph{creation} and \emph{annihilation operators}, denoted $a_j^\dagger$ and~$a_j$ for $j \in [m]$, which satisfy the canonical anticommutation relations
\begin{align*}
    \{a_j, a_k^\dagger\} = \delta_{jk} \id, \qquad \{a_j^\dagger, a_k^\dagger\} = \{a_j, a_k\} = 0.
\end{align*}
We also define number operators as $N_j = a_j^\dagger a_j$. They satisfy $[N_j, N_k] = 0$.
For $j \in [2m]$ we define the \emph{Majorana operators}~$c_j$ by
$c_{2j-1} = a_j^\dagger + a_j$ and $c_{2j} =  i(a_j^\dagger - a_j)$.
They satisfy
\begin{align*}
    \{c_j, c_k\} = 2\delta_{jk} \id.
\end{align*}
The \emph{parity operator} is defined by $P = i^m c_1 c_2 \dots c_{2m} = (-1)^{\sum_i N_i}$, and it anticommutes with all Majorana operators.
The Hilbert space splits into the~$\pm 1$ eigenspaces of~$P$, which are called the \emph{even} and \emph{odd parity subspaces}~$\cH_m^+$ and~$\cH_m^-$. They are spanned by basis states~$\ket x$ with respectively even or odd Hamming weight $\abs{x}$.

The fermionic Fock space can be identified with the space of $m$ qubits,
\begin{align*}
    \cH_m \cong (\C^2)^{\ot m}.
\end{align*}
We write~$\ket x$ for the computational basis state with $x\in\bit^m$.
The Jordan--Wigner map allows us to define the fermionic algebra on the Hilbert space of $m$-qubits.
The mapping is determined by its action on Majorana operators
\begin{align*}
    c_{2j-1} \mapsto Z_1 \dots Z_{j-1} X_j\qquad \text{and} \qquad c_{2j} \mapsto Z_1\dots Z_{j-1} Y_j,
\end{align*}
where $P = X,Y,Z$ are the Pauli operators, and $P_j$ denotes the Pauli operator on the $j$th qubit.
Under this identification, the number operator $N_j$ acts as $\proj{1}$ on the $j$th qubit, while the parity operator is represented as
\begin{align*}
    P \mapsto Z_1 Z_2 \cdots Z_m.
\end{align*}

\paragraph{Gaussian unitaries:}
For each $R \in \SO(2m)$, there exists a unitariy $U_R$ satisfying
\begin{align}\label{eq:gaussian U}
    U_R c_j U_R^\dagger = \sum_{k=1}^{2m} R_{kj} c_k.
\end{align}
The map $R \mapsto U_R$ defines a projective representation\footnote{There is a $\pm 1$ sign ambiguity; the projective representation can be lifted to a genuine representation of $\Spin(2m)$, the double cover of $\SO(2m)$.} of the group $\SO(2m)$ acting on $\cH_m$.
The parity sectors $\cH_m^+$ and $\cH_m^-$ are irreducible subrepresentations.
We call any unitary satisfying~\eqref{eq:gaussian U} a Gaussian unitary.
A \emph{quadratic Hamiltonian} is a Hermitian operator $H$ of the form
\begin{align}\label{eq:quadratic hamiltonian}
    H = i\sum_{j,k=1}^{2m} h_{jk} c_j c_k = 2i \sum_{j<k} h_{jk} c_j c_k,
\end{align}
where $h_{jk} = -h_{kj} \in \RR$ (so $h \in \so(2m)$).
Then $U_R = \exp(i H t)$ is a Gaussian unitary, and all Gaussian unitaries can be written in this way, up to an irrelevant phase.

\paragraph{Gaussian states:}
A pure state in~$\cH_m$ is Gaussian if it is the nondegenerate ground state of a quadratic Hamiltonian of the form~\eqref{eq:quadratic hamiltonian}.
We write $\cG_m$ for the set of pure Gaussian states, and $\cG_m^+$ and $\cG_m^-$ for the subsets of even- and odd-parity Gaussian states, respectively.

More generally, a mixed Gaussian state is a Gibbs state of a quadratic Hamiltonian,
\begin{align}\label{eq:fermionic gibbs state}
    \rho = \frac1Z e^{-H}, \quad Z = \tr[e^{-H}],
\end{align}
or a limit of such states.
A suitable choice of Gaussian unitary transforms the Hamiltonian $H$ to
\begin{align}\label{eq:std H}
    H = \sum_{j=1}^m \beta_j N_j
\end{align}
up to an additive constant that is irrelevant to the Gibbs state.
Hence every Gaussian state admits a diagonal standard form.

Moreover, we will fix pure Gaussian states~$\ket{\psi_m^{\pm}} \in \cH_m^{\pm}$ as highest-weight vectors for $\so(2m)$ relative to a choice of positive roots. Every pure Gaussian state can then be expressed as $U_R\ket{\psi_+}$ or $U_R\ket{\psi_-}$ for some Gaussian unitary~$U_R$.
The Haar measure on~$\SO(2m)$ then defines unique probability measures that are invariant under Gaussian unitaries on~$\cG_m^+$ and~$\cG_m^-$; we denote both by $d\phi$.
Concretely, $\so(2m)$ (and its complexification) are realized through \cref{eq:quadratic hamiltonian}.
We may take a Cartan subalgebra $\mathfrak h$ spanned by the operators $H_j = \frac{i}{2} c_{2j-1} c_{2j} = N_j - \frac12$.
If $\eps_j \in \mathfrak{h}^*$ is defined by $\eps_j(H_k) = \delta_{jk}$, the roots of $\so(2m)$ are given by
\begin{align*}
    \Delta = \{\pm \eps_i \pm \eps_j \text{ for } 1 \leq i < j \leq m\}
\end{align*}
and we choose as the simple roots $L_i = \eps_i - \eps_{i+1}$ for $i=1,\dots,m-1$ together with $L_m = \eps_{m-1} + \eps_m$.
The basis states $\ket{x}$ are weight vectors, and $\ket{x}$ has weight $\sum_j (x_j - \frac12)\eps_j$.
With this choice, the vacuum state $\ket{0}^{\ot m}$ is a lowest-weight vector.
The highest weight vectors are given by $\ket{1}^{\ot m}$ (the fully occupied state) and $\ket{1}^{\ot (m-1)} \ot \ket{0}$, and the corresponding highest weights are given~by
\begin{align}\label{eq:highest weight}
    \alpha = (\tfrac12, \dots, \tfrac12, \tfrac12), \qquad  \beta = (\tfrac12, \dots, \tfrac12, -\tfrac12).
\end{align}
For even $m$, the state $\ket{1}^{\ot m}$ has even parity, so $\cG_m^+$ has highest weight~$\alpha$ and $\cG_m^-$ has highest weight~$\beta$, while for odd $m$ this is reversed.
Further information on these representations can be found in~\cite{fulton2013representation}, see Prop.~20.15 therein for a detailed description.

\subsection{Random purification of Gaussian states}\label{sub:gaussian purification}
Our goal in this section is to construct a random purification channel that sends every mixed Gaussian state to a random Gaussian purification.
This will follow directly from \cref{thm:main technical} but for a small subtlety that we have to address carefully---namely, when we take the tensor product~$\cH_m \ot \tilde \cH_m$ to purify the state, we want to identify this with the Hilbert space~$\cH_{2m}$ of~$2m$ fermionic modes, and we want the purification to be a Gaussian state on~$2m$ modes.
To this end we take $\tilde{\cH}_m = \cH_m$, so that $\cH_m \ot \tilde{\cH}_m \cong (\C^2)^{\ot m} \ot (\C^2)^{\ot m} \cong (\C^2)^{\ot 2m} \cong \cH_{2m}$.
Using the Jordan--Wigner representation, the Majorana operators~$\tilde c_j$ for $j \in [4m]$ take the form
\begin{align*}
    \tilde c_j &= c_j \ot \id \quad \text{ for } j \in [2m], \\
    \tilde c_{2m + j} &= P \ot c_{j}  \quad \text{ for } j \in [2m].
\end{align*}
We will now choose the following basis on~$\tilde{\cH}_m$,
\begin{align*}
    \ket{\tilde{x}} = (-1)^{\gamma(x)} \ket x \qquad (x \in \bit^m),
\end{align*}
where
\begin{align*}
    \gamma(x) = \begin{cases}
        0 & \text{ if } \abs{x} \equiv 0 \text{ or } 1 \pmod 4, \\
        1 & \text{ if } \abs{x} \equiv 2 \text{ or } 3 \pmod 4.
    \end{cases}
\end{align*}
Note that the transpose and conjugate with respect to this basis equal the transpose and conjugate, respectively, with respect to the standard basis.
With respect to our choice of bases, the standard purification of a state~$\rho \in \S(\cH_m)$ is given by
\begin{align*}
    \ket{\psi^\text{std}_\rho}
= \left( \sqrt\rho \ot \id \right) \left( \sum_{x \in \bit^m} \ket{x\tilde{x}} \right)
= \left( \sqrt\rho \ot \id \right) \left( \sum_{x \in \bit^m} (-1)^{\gamma(x)} \ket{xx} \right).
\end{align*}
We will now show that if $\rho$ is Gaussian, then its standard purification is again a Gaussian state on~$\cH_{2m}$.

\begin{lem}\label{lem:gaussian purification}
    Let $\cH_m$ denote the Hilbert space of an $m$-mode fermionic system.
    \begin{enumerate}
        \item For a Gaussian unitary $U$, both $\overline{U}$ and $U\tran$ are Gaussian unitaries. The tensor-product unitaries $U \ot \id$ and $\id \ot U$ are Gaussian on $\cH_{2m}$ as well.
        \item If $\rho \in \S(\cH_m)$ is a Gaussian state, then $\ket{\psi^\text{std}_\rho}$ is a pure Gaussian state in $\cH_{2m}^+$.
    \end{enumerate}
\end{lem}

\begin{proof}
\begin{enumerate}
\item
    Up to an overall phase, every Gaussian unitary can be written as $U = \exp(iHt)$, with $H$ quadratic in the Majorana operators.
    Because $c_j\tran = \overline{c_j} = \pm c_j$ for every $j$, both $\overline{H}$ and $H\tran$ remain quadratic in the Majorana operators; consequently, $\overline{U}$ and~$U\tran$ are Gaussian unitaries.
    Moreover, for $H$ as in \cref{eq:quadratic hamiltonian}, the identity~$P^2 = \id$ gives
    \begin{align*}
        \id \ot H = i\sum_{j,k} h_{jk} \id \ot c_j c_k = i\sum_{j,k} h_{jk} (P \ot c_j)(P \ot c_k) =  i\sum_{j,k} h_{jk} \tilde c_{2m+j} \tilde c_{2m+k}
    \end{align*}
    which is quadratic. It is clear that $H \ot \id$ is also quadratic. It follows that both $U \ot \id$ and $\id \ot U$ are Gaussian unitaries.

\item
    Consider an arbitrary Gaussian state $\rho \in \S(\cH_m)$.
    Continuity and compactness allow us, without loss of generality, to take $\rho$ to have the form in \cref{eq:fermionic gibbs state}.
    There must then exist a Gaussian unitary~$U$ such that (see \cref{eq:std H}),
    \begin{align*}
        \sigma \coloneqq U \rho U^\dagger = \frac{e^{- H}}{\tr[e^{- H}]}, \qquad H = \sum_{j=1}^m \beta_j N_j.
    \end{align*}
    The standard purifications of $\rho$ and $\sigma$ are related by
    \begin{align*}
        \ket{\psi^\text{std}_\rho} = (U^\dagger \ot U\tran) \ket{\psi^\text{std}_\sigma}
    \end{align*}
    by the transpose trick.
   
    By the first part of the lemma, $U^\dagger \ot U\tran$ is Gaussian; it therefore remains only to prove that $\ket{\psi^\text{std}_\sigma}$ is Gaussian.
    To do so, we start by writing
    \begin{align*}
        \sqrt{\sigma} = \frac{\exp(- H/2)}{\sqrt{\tr[\exp(- H])}} = \prod_{i=1}^m \frac{\exp(- \beta_i N_i / 2)}{\sqrt{1 + e^{- \beta_i }}} .
    \end{align*}
    Using the identity
    \begin{align*}
        \exp(- \beta_i N_i / 2) = \proj{0}_i + e^{- \beta_i / 2} \proj{1}_i .
    \end{align*}
    we see that
    \begin{align*}
        \ket{\psi^\text{std}_\sigma} = \sum_{x \in \bit^m} \parens*{ \prod_{j=1}^m \cos(\theta_j)^{1 - x_j}\sin(\theta_j)^{x_j} } (-1)^{\gamma(x)} \ket{xx}
    \end{align*}
    for angles $\theta_j \in \RR$ chosen such that
    \begin{align*}
        \cos(\theta_j) = \frac{1}{\sqrt{1 + e^{- \beta_j}}}, \qquad \sin(\theta_j) = \frac{e^{- \beta_j / 2}}{\sqrt{1 + e^{- \beta_j}}} \, .
    \end{align*}
    The claim then follows from \cref{lem:gaussian diagonal purification} below.
\end{enumerate}
\end{proof}

\begin{lem}\label{lem:gaussian diagonal purification}
    Let $\theta_j$ for $j \in [m]$ be real numbers.
    Then, the state
    \begin{align*}
        \ket{\psi} = \sum_{x \in \bit^m} \parens*{ \prod_{j=1}^m \cos(\theta_j)^{1 - x_j}\sin(\theta_j)^{x_j} } (-1)^{\gamma(x)} \ket{xx}
    \end{align*}
    is a pure Gaussian state with even parity, i.e. $\ket{\psi} \in \cH_{2m}^+$.
\end{lem}
\begin{proof}
    For $k=0,\dots,m$, define
    \begin{align*}
        \ket{\psi_k} = \sum_{\substack{y \in \bit^k,\\ x = (y,0^{m-k})}} \parens*{ \prod_{j=1}^k \cos(\theta_j)^{1 - x_j}\sin(\theta_j)^{x_j} } (-1)^{\gamma(x)} \ket{xx}.
    \end{align*}
    We prove by induction on $k$ that every $\ket{\psi_k}$ is Gaussian.
    The base case is immediate: $\ket{\psi_0} = \ket{0}^{\ot 2m}$ is Gaussian.
    Assume now that $\ket{\psi_{k-1}}$ is Gaussian.
    Set $H_k = -i \tilde c_{2k} \tilde c_{2m+2k}$, then
    \begin{align*}
        H_k = -i Y_k Z_{k} Z_{k+1}\dots Z_{m + k - 1} Y_{m+k} = X_k Z_{k+1}\dots Z_{m + k - 1} Y_{m+k}.
    \end{align*}
    In order to show that $\ket{\psi_k}$ is Gaussian, we will now argue that the Gaussian unitary $U_k = \exp(-i H_k \theta_k)$ sends $\ket{\psi_{k-1}}$ to $\ket{\psi_k}$.
    Indeed, acting on $\ket{\psi_{k-1}}$ gives
    \begin{align*}
        U_k \ket{\psi_{k-1}} = \sum_{\substack{y \in \bit^{k-1},\\ x = (y,0^{m-k +1})}} \parens*{ \prod_{j=1}^{k-1} \cos(\theta_j)^{1 - x_j}\sin(\theta_j)^{x_j} } (-1)^{\gamma(x)}\left(\cos(\theta_k) \id - (-1)^{\abs{x}} i \sin(\theta_k) X_k Y_{m+k} \right)\ket{xx}.
    \end{align*}
    We regroup the expression as a sum over $\tilde y \in \bit^k$, with $\tilde x = (\tilde y,0^{m-k})$.
    The first summand corresponds to~$\tilde y = (y,0)$, whereas the second corresponds to~$\tilde y = (y,1)$.
    In the former case we have $\ket{xx} = \ket{\tilde x \tilde x}$, while in the latter case we have $X_k Y_{m+k} \ket{xx} = i \ket{\tilde x \tilde x}$, so we can rewrite this as
    \begin{align*}
        U_k \ket{\psi_{k-1}} = \sum_{\substack{\tilde y \in \bit^{k},\\ \tilde x = (\tilde y,0^{m-k})}} \parens*{ \prod_{j=1}^{k} \cos(\theta_j)^{1 - \tilde x_j}\sin(\theta_j)^{\tilde x_j} } \alpha(\tilde x) \ket{\tilde x \tilde x}.
    \end{align*}
    Here $\alpha(\tilde x) = (-1)^{\gamma(x) + \tilde x_k \abs{x}}$.
    Checking the two cases $\tilde x_k \in \bit$ shows that $\alpha(\tilde x) = (-1)^{\gamma(\tilde x)}$.
    Hence $U_k \ket{\psi_{k-1}} = \ket{\psi_k}$.
    Since $\ket{\psi_m} = \ket{\psi}$, the induction completes the proof.
\end{proof}

{It remains to establish the required symmetry property of Gaussian states.}
The following lemma avoids appealing to Howe duality.

\begin{lem}\label{lem:gaussian state in algebra}
    If $\sigma$ is an $m$-mode fermionic Gaussian state, then $\sigma^{\ot n} \in \C\{U_R^{\ot n}  :  R \in \SO(2m) \}$.
\end{lem}

\begin{proof}
    It is enough to work with the dense family of fermionic Gaussian states of the form~\eqref{eq:fermionic gibbs state}, which are Gibbs states with Hamiltonian~$H$.
    Set $\cB = \C\{U_R^{\ot n}  : R \in \SO(2m) \}$.
    Since $V(t) = \exp(iHt)^{\ot n} \in \cB$, differentiation at $t=0$ shows that
    \begin{align*}
         -i\partial_t \big|_{t=0} V(t) = \sum_{j=1}^n H_j \in \cB
    \end{align*}
    where $H_j$ acts as $H$ on the $j$th tensor factor of $\cH_m$ and as $\id$ on all remaining factors.
    We conclude that $\exp(- \sum_j H_j) = \exp(-H)^{\ot n} \in \cB$ as well.
\end{proof}

We thus obtain the Gaussian random purification result as stated in the introduction.

\begin{proof}[Proof of \cref{cor:fermions}]
The result follows directly from \cref{thm:main technical} and \cref{lem:gaussian state in algebra}, together with \cref{lem:gaussian purification}, which shows that the standard purification is Gaussian.
\end{proof}

\subsection{Tomography of pure Gaussian states}\label{sub:gaussian pure tomo}
We first propose a tomography scheme for pure Gaussian states, analogous to Hayashi's POVM for ordinary pure state tomography~\cite{hayashi1998asymptotic}.
Recall that $\cG_m^{\pm}$ consists of the even- and odd-parity pure Gaussian states, respectively.
Given $n$ copies of the system, let $\cV_{n,m}^\pm$ denote the span of~$\ket{\psi}^{\ot n}$ for~$\ket{\psi} \in \cG_m^{\pm}$.
The space $\cV_{1,m}^\pm = \cH_m^\pm$ is an irreducible $\so(2m)$-representation with highest weight~$\alpha$ or~$\beta$ (according to the parity of $m$; see \cref{eq:highest weight}) and highest-weight vector $\ket{\psi_m^\pm} \in \cG_m^{\pm}$.
Consequently, for every~$n$, $\cV_{n,m}^\pm$ is an irreducible $\so(2m)$-representation with highest weight~$n\alpha$ or~$n\beta$ and highest-weight vector~$\ket{\psi_m^{\pm}}^{\ot n}$.

\begin{lem}\label{lem:dimension multicopy fermions}
    The two parity sectors have equal dimension, $\dim \cV_{n,m}^+ = \dim \cV_{n,m}^- = d_{n,m}$, where
    \begin{align}\label{eq:dimension multicopy fermions}
        d_{n,m} = \prod_{1 \leq j < k \leq m} \frac{2m + n - (j+k)}{2m - (j+k)}
    \end{align}
\end{lem}
\begin{proof}
This follows from the Weyl character formula.
For $\so(2m)$, it gives the following formula for the dimension of an irreducible representation with highest weight~$\lambda$~\cite[Eq.~(24.41)]{fulton2013representation}:
\begin{align*}
    d_{n,m} = \prod_{1 \leq j < k \leq m} \frac{M_j^2 - M_k^2}{K_j^2 - K_k^2} \qquad K_j = m - j, \ M_j = \lambda_j + m - j.
\end{align*}
Note that $M_m^2 = \lambda_m^2$ does not depend on the sign of $\lambda_m$, and in particular $\dim \cV_+ = \dim \cV_-$.
By substituting~$\lambda = n\alpha$ or~$n\beta$, we obtain \cref{eq:dimension multicopy fermions}.
\end{proof}

Recall that we write $\dd \phi$ for the unique probability measure on~$\cG_m^\pm$ that is invariant under Gaussian unitaries.
Consider the operator
\begin{align*}
    \Pi_{n,m}^\pm = d_{n,m} \, \int_{\cG_m^{\pm}} \dd \phi \, \proj{\phi}^{\ot n}
\end{align*}
which commutes with every Gaussian unitary and is therefore an intertwiner for the action of~$\so(2m)$.
By Schur's lemma, we have that $\Pi_{n,m}^\pm$ must be proportional to the projection on $\cV_{n,m}^\pm$.
Then, by taking the trace, we see that~$\Pi_{n,m}^{\pm}$ is exactly equal to the projection operator on $\cV_{n,m}^\pm$.
This implies that $\mu_{n,\pm}$, defined by
\begin{align}\label{eq:haar measure projection}
    \dd \mu_{n,m}^\pm(\phi) = d_{n,m} \,\dd \phi \, \proj{\phi}^{\ot n},
\end{align}
is a POVM on~$\cV_{n,m}^\pm$ with outcomes in~$\cG_m^\pm$, where $d_{n,m}$ is given in \cref{eq:dimension multicopy fermions}.
Clearly, we can combine $\mu_{n,m}^\pm$ into a single POVM~$\mu_{n,m}$ on~$\cV_{n,m}^+ \oplus \cV_{n,m}^-$ with outcomes in~$\cG_m$.

For a fixed $\psi$, let us now consider the fidelity $F = \abs{\braket{\psi | \hat{\psi}}}$ as a random variable (it depends on the measurement outcome, which is a random variable for fixed $\psi$), and compute moments of this random variable.

\begin{lem}\label{lem:moments pure fermion measurement}
Let $\psi \in \cG_m$ be a pure fermionic Gaussian state.
Let $\hat\psi$ be the outcome resulting from the POVM~$\mu_{n,m}$ on the input $\psi^{\ot n}$.
Then for integer $k \geq 1$, we can evaluate
\begin{align*}
    \Ex F^{2(k-n)} = \Ex \abs{\braket{\psi | \hat{\psi}}}^{2(k - n)} = \frac{d_{n,m}}{d_{k,m}} = \prod_{1 \leq i < j \leq m} \frac{2m + n - (i+j)}{2m + k - (i+j)}.
\end{align*}
\end{lem}

\begin{proof}
    We consider the case where $\ket\psi \in \cG_m^+$ has even parity (the odd case is identical):
    \begin{align*}
        \Ex \abs{\braket{\psi | \hat{\psi}}}^{2(k - n)}
        &= \int_{\phi \in \cG_m^+} \tr[ \dd \mu_{n,m}(\phi) \proj{\psi}^{\ot n} ] \ \abs{\braket{\psi|\phi}}^{2(k - n)} \\
        &= d_{n,m} \int_{\phi \in \cG_m^+} \dd\phi \, \tr[ \proj\phi^{\ot n} \proj{\psi}^{\ot n} ] \ \abs{\braket{\psi|\phi}}^{2(k - n)} \\
        &= d_{n,m} \int_{\phi \in \cG_m^+} \dd\phi \, \abs{\braket{\psi|\phi}}^{2k} \\
        &= \frac {d_{n,m}} {d_{k,m}} \int_{\phi \in \cG_m^+} \tr[ \dd \mu_{k,m}(\phi) \proj{\psi}^{\ot k} ] \\
        &= \frac {d_{n,m}} {d_{k,m}} \bra\psi^{\ot k} \Pi_{k,m}^+ \ket\psi^{\ot k}  \\
        &= \frac {d_{n,m}} {d_{k,m}}.
    \end{align*}
    This ratio of dimensions can be given an explicit expression by \cref{lem:dimension multicopy fermions} as
\begin{align*}
    \frac{d_{n,m}}{d_{k,m}} &= \prod_{1 \leq i < j \leq m} \frac{2m + n - (i+j)}{2m + k - (i+j)}. \qedhere
\end{align*}
\end{proof}
Let us first bound the expectation value of the squared fidelity. For the fidelity $F\coloneqq\abs{\braket{\psi|\hat\psi}}$, by taking \(k=n+1\) in \cref{lem:moments pure fermion measurement}, we obtain
\begin{align*}
    \Ex F^2
    &=
    \frac{d_{n,m}}{d_{n+1,m}}                                                     \\
    &=
    \prod_{1 \leq i < j \leq m}
    \frac{2m+n-(i+j)}{2m+n+1-(i+j)}                                                \\
    &=
    \prod_{1 \leq i < j \leq m}
    \left(
        1-\frac{1}{2m+n+1-(i+j)}
    \right)                                                                        \\
    &\geq
    \left(1-\frac1n\right)^{m(m-1)/2}.
\end{align*}
Therefore,
\[
    \Ex(1- F^2)
    \leq
    1-
    \left(1-\frac1n\right)^{m(m-1)/2}
    \leq
    \frac{m(m-1)}{2n}.
\]
Applying Markov's inequality to the nonnegative random variable \(1-F\), we
obtain
\begin{align*}
    \Pr\mleft(F^2\leq 1-\eps^2\mright)
    =
    \Pr\mleft(1-F^2\geq \eps^2\mright)
    \leq
    \frac{\Ex(1-F^2)}{\eps^2}
    \leq
    \frac{m(m-1)}{2n\eps^2}.
\end{align*}
Thus \(n=\bigO(m^2/(\eps^2\delta))\) copies already suffice to obtain an
\(\eps\)-accurate estimate with probability at least \(1-\delta\).  This is the optimal scaling in $m$ and $\eps$. Taking into account different moments, we can prove a slightly sharper statement that also has optimal dependence on $\delta$.

\begin{cor}[Tomography of pure fermionic Gaussian states]\label{cor:pure fermion tomography}
    Let \(0<\eps,\delta<1\).  Given
    \begin{align*}
        n
        =
        \bigO\mleft(
            \frac{m^2+\log(\delta^{-1})}{\eps^2}
        \mright)
    \end{align*}
    copies of a pure fermionic Gaussian state~\(\psi\in\cG_m\), the measurement
    \(\mu_{n,m}\) outputs an estimate~\(\hat\psi\) such that
    \[
        \abs{\braket{\psi|\hat\psi}}^2
        \geq
        1-\eps^2
    \]
    with probability at least \(1-\delta\).
\end{cor}

\begin{proof}
    Write $F\coloneqq\abs{\braket{\psi|\hat\psi}}$. Let $r\coloneqq\ceil*{\frac n2}$.
    Then \(n-r=\floor*{n/2}\).  By
    \cref{lem:moments pure fermion measurement}, applied with \(k=r\), we have
    \[
        \Ex F^{2(r-n)}
        =
        \frac{d_{n,m}}{d_{r,m}}.
    \]
    We first estimate the dimension ratio.  Using
    \cref{eq:dimension multicopy fermions},
    \begin{align*}
        \frac{d_{n,m}}{d_{r,m}}
        &=
        \prod_{1\leq i<j\leq m}
        \frac{2m+n-(i+j)}{2m+r-(i+j)}=
        \prod_{1\leq i<j\leq m}
        \left(
            1+
            \frac{n-r}{2m+r-(i+j)}
        \right).
    \end{align*}
    Since \(i+j\leq 2m-1\), the denominator in each factor satisfies
\[
    2m+r-(i+j)\geq r+1.
\]
Moreover, \(n-r\leq r\). Hence every factor in the product is at most \(2\). Since there are \(m(m-1)/2\) factors,
    \[
        \frac{d_{n,m}}{d_{r,m}}
        \leq
        2^{m(m-1)/2}.
    \]
We now apply Markov's inequality to the nonnegative random variable
    \(F^{-2(n-r)}\):
    \begin{align*}
        \Pr\mleft(F^2\leq 1-\eps^2\mright)
        &=
        \Pr\mleft(
            F^{-2(n-r)}
            \geq
            (1-\eps^2)^{-(n-r)}
        \mright)                                                          \\
        &\leq
        (1-\eps^2)^{n-r}
        \Ex F^{-2(n-r)}                                                     \\
        &=
        (1-\eps^2)^{n-r}
        \frac{d_{n,m}}{d_{r,m}}                                             \\
        &\leq
        (1-\eps^2)^{\floor*{n/2}}
        2^{m(m-1)/2}.
    \end{align*}
    Using \(1-\eps\leq e^{-\eps}\), we obtain
    \[
        \Pr\mleft(F^2\leq 1-\eps^2\mright)
        \leq
        \exp\mleft(
            -\eps^2\floor*{\frac n2}
            +
            \frac{m(m-1)}2\log 2
        \mright).
    \]
    Thus it suffices to choose $n
        =
        \bigO\mleft(
            \frac{m^2+\log(\delta^{-1})}{\eps^2}
        \mright)$, to make the failure probability at most \(\delta\).  This is as claimed.
\end{proof}

This sample complexity is optimal, as stated in
\cref{thm:tomo fermion lower bound}, proven in
\cref{sec:fermions lower bound}.

\subsection{Tomography of mixed Gaussian states}
Finally, we can combine the random Gaussian purification channel (\cref{sub:gaussian purification}) with our result on tomography for Gaussian pure states (\cref{sub:gaussian pure tomo}) to obtain a tomography algorithm for Gaussian mixed states with quadratic sample complexity.

\begin{proof}[Proof of \cref{thm:tomo fermion}]
The tomography algorithm proceeds as follows:
\begin{enumerate}[noitemsep]
\item First apply the random purification channel $\cP^{\text{Fermi}}_{n,m}$ of \cref{cor:fermions} to $\sigma^{\ot n}$
\item Then perform the pure-state tomography protocol of \cref{cor:pure fermion tomography} on the output, to obtain an estimate $\hat{\psi}$ of a $2m$-mode pure fermionic Gaussian state.
\item Return $\hat\sigma$, which is obtained by taking the partial trace of $\hat\psi$ over the purifying system.
\end{enumerate}
The analysis is completely straightforward. By \cref{cor:fermions}, applying $\cP^{\text{Fermi}}_{n,m}$ is equivalent to preparing $n$ copies of a random Gaussian purification of $\sigma$.
By \cref{cor:pure fermion tomography}, it follows that the state $\psi$ satisfies $\abs{\braket{\psi | \hat \psi}}^2 \geq 1 - \eps^2$ for \emph{some} Gaussian purification~$\ket{\psi}$ of~$\sigma$, with high probability.
By taking a partial trace we obtain an estimate~$\hat\sigma$ such that $F(\sigma,\hat \sigma)^2 \geq \abs{\braket{\psi | \hat \psi}}^2 \geq 1 - \eps^2$, and hence $P(\sigma, \hat \sigma) \leq \eps$.
\end{proof}

For a general full-rank state on $\C^d$, a purification has $O(d^2)$ degrees of freedom.
Accordingly, pure-state tomography uses a number of samples scaling as $d$, whereas mixed-state tomography scales as $d^2$.
For fermions, purifying an $m$-mode system produces a Hilbert space with $2m$ modes, yielding only a constant multiplicative overhead.
This again squares the Hilbert space dimension, but since an $m$-mode fermionic state is specified by~$O(m^2)$ parameters, it is natural to find that fermionic Gaussian tomography has a sample complexity scaling with $m^2$ for both pure and mixed states.

\subsection{Sample complexity lower bound}\label{sec:fermions lower bound}
We now show that $\bigO(m^2/\eps^2)$ is in fact has a matching lower bound, and hence is the \emph{optimal} sample complexity.
Because pure and mixed states have the same scaling of sample complexity, it suffices to derive the lower bound for pure-state estimation, substantially simplifying the argument.
We adapt the lower-bound argument of~\cite{scharnhorst2025optimal} (see also \cite{harrow2013church}) for arbitrary pure-state tomography.
Again, the role of the symmetric subspace is now taken by the spaces $\cV_{n,m}^{\pm}$, otherwise only minor modifications are needed.

Suppose we have an algorithm taking $n$ copies of an even-parity pure Gaussian fermionic state $\psi \in \cG_m^+$, and returning an estimate $\ket{\hat{\psi}} \in \cH_m^+$.
We can assume without loss of generality that the estimate is itself an even-parity pure Gaussian fermionic state.
Such a procedure is described by a POVM~$\mu$ and induces the outcome distribution $P_{\psi}$ defined by
\begin{align*}
    \dd P_{\psi}(\phi) = \tr[\dd \mu(\phi) \proj{\psi}^{\ot n}]
\end{align*}
when the input consists of $n$ copies of~$\psi$.
Given this measurement, we are now going to bound its average performance when the target state is chosen uniformly at random.

\begin{lem}\label{lem:moments fermion fidelity}
    Suppose that $\psi$ is distributed according to the Haar measure on $\cG_m^+$ or on $\cG_m^-$, and $\hat{\psi}$ denotes the estimate from the measurement~$\mu$ on $n$ copies of $\psi$.
    One then has the following bound on the expected fidelity:
    \begin{align*}
        \Ex_{\psi} \, \Ex_{\hat\psi \sim P_{\psi}} \, \abs{\braket{\psi | \hat\psi}}^{2k} \leq \frac{d_{n,m}}{d_{n+k,m}},
    \end{align*}
    with $d_{n,m}$ as in \cref{eq:dimension multicopy fermions}.
\end{lem}

\begin{proof}
    We treat only the even-parity case, since the odd-parity case is identical.
    Expanding the two expectations gives
    \begin{align*}
        \Ex_{\psi} \, \Ex_{\hat\psi \sim P_{\psi}} \, \abs{\braket{\psi | \hat\psi}}^{2k}
        &= \int_{\cG_m^+} \int_{\cG_m^+} \tr\bigl[ \dd \mu(\hat\psi) \proj{\psi}^{\ot n} \bigr] \, \abs{\braket{\psi | \hat\psi}}^{2k} \dd \psi \\
        &= \int_{\cG_m^+} \int_{\cG_m^+} \tr\bigl[ (\dd \mu(\hat\psi) \ot \proj{\hat{\psi}}^{\ot k}) \proj{\psi}^{\ot(n+k)} \bigr] \dd \psi \\
        &= \frac{1}{d_{n+k,m}} \int_{\cG_m^+} \tr\bigl[ (\dd \mu(\hat\psi) \ot \proj{\hat{\psi}}^{\ot k}) \Pi_{n+k,m}^+ \bigr],
    \end{align*}
    where the last identity uses \cref{eq:haar measure projection}.
    The result now follows from
    \begin{align*}
        \int_{\cG_m^+} \tr[(\dd \mu(\hat\psi) \ot \proj{\hat{\psi}}^{\ot k}) \Pi_{n+k,m}^+]
        &\leq \int_{\cG_m^+} \tr[\dd \mu(\hat\psi) \ot \proj{\hat{\psi}}^{\ot k}] \\
        &= \int_{\cG_m^+} \tr[\dd \mu(\hat\psi)]
        = \tr[\Pi_{n,m}^+] = d_{n,m},
    \end{align*}
    where we use that $\Pi_{n+k,m}^+ \leq \id$ and we may assume without loss of generality that $\mu$ is supported on $\cV_{n,m}^+$, the space spanned by $n$~identical copies of even parity Gaussian fermionic states.
\end{proof}

We now use this to bound the required number of copies for tomography of Gaussian fermionic states, proving \cref{thm:tomo fermion lower bound}.

\begin{proof}[Proof of \cref{thm:tomo fermion lower bound}]
It suffices to restrict to even parity states.
    Let $\mu$ be an $n$-copy measurement, for which we assume that for any $\psi \in \cG_m^+$, the resulting estimate $\hat{\psi}$ satisfies $\abs{\braket{\psi|\hat\psi}}^2 \geq 1 - \eps^2$ with probability at least 0.99.
        Together with the fact that $\abs{\braket{\psi | \hat\psi}}^{2k} \geq 0$, we have that for every $\psi \in \cG_m^+$,
    \begin{align*}
        \Ex_{\hat\psi \sim P_{\psi}} \, \abs{\braket{\psi | \hat\psi}}^{2k} \geq  0.99(1-\eps^2)^{k} .
    \end{align*}
    On the other hand, by \cref{lem:moments fermion fidelity} and \cref{lem:dimension multicopy fermions} we have
    \begin{align*}
        \Ex_{\psi} \, \Ex_{\hat\psi \sim P_{\psi}} \, \abs{\braket{\psi | \hat\psi}}^{2k} &\leq \frac{d_{n,m}}{d_{n+k,m}}
        = \prod_{1 \leq i < j \leq m} \frac{2m + n - (i+j)}{2m + n + k - (i+j)}\\
        &= \prod_{1 \leq i < j \leq m} \left(1 - \frac{k}{2m + n + k - (i+j)}\right)
        \leq \left(1 - \frac{k}{2m + n + k}\right)^{m(m-1)/2}.
    \end{align*}
    We conclude that for every $k$
    \begin{align*}
        0.99(1-\eps^2)^{k} \leq \Ex_{\psi} \, \Ex_{\hat\psi \sim P_{\psi}} \, \abs{\braket{\psi | \hat\psi}}^{2k} \leq \left(1 - \frac{k}{2m + n + k}\right)^{m(m-1)/2}.
    \end{align*}
    Using $1 + xy \leq (1+x)^y \leq e^{xy}$, which holds for $x \geq -1$, we deduce that
    \begin{align*}
        0.99(1-k\eps^2) \leq \exp\mleft( -\frac{km(m-1)}{2(2m + n + k)}\mright).
    \end{align*}
    We now choose $k = \lfloor 1/(4\eps^2) \rfloor$, so $0.99(1-k\eps^2) \geq e^{-1/2}$ and taking logarithms we get
    \begin{align*}
        \frac{km(m-1)}{2(2m + n + k)} \leq \frac12,
    \end{align*}
    which we can rewrite as
    \begin{equation*}
        n \geq k(m^2 - m - 1) - 2m = \Omega(m^2/\eps^2).
        \qedhere
    \end{equation*}
        This proves the result for constant probability of error $\delta$. The scaling of the lower bound with $\delta$ follows from the reduction to distinguishing two nearby states as in \cref{eq:distinguising two states}.
\end{proof}

\section{Random purification and tomography for bosonic Gaussian states}\label{sec:bosons}
The goal of this section is to follow a similar approach as for fermionic Gaussian states to derive the optimal sample complexity for bosonic Gaussian states. Along the way, we deal with some technical challenges arising from the fact that the Hilbert spaces are infinite-dimensional and the groups of Gaussian unitaries are non-compact. In particular, there is no random purification channel for arbitrary Gaussian states, but we show it does exist for gauge-invariant bosonic Gaussian states.

\subsection{Gaussian states and unitaries}\label{sub:boson prelims}
We start by defining these states and related concepts, see e.g.\ \cite{weedbrook2012gaussian, BUCCO, HOLEVO}.

\paragraph{Bosonic Hilbert space and operators:}
Consider $m$ bosonic modes. The associated Hilbert space is given by
\begin{align*}
    \cH_m = \L^2(\R^m).
\end{align*}
There is a natural algebra on $\cH_m$ generated by the bosonic creation and annihilation operators -- denoted as $a_j^\dagger$ and~$a_j$ for $j \in [m]$, respectively -- which satisfy the canonical commutation relations
\begin{align*}
    [a_j, a_k^\dagger] = \delta_{jk} \id, \qquad [a_j^\dagger, a_k^\dagger] = [a_j, a_k] = 0.
\end{align*}
The corresponding position and momentum quadratures are $q_j = (a_j^\dagger + a_j)/\sqrt2$ and $p_j = i(a_j^\dagger - a_j)/\sqrt2$, respectively, for~$j\in[m]$.
The total number operator is defined as~$N = \sum_j a_j^\dagger a_j$.
We call the associated unitaries~$e^{i\theta N}$ gauge transformations, they satisfy $e^{i\theta N} a_j e^{-i\theta N} = e^{-i\theta} a_j$ for every~$j\in[m]$. 

\paragraph{Gaussian unitaries:}
For $S \in \Sp(2m)$, there are unitaries $U_S$ that map
\begin{align}\label{eq:gaussian U boson}
    U_S r_j U_S^\dagger = \sum_{k=1}^{2m} S_{kj} r_k,
\end{align}
where $(r_1,\dots,r_{2m}) = (q_1,p_1,\dots,q_m,p_m)$ are the quadrature operators.
This defines a projective representation of $\Sp(2m)$ on~$\cH_m$, which can be lifted to a genuine representation of the metaplectic group~$\Mp(2m)$, the double cover of $\Sp(2m)$.
For $S \in \Sp(2m)$ we denote the associated symplectic Gaussian unitary by $U_S$.
A unitary is called a Gaussian unitary if it is a product of a unitary as in~\eqref{eq:gaussian U boson} with an element of the Heisenberg--Weyl group (that is, a unitary generated by a displacement operator).
This way, the group of Gaussian unitaries form a projective representation of the affine symplectic group~$\ASp(2m)$.

A Hamiltonian~$H$ is quadratic when it has degree at most two in the creation and annihilation operators, or equivalently in the quadrature operators.
For such Hamiltonians $U_S = \exp(i H t)$ is Gaussian, and, up to an overall phase, every Gaussian unitary admits such a representation.
We say that $H$ is gauge invariant, or number preserving, or \emph{passive} when it commutes with every gauge transformation; equivalently, it commutes with the number operator.
Such Hamiltonians are of the form
\begin{align}\label{eq:quadratic gauge invariant hamiltonian}
    H = \sum_{j,k=1}^m h_{jk} a_j^\dagger a_k
\end{align}
where $h = h^\dagger$ (so $ih$ is in the Lie algebra~$\mathfrak{u}(m)$), up to an irrelevant additive constant.
These gauge-invariant unitaries define a projective representation of~$\U(m)$, the maximal compact subgroup of~$\Sp(2m)$.
Every number-preserving Hamiltonian of the form~\eqref{eq:quadratic gauge invariant hamiltonian} can be diagonalized by a passive Gaussian unitary as
\begin{align}\label{eq:diag gauge invariant hamiltonian}
    H = \sum_{j=1}^m \beta_j a_j^\dagger a_j + c
\end{align}
with $\beta_j \in \R$ and $c \in \R$.

\paragraph{Gaussian states:}
A Gaussian state is a Gibbs state
\begin{align}\label{eq:bosonic gibbs state}
    \rho = \frac1Z e^{-H}, \quad Z = \tr[e^{-H}],
\end{align}
for a quadratic Hamiltonian $H$ which is bounded from below, or a (trace norm) limit of such states.
One can always use a Gaussian unitary that transforms the Hamiltonian~$H$ into that of independent harmonic oscillators, i.e.,
\begin{align}\label{eq:std H boson}
    H = \sum_{j=1}^m \beta_j \left( a_j^\dagger a_j + \frac12 \right)
\end{align}
where $\beta_j > 0$ for every $j$, up to an irrelevant additive constant.

A Gaussian state $\rho$ has mean zero if $\tr[\rho a_i] = \tr[\rho a_i^\dagger] = 0$ for $i=1,\dots,m$.
We denote by $\cG_m^0$ the set of pure bosonic $m$-mode Gaussian states of mean zero. It is the orbit of the vacuum state under the symplectic Gaussian unitaries.

A Gaussian state is called \emph{gauge-invariant} (or \emph{number-preserving)} if it commutes with all gauge transformations (equivalently, with the number operator).
In this case, $H$ must be gauge-invariant and, by a passive Gaussian unitary, it can be brought into the form \cref{eq:std H boson}.

\begin{lem}\label{lem:gaussian_purification_bos}
    Let $\cH_m$ be the Hilbert space of $m$ bosonic modes, with the Fock basis.
    \begin{enumerate}
        \item If $U$ is a symplectic Gaussian unitary, then $\overline{U}$ and $U\tran$ are also Gaussian unitaries. In addition, $U \ot \id$ and $\id \ot U$ are Gaussian unitaries on $\cH_{2m}$.
        \item If $\rho \in \S(\cH_m)$ is Gaussian, then $\ket{\psi^\text{std}_\rho}$ is a pure Gaussian state in $\cH_{2m}$. If $\rho$ is a gauge-invariant state, then the mean photon number of the standard purification is twice that of $\rho$.
    \end{enumerate}
\end{lem}

\begin{proof}
Complex conjugation and transposition are defined with respect to the Fock basis. With this choice, the first statement follows from the fact that quadratic Hamiltonians remain quadratic under transpose and complex conjugation.
Let $\ket{x}$ denote the $m$-mode Fock state associated with $x \in \N^m$. Defining the unnormalized maximally entangled state on two systems with $m$ modes as the formal vector
\begin{align}\label{def_Gamma}
    \ket{\Gamma}_{AB} \coloneqq \sum_{x\in \N^m} \ket{x}_{A} \ket{x}_{B}\,,
\end{align}
then the \emph{standard purification} of an $m$-mode state $\rho$ is given by
\begin{equation}\label{def_can_pur}
    \ket{\psi^\text{std}_\rho}_{AB} \coloneqq (\sqrt{\rho_A} \otimes I_B) \, \ket{\Gamma}_{AB}.
\end{equation}
Note that this is well-defined since $\rho$ is a trace-class operator.
In particular, if $\tau$ is the Gibbs state of a Hamiltonian as in \cref{eq:std H boson}, so $\tau$ is a tensor product state of $m$ bosonic modes, each of which is in a state proportional to
\begin{align*}
    \sum_{x_j=0}^{\infty} e^{-\beta_j x_j} \proj{x_j},
\end{align*}
and the standard purification $\ket{\psi^\text{std}_\tau}_{AB}$ is a tensor product of states proportional to
\begin{align*}
    \sum_{x_j=0}^{\infty} e^{-\beta_j x_j/2} \ket{x_j x_j}.
\end{align*}
It is easy to see that the mean particle number of $\psi^\text{std}_\tau$ is twice that of $\tau$.
More generally, let $\rho$ be an arbitrary Gaussian state, and write $\rho = U \tau U^\dagger$ for a Gaussian unitary $U$, where $\tau$ is the Gibbs state of a Hamiltonian as in \cref{eq:std H boson}.
Then
\begin{align}\label{eq:purification_trick}
    \ket{\psi^\text{std}_\rho}_{AB}=(\sqrt{\rho_A}\otimes I_B)\ket{\Gamma} =(U\sqrt{\tau} U^\dagger\otimes I_B) \ket{\Gamma}_{AB} = (U \sqrt{\tau} \otimes \bar{U}) \ket{\Gamma}_{AB} = (U\otimes \bar{U} )\ket{\psi^\text{std}_\tau}_{AB}\,,
\end{align}
by the transpose trick.
Now, $U\otimes \bar{U}$ is a Gaussian unitary, so $\ket{\psi^\text{std}_\rho}_{AB}$ is a Gaussian state. If $\rho$ is gauge-invariant, we may take $U = U_O$ to be a passive Gaussian unitary, which preserves particle number, and the mean particle number of $\ket{\tau}$ is twice that of $\tau$.
\end{proof}

\subsection{Challenges for tomography for bosonic Gaussian states}
Trying to follow the same tomography strategy as in the fermionic case for general bosonic Gaussian states runs into two difficulties, closely related:
\begin{enumerate}
    \item There is no (normalized) Haar integral, since the group is not compact.
    \item In the random purification channel, one needs to prepare maximally mixed states, but in the bosonic case these would be on infinite-dimensional spaces, which is not well-defined.
\end{enumerate}
The high-level solution to these two problems is as follows.
\begin{enumerate}
    \item For $n \geq 2m+1$, the matrix elements of $\proj{\psi}^{\ot n}$ define integrable functions over the non-compact symplectic group; this means that the covariant measurement can be defined.
    \item In fact, for the sake of tomography it is not necessary to prepare a maximally mixed state. Any fixed state on this register suffices. This can easily be seen in the finite-dimensional case as well, as argued in \cref{rem:any staet on V_lambda}.
\end{enumerate}
We will now introduce the relevant facts from the representation theory of $\Sp(2m)$, and then construct sample-optimal tomography protocols.

\subsection{Representation theory background}\label{sec:rep theory semisimple}

The set-up of (Gaussian) bosonic states is as in \cref{sub:boson prelims}.
The Hilbert space has a projective representation of $\Sp(2m)$, or equivalently a representation of the metaplectic group $\Mp(2m)$, the double cover of $\Sp(2m)$. These are reductive semisimple Lie groups with Lie algebra $\sp(2m)$. The group $\Sp(2m)$ has a well-defined Haar measure that is both left- and right-invariant (since it is a connected semisimple Lie group).

An irreducible unitary representation $\phi: G \to U(\cH)$ of a Lie group $G$ is in the discrete series if
\begin{align*}
    f(g) \coloneqq \bra{v} \phi(g) \ket{w}
\end{align*}
is square-integrable for any $\ket{v}, \ket{w} \in \cH$.
We then have the following version of Schur orthogonality, see Theorem 9.2 in \cite{knapp2001representation}, and Theorem 1 in \cite{harish1956representations}.

\begin{thm}[Schur~orthogonality]\label{lem:schur ortho}
    Let $\phi: G \to U(\cH_1)$ and $\psi: G \to U(\cH_2)$ be inequivalent irreducible representations in the discrete series.
    Then
    \begin{align*}
        \int  \bra{v_1} \phi(g) \ket{u_1}\overline{\bra{v_2} \psi(g) \ket{u_2}} \dd g= 0
    \end{align*}
    for all $u_i, v_i \in \cH_i$.
    If $\phi: G \to U(\cH)$ is an irreducible representation in the discrete series,
    \begin{align*}
        \int \bra{v_1} \phi(g) \ket{u_1} \overline{\bra{v_2} \phi(g) \ket{u_2}} \dd g = d_{\phi}^{-1} \braket{v_1|v_2} \braket{u_2|u_1} .
    \end{align*}
    where $d_{\phi}$ is the formal degree of the representation.
\end{thm}

We will need a finer characterization of the formal degree in our calculations. To establish it, we first recall a few concepts and facts from the theory of
reductive semisimple Lie groups, with~\cite{knapp2001representation} as the main reference.
A real Lie algebra $\mathfrak{g}$ of a Lie group $G$ can be decomposed as $\mathfrak{g}=\mathfrak{k}\oplus\mathfrak{p}$, with $\mathfrak{k}$ being the antisymmetric part and $\mathfrak{p}$ being the symmetric part. The complexification of $\mathfrak{g}$ is denoted as $\mathfrak{g}^{\C}=\mathfrak{k}^{\C}\oplus\mathfrak{p}^{\C}$. Let $\mathfrak{b}$ be an abelian Lie algebra.
For each element $\alpha$ in its dual $\mathfrak{b}'$, we can define
\[
\mathfrak{g}_{\alpha}\coloneqq \{X\in\mathfrak{g}: [A,X]=\braket{\alpha, A}X\, \, \forall A\in \mathfrak{b} \}.
\]

The vectors $\alpha\neq0$ such that $\mathfrak{g}_{\alpha}\neq 0$ form the \textit{root system} of the pair $(\mathfrak{g}^{\C},\mathfrak{b})$, denoted as $\Delta(\mathfrak{b},\mathfrak{g})$. We note that each root appears with both signs. Given a root system $\Delta$, a set of \textit{positive roots} $\Delta_+$ is a subset of $\Delta$ such that:
\begin{itemize}
    \item For each root $\alpha \in \Delta$, exactly one of the roots $\alpha$, $-\alpha$ is contained in $\Delta_+$.
    \item For any two $\alpha, \beta \in \Delta_+$ such that $\alpha + \beta$ is a root, $\alpha + \beta \in \Delta_+$.
\end{itemize}
Given a set of positive roots one can define an order on $\mathfrak{b}$:
\begin{align}\label{eq:order}
\mu\geq\lambda \Leftrightarrow\mu=\lambda+\sum_{\alpha\in\Delta_+}n_{\alpha}\alpha,  \, n_{\alpha}\geq 0.
\end{align}
The symplectic groups $\mathrm{Sp}(2m,\RR)$ have the important property that there exists a Cartan subalgebra $\mathfrak{b}$ satisfying $\mathfrak{b}\subseteq \mathfrak{k}\subseteq \mathfrak{g}$, therefore one can consider the root systems $\Delta\coloneqq\Delta(\mathfrak{b}^{\C},\mathfrak{g}^{\C})$ and $\Delta_K=\Delta(\mathfrak{b}^{\C},\mathfrak{k}^{\C})$. For complex Lie algebras and Cartan subalgegras the root spaces are one-dimensional, and since $[\mathfrak{b}^{\C}, \mathfrak{k}^{\C}]\subseteq\mathfrak{k}^{\C}$, $[\mathfrak{b}^{\C}, \mathfrak{p}^{\C}]=\mathfrak{p}^{\C}$, the root spaces are either inside $\mathfrak{k}^{\C}$ or $\mathfrak{p}^{\C}$. The corresponding roots are then said to be compact and non-compact, respectively, and the compact roots coincide with $\Delta_K$.

The following is a simplification of Theorem 16 in~\cite{harish-chandra1966II}.
\begin{prp}[Discrete series classification via the Harish-Chandra parameters]\label{prp:HC}
There exists a set of points $L'\subseteq (i\mathfrak{b})'$ with the property that, for any $\lambda\in L'$, $\braket{\lambda,\frac{2\alpha}{\braket{\alpha,\alpha}}}\neq 0$ is an integer for every $\alpha\in \Delta$, and a surjective function from $L'$ to representations of the discrete series of a semisimple Lie group $G$. Let $\phi_\lambda$ be the representation identified by $\lambda\in L'$. We say that $\lambda$ is a Harish-Chandra parameter of $\phi_\lambda$. The formal degree of $\phi_\lambda$ can be computed as
\begin{align}\label{eq:HC-formaldegree}
d_{\lambda}=c\prod_{\alpha\in\Delta_+}|\braket{\lambda,\alpha}|,
\end{align}
for $c$ depending only on the normalization of the Haar measure. It follows that $\min_{\lambda} d_{\lambda}>0$.
\end{prp}

Harish-Chandra's original characterization, while exhaustive, does not make it clear how to determine $\lambda$ when given a representation, and thus to compute the formal degree. A related characterization is given in terms of the Blattner parameter and fills this gap. It is based on the notion of minimal $K$-type of a representation $\phi$ of $G$, e.g. a highest weight of a maximal compact subgroup $K$ of $G$, in the reduction of $\phi$ into irreducible representations of $K$, such that all the other highest weights are larger according to the order~\cref{eq:order} with respect to positive compact roots.  While the correspondence classical, it is not easy to locate a complete statement in the literature. We provide such a statement as a consequence of Theorem 9.20, 12.21, 15.5 and related bibliographical notes in~\cite{knapp2001representation}.

\begin{prp}
For every Harish-Chandra parameter $\lambda$ and the choice of positive roots $\Delta_+(\lambda)=\{\alpha\in \Delta:\braket{\lambda,\alpha}>0\}$, the Blattner parameter is the minimal $K$-type of $\phi_{\lambda}$ as a representation of the maximal compact subgroup $K$, and it is equal to 
\begin{align}\label{eq:HCparameter}
\lambda=\Lambda+\rho_c-\rho_n,
\end{align}
where
\begin{align*}
        \rho_c = \frac12 \sum_{\alpha \in \Delta_{+,c}} \alpha, \quad \rho_n = \frac12 \sum_{\alpha \in \Delta_{+,n}} \alpha\,.
    \end{align*}
For any choice of compact positive roots $\Delta_{+,c}$ there is a valid Harish-Chandra parameter $\lambda$ such that $\Delta_{+,c}\subseteq\Delta_{+}(\lambda)$ and the corresponding minimal $K$-type completely determines the representation $\phi_{\lambda}$.

\end{prp}

We can now compute the formal degree of the irreducible representation of $\mathrm{Mp}(2m,\RR)$ in $L^2(\RR^{m})^{\otimes n}$ containing $n$ copies of the vacuum state.
In the KV decomposition, we denote by $\cV_{n,m}$ the component that corresponds to the trivial one-dimensional representation of $\O(n)$, and we denote by $\Pi_{n,m}$ the projection onto the corresponding isotypical component, which is an irreducible representation of $\Mp(2m)$.

We denote the formal degree of this representation by $d_{n,m}$, and by \cref{lem:schur ortho} and \cref{eq:half positive roots} a short calculation (recall that we normalize the Haar measure so $c=1$ in \cref{eq:HC-formaldegree}) gives

\begin{lem}\label{lem:formal degree n copies}
The irreducible representation $\cV_{n,m}$ has formal degree
\begin{align}\label{eq:vacuum_formaldegree}
    d_{n,m} = \left(\prod_{i<j}j-i \right)\left(\prod_{1 \leq i \leq j \leq m} (n - (i+j))\right)
\end{align}
for $n \geq 2m+1$.
\end{lem}

\begin{proof}
A maximal compact subgroup of $\mathrm{Sp}(2m,\RR)$ is $\mathrm{Sp}(2m,\RR)\cap O(2m)$, which is isomorphic to $U(m)$. The associated Cartan subalgebra has the form $\{A \, \: \, A=\mathrm{diag}(i\theta_1 \sigma_{Y},..., i\theta_m \sigma_{Y})\}$. The metaplectic group has the same Lie algebra. 
Let $\varepsilon_i\in \mathfrak{b}^{\C}{}' $ be defined from $\braket{\varepsilon_i,A}=\theta_i$, with the above notation for $A$. 
One can show that the roots are
\begin{align}\label{eq:rootsys}
\pm\varepsilon_{i}\pm\varepsilon_j \,\,, i\neq j,  \qquad \pm 2\varepsilon_k,\, k\in[m],
\end{align}

A valid choice of positive roots is
\begin{align*}
    \Sigma_+ = \{ \varepsilon_{i}\pm\varepsilon_j \; \text{ for } i < j \} \cup \{2\varepsilon_k,\; \text{ for } k\in[m]\}.
\end{align*}
The set of compact positive roots $\Sigma_{+,c}$ consists of $\eps_i - \eps_j$ for $i < j$, and the set of noncompact positive roots $\Sigma_{+,n}$ consists of $\eps_i + \eps_j$ for $i \leq j$. We do not report the computation, but it can be helpful to recall that compact root spaces are given by the generators of beam-splitters, while non-compact root spaces are given by algebra terms corresponding to $a^{\dagger}_ia^{\dagger}_j$ and $a_ia_j$, related to single-mode and two-mode squeezing.
We make the choice of positive roots $\Sigma_{+}$, we compute the corresponding Blattner parameter $\Lambda$ and we check that the Harish-Chandra parameter $\lambda=\Lambda+\rho_c-\rho_n$  satisfies $\braket{\lambda,\alpha}>0$ for every $\alpha\in \Sigma_{+}$. The minimum highest weight in the whole $L^2(\RR^{m})^{\otimes n}$ is the list of minimum eigenvalues of the single-mode energy operators $\frac{X_i^2+P_i^2}{2}=a_i^{\dagger}a_i+\frac{1}{2}$, i.e. $\Lambda_0=(\frac{n}{2},...\frac{n}{2})$, corresponding to the vacuum state $\ket{0,...,0}^{\otimes n}$, which hosts a one-dimensional representation of the maximal compact subgroup. Since this state is also invariant under $O(n)$, it is in $\cV_{n,m}$, and the Blattner parameter of $\cV_{n,m}$ is $\Lambda_0$.
For $\Sp(2m)$, with our choice of positive roots, we have
\begin{align}\label{eq:half positive roots}
\begin{split}
    \rho_c - \rho_n &=  \frac12  \sum_{i < j} \eps_i - \eps_j - \frac12 \sum_{i \leq j} \eps_i + \eps_j =
    -\sum_{i=1}^m (m+1-i)\eps_i = (-m, -m+1, \dots, -1).
\end{split}
\end{align}
By inspection, we then indeed have $\braket{\lambda,\alpha}>0$ for every $\alpha\in \Sigma_{+}$, and therefore the $\lambda$ is a valid Harish-Chandra parameter for the representation. The value of the formal degree can then be computed using~\cref{eq:HC-formaldegree}.
\end{proof}

\begin{rem}
For the special case of irreducible representations with highest weight $\Lambda$ with respect to a set of positive roots $\Delta_{+}$ (including non-compact ones), there is a simpler characterization of the formal degree, i.e.
\[
d_{\Lambda}=\prod_{\alpha \in \Delta^+}\left|\frac{\braket{\Lambda+\rho_c+\rho_{n},\alpha}}{\braket{\rho_c+\rho_{n},\alpha}}\right|\,,
\]
see Theorem 4 in \cite{harish1956representations} and its remarks. The apparent discrepancy with~\cref{eq:HC-formaldegree} and~\cref{eq:HCparameter} is solved by changing the sign of positive non-compact roots. In the case of the oscillator representation of the metaplectic group, all the irreducible representations  appearing in the KV decomposition for $n\geq 2m+1$ are of this form.. The above calculation can be equivalently done following this approach and noting that the vacuum is a highest weight when the roots $-\varepsilon_i-\varepsilon_j$ are positive. Nonetheless, the general treatment based on the Blattner parameter may be useful in other contexts.
\end{rem}

We now focus on the bosonic Hilbert space $\cH = L^2(\RR^m)$, as a representation of $\Mp(2m)$. The tensor product representation has the following decomposition.

\begin{thm}[{\cite[Theorem~(7.2), (8.2)]{kashiwara1978segal}}]\label{thm:KV}
    As a representation of $\Mp(2m) \times \O(n)$, one has the multiplicity-free decomposition
    \begin{align*}
        \cH^{\ot n} \cong L^2(\R^{m \times n}) \cong \bigoplus_{\lambda \in \Sigma_{m,n}} \cV_{\lambda} \ot \cW_{\lambda}
    \end{align*}
    where $\cV_\lambda$ and $\cW_\lambda$ are irreducible representations of $\Mp(2m)$ and $O(n)$ respectively, for some explicit index set $\Sigma_{m,n}$ described in~\cite{kashiwara1978segal}. If $n \geq 2m+1$, all $\cV_{\lambda}$ are in the discrete series.
\end{thm}

Note that the representations $\cW_\lambda$ are irreducible representations of the compact group $\O(n)$ and thereby finite dimensional.
Crucially, for $n \geq 2m+1$, all representations that appear in the KV decomposition in \cref{thm:KV} are in the discrete series, and hence the matrix coefficients of the $n$-fold tensor product action of Gaussian unitaries are square integrable.
In this case, we can define a twirling map on trace-class operators.
We write $B(\cH)$ and $\trcl(\cH)$ for bounded and trace class operators on a Hilbert space $\cH$ respectively.

\begin{lem}\label{lem:twirl symplectic}
    Let $n \geq 2m+1$. Then, there exists a map $\mathcal{T} : \trcl(\cH^{\ot n}) \to \bounded(\cH^{\ot n})$, defined by
    \begin{align*}
        \bra{\psi} \mathcal{T}[X] \ket{\phi} = \int_{\Sp(2m)} \bra{\psi} U_S^{\ot n} X U_S^{\dagger, \ot n} \ket{\phi} \dd S
    \end{align*}
    for $X \in \trcl(\cH^{\ot n})$ and for all $\ket{\psi}, \ket{\phi} \in \cH^{\ot n}$.
    It satisfies $\norm{\mathcal{T}[X]} \leq C \norm{X}_1$ where $C = (\min_{\lambda \in \Sigma_{m,n}} d_\lambda)^{-1}$. With respect to the KV decomposition it can be written as
    \begin{align}\label{eq:expression Phi}
        \mathcal{T}[X] \cong \bigoplus_\lambda d_{\lambda}^{-1} \id_{\cV_\lambda} \ot X_{\cW_\lambda}
    \end{align}
    where $X_{\cW_\lambda} = \tr_{\cV_\lambda}[\Pi_\lambda X \Pi_\lambda]$ and $\Pi_\lambda$ is the orthogonal projection onto the isotypical component $\cV_\lambda \ot \cW_\lambda$.
\end{lem}

\begin{proof}
    We need to show that the integral defining $\bra{v} \mathcal{T}[X] \ket{w}$ is well-defined, and we bound it.
    We first show that the function $f_{\psi,\xi}(S) = \abs{\bra{\psi} U_S^{\ot n} \ket{\xi}}^2$ is integrable with respect to the Haar measure for arbitrary $\ket{\psi}$ and $\ket{\xi}$.
    First consider vectors $\ket{\psi}$, $\ket{\xi}$ such that $\Pi_\lambda \ket{\psi} = 0 = \Pi_\lambda \ket{\xi}$ for all $\lambda \notin L$, for a finite set $L$.
    Choose orthonormal bases $\{\ket{i,\lambda}\}_{i=1}^{\dim \cW_\lambda}$ for the $\cW_\lambda$.
    For $\lambda \in L$, write $\Pi_\lambda \ket{\psi} = \ket{\psi_\lambda} = \sum_i \ket{\psi_{i,\lambda}} \ket{i,\lambda}$ for vectors $\ket{\psi_{i,\lambda}} \in \cV_\lambda$.
    The unitaries $U_S^{\ot n}$ act block diagonally,
    \begin{align*}
        U_S^{\ot n} \cong \bigoplus_\lambda S_\lambda \ot I_{\cW_\lambda},
    \end{align*}
    from which it follows that
    \begin{align*}
        \bra{\psi} U_S^{\ot n} \ket{\xi} = \sum_{\lambda \in L} \sum_{i=1}^{\dim \cW_\lambda} \bra{\psi_{i,\lambda}} S_\lambda \ket{\xi_{i,\lambda}}.
    \end{align*}
    We can then expand
    \begin{align*}
        \int_{\Sp(2m)} \abs{\bra{\psi} U_S^{\ot n} \ket{\xi}}^2 \dd S 
        &= \sum_{\mu,\lambda \in L} \sum_{i,j} \int_{\Sp(2m)} \bra{\psi_{i,\lambda}} \bra{i,\lambda} U_S^{\ot n} \ket{\xi_{i,\lambda}} \ket{i,\lambda} \overline{\bra{\psi_{j,\mu}} \bra{j,\mu} U_S^{\ot n} \ket{\xi_{j,\mu}} \ket{j,\mu}} \dd S \\
        &= \sum_{\mu,\lambda \in L} \sum_{i,j} \int_{\Sp(2m)} \bra{\psi_{i,\lambda}} S_\lambda \ket{\xi_{i,\lambda}}  \overline{\bra{\psi_{j,\mu}}  S_\mu \ket{\xi_{j,\mu}}}  \dd S \\
        &= \sum_{\mu,\lambda \in L} \sum_{i,j} \delta_{\mu,\lambda} d_{\lambda}^{-1} \braket{\psi_{i,\lambda} | \psi_{j,\lambda}} {\braket{\xi_{j,\mu} | \xi_{i,\mu}}} \\
        &\leq C_L \sum_{\lambda \in L} \sum_{i,j}  \braket{\psi_{i,\lambda} | \psi_{j,\lambda}} \braket{\xi_{j,\lambda} | \xi_{i,\lambda}}
    \end{align*}
    where we noted that the sums over $\mu,\lambda$ and over $i,j$ are all finite, so that we can exchange the integral and the sum, we used \cref{lem:schur ortho}\footnote{Strictly speaking, the integral should be on $\mathrm{Mp}(2m)$ rather than $\mathrm{Sp}(2m)$ to make the following representation theoretic arguments go through. In fact there is no difference apart from a multiplicative constant, since $\mathrm{Mp}(2m)$ is a double cover of $\mathrm{Sp}(2m)$, the Haar measure of $\mathrm{Mp}(2m)$ is just the pull-back of that of $\mathrm{Sp}(2m)$, and the function we are integrating depends only on the $\mathrm{Sp}(2m)$ projection since the phases cancel.}, and we let $C_L = (\min_{\lambda \in L} d_\lambda)^{-1}$.
    If we define Gram matrices $P$ and $Q$ by $P_{ij} = \braket{\psi_{i,\lambda} | \psi_{j,\lambda}}$, $Q_{ij} = \braket{\xi_{i,\lambda} | \xi_{j,\lambda}}$, then $P$ and $Q$ are positive and we can bound
    \begin{align*}
        \sum_{i,j}  \braket{\psi_{i,\lambda} | \psi_{j,\lambda}} \braket{\xi_{j,\lambda} | \xi_{i,\lambda}} = \tr[PQ] \leq \tr[P]\tr[Q] \leq \norm{\psi_\lambda}^2 \norm{\xi_\lambda}^2.
    \end{align*}
    This yields
        \begin{align}
        \int_{\Sp(2m)} \abs{\bra{\psi} U_S^{\ot n} \ket{\xi}}^2 \dd S
        &\leq C_L \sum_{\lambda \in L} \norm{\psi_\lambda}^2 \norm{\xi_{\lambda}}^2 \leq C_L \norm{\psi}^2 \norm{\xi}^2
    \end{align}
    Next, $C_L$ is uniformly bounded by $C = (\min_{\lambda \in \Sigma_{m,n}} d_{\lambda})^{-1} > 0$ where the minimum is over all weights in the KV decomposition. Here, $\Sigma_{m,n}$ is an infinite set, but for representations in the discrete series, the weights satisfy integrality conditions (see Proposition 4.14, Theorem 9.20 in \cite{knapp2001representation}) that guarantee a uniform lower bound, see \cref{prp:HC}.
    We may now consider arbitrary $\ket{\psi}$ and $\ket{\xi}$ (which do not necessarily have finite support on the $\lambda$). They can be approximated by vectors with finite support, and by Fatou's lemma and the uniform bound we derived, $g_{\psi,\xi}(S) = \abs{\bra{\psi} U_S^{\ot n} \ket{\xi}}^2$ is integrable, and
    \begin{align}\label{eq:norm f}
        \norm{f_{\psi,\xi}}_{L^1(\Sp(2m))} \leq \sqrt{C} \norm{\psi} \norm{\xi}.
    \end{align}
    To see this explicitly, consider an enumeration $\lambda_1,\lambda_2,\ldots$ of the elements of $\Sigma_{m,n}$; setting $P_N \coloneqq \sum_{k=1}^N \Pi_{\lambda_k}$, we can write
    \begin{align}
    \int_{\Sp(2m)} \abs{\bra{\psi} U_S^{\ot n} \ket{\xi}}^2 \dd S &= \int_{\Sp(2m)} \lim_{N\to\infty} \abs{\bra{\psi} P_N U_S^{\ot n} P_N \ket{\xi}}^2 \dd S \\
    &\leq \liminf_{N\to\infty} \int_{\Sp(2m)} \abs{\bra{\psi} P_N U_S^{\ot n} P_N \ket{\xi}}^2 \dd S \\
    &\leq C \|\psi\|^2\|\xi\|^2 ,
    \end{align}
    where the first inequality holds by Fatou's lemma, and the second by the above calculation, due to the fact that $P_N\ket{\psi}$ and $P_N\ket{\xi}$ are supported on finitely many subspaces $\cV_\lambda \otimes \cW_\lambda$.

    Now, consider a trace class operator $X \in \cT(\cH^{\ot n})$, and take a singular value decomposition
    \begin{align*}
        X = \sum_i \alpha_i \ket{e_i}\bra{f_i}, \qquad \sum_i \abs{\alpha_i} = \norm{X}_1
    \end{align*}
    with orthonormal bases $\{\ket{e_i}\}$ and $\{\ket{f_i}\}$.
    Then, our goal is to show that $g_{\psi,\phi,X}(S) = \bra{\psi} U_S^{\ot n} X U_S^{\dagger, \ot n} \ket{\phi}$ is integrable.
    \begin{align*}
        \int_{\Sp(2m)} \abs{\bra{\psi} U_S^{\ot n} X U_S^{\dagger, \ot n} \ket{\phi}} \dd S &\leq \sum_i \abs{\alpha_i} \int_{\Sp(2m)} \abs{\bra{\psi} U_S^{\ot n} \ket{e_i} \bra{f_i} U_S^{\dagger, \ot n} \ket{\phi}} \dd S \\
        &\leq \sum_i \abs{\alpha_i} \norm{f_{\psi,e_i}}_{L^2(\Sp(2m))} \norm{f_{\xi,f_i}}_{L^2(\Sp(2m))}\\
        &\leq C \norm{X}_1
    \end{align*}
    using the Cauchy-Schwarz inequality and \cref{eq:norm f}.
    This shows that $\mathcal{T}[X]$ is well-defined.
    Next, consider vectors
    \begin{align*}
        \ket{e_i} = \ket{v_i} \ket{w_i} \qquad \text{ for } \ket{v_i} \in \cV_{\lambda_i}, \ket{w_i} \in \cW_{\lambda_i}
    \end{align*}
    for $i = 1,2,3,4$ in the $\lambda_i$-isotypical component of the KV decomposition.
    Then for the rank-one operator $X = \ket{e_1}\bra{e_2}$ we have
    \begin{align*}
        \bra{e_3} \mathcal{T}[X] \ket{e_4} = \begin{cases}
            d_{\lambda}^{-1} \braket{v_2 | v_1} \braket{v_3 | v_4} \braket{w_3 | w_1} \braket{w_2 | w_4} & \text{ if } \lambda_1 = \lambda_2 = \lambda_3 = \lambda_4 = \lambda\\
            0 & \text{ otherwise.}
        \end{cases}
    \end{align*}
    This shows that \cref{eq:expression Phi} is valid for $X$ of this form, and hence for arbitrary $X$ by linearity.
\end{proof}

\subsection{Tomography of pure Gaussian states}
We now show that there exists a tomography protocol for pure bosonic Gaussian states proceeding exactly analogously to the fermionic case. The only additional subtlety is the appearance of integrals over non-compact sets, but this can be dealt with by the square-integrability of matrix coefficients for sufficiently large $n$.

From the Haar measure on the symplectic group, we also get a unique (up to constant scaling) measure on Gaussian mean-zero states.
It now follows from \cref{lem:schur ortho} that for $n \geq 2m+1$,
\begin{align}\label{eq:projection random gaussian}
    d_{n,m}\int_{\cG_m^0} \proj{\phi}^{\ot n} \dd \phi = d_{n,m}\int_{\Sp(2m)} U_S^{\ot n} \proj{0}^{\ot n} U_S^{\dagger, \ot n} \dd S = \Pi_{n,m}.
\end{align}
This implies that $\mu_{n,m}$ defined by
\begin{align}\label{eq:haar measure projection-bos}
    \dd \mu_{n,m}(\phi) = d_{n,m} \,\dd \phi \, \proj{\phi}^{\ot n},
\end{align}
where $d_{n,m}$ is given in \cref{lem:formal degree n copies}, is a POVM on~$\cV_{n,m}$ with outcomes in~$\cG_m^0$.

We now proceed exactly as in the fermionic case, with the only difference that $\mu_{n,m}$ and $d_{n,m}$ are only well-defined for $n \geq 2m+1$.
For a fixed $\psi$, we now consider the fidelity $\abs{\braket{\psi | \hat{\psi}}}$ as a random variable, and compute moments of this random variable.

\begin{lem}\label{lem:moments pure boson measurement}
Let $\psi \in \cG_m^0$ be a pure bosonic Gaussian state with mean zero.
Let $\hat\psi$ be the outcome resulting from the POVM~$\mu_{n,m}$ on input $\psi^{\ot n}$.
Then for $n \geq 2m + 1$ and integer $k \geq 2m + 1$, we can evaluate
\begin{align*}
    \Ex \abs{\braket{\psi | \hat{\psi}}}^{2(k - n)} = \frac{d_{n,m}}{d_{k,m}} = \prod_{1 \leq i \leq j \leq m} \frac{n - (i+j)}{k - (i+j)}.
\end{align*}
\end{lem}

\begin{proof}
The proof is exactly the same as that of \cref{lem:moments pure fermion measurement}.
    We compute
    \begin{align*}
        \Ex \abs{\braket{\psi | \hat{\psi}}}^{2(k - n)}
        &= \int_{\phi \in \cG_m^0} \tr[ \dd \mu_{n,m}(\phi) \proj{\psi}^{\ot n} ] \ \abs{\braket{\psi|\phi}}^{2(k - n)} \\
        &= d_{n,m} \int_{\phi \in \cG_m^0} \dd\phi \, \tr[ \proj\phi^{\ot n} \proj{\psi}^{\ot n} ] \ \abs{\braket{\psi|\phi}}^{2(k - n)} \\
        &= d_{n,m} \int_{\phi \in \cG_m^0} \dd\phi \, \abs{\braket{\psi|\phi}}^{2k} \\
        &= \frac {d_{n,m}} {d_{k,m}} \int_{\phi \in \cG_m^0} \tr[ \dd \mu_{k,m}(\phi) \proj{\psi}^{\ot k} ] \\
        &= \frac {d_{n,m}} {d_{k,m}} \bra\psi^{\ot k} \Pi_{k,m} \ket\psi^{\ot k}  = \frac {d_{n,m}} {d_{k,m}}.
    \end{align*}
    This ratio of dimensions can be given an explicit expression by \cref{lem:formal degree n copies} as
\begin{align*}
    \frac{d_{n,m}}{d_{k,m}} &= \prod_{1 \leq i \leq j \leq m} \frac{n - (i+j)}{k - (i+j)}. \qedhere
\end{align*}
\end{proof}
Let us first bound the expectation value of the squared fidelity $F\coloneqq \abs{\braket{\psi|\hat\psi}}$. Assume \(n\geq 4m\).  Taking \(k=n+1\) in
\cref{lem:moments pure boson measurement}, we obtain
\begin{align*}
    \Ex F^2
    &=
    \frac{d_{n,m}}{d_{n+1,m}}                                                     \\
    &=
    \prod_{1 \leq i \leq j \leq m}
    \frac{n-(i+j)}{n+1-(i+j)}                                                      \\
    &=
    \prod_{1 \leq i \leq j \leq m}
    \left(
        1-\frac{1}{n+1-(i+j)}
    \right).
\end{align*}
Since \(i+j\leq 2m\) and \(n\geq 4m\), we have
\(n+1-(i+j)\geq n/2\).  Hence
\[
    \Ex F^2
    \geq
    \left(1-\frac{2}{n}\right)^{m(m+1)/2}.
\]
Therefore,
\[
    \Ex(1-F^2)
    \leq
    1-
    \left(1-\frac{2}{n}\right)^{m(m+1)/2}
    \leq
    \frac{m(m+1)}{n}.
\]
Applying Markov's inequality to the nonnegative random variable \(1-F\), we get
\[
    \Pr\mleft(F^2\leq 1-\eps^2\mright)
    =
    \Pr\mleft(1-F^2\geq \eps^2\mright)
    \leq
    \frac{m(m+1)}{n\eps^2}.
\]
Thus \(n=\bigO(m^2/(\eps^2\delta))\) copies already suffice to obtain an
\(\eps\)-accurate estimate with probability at least \(1-\delta\).  As in the
fermionic case, the optimal dependence on \(\delta\) follows by using negative
moments of the fidelity.
\begin{cor}[Tomography of pure bosonic Gaussian states]\label{cor:pure boson tomography}
    Let \(0<\eps,\delta<1\).  Given
    \begin{align*}
        n
        =
        \bigO\mleft(
            \frac{m^2+\log(\delta^{-1})}{\eps^2}
        \mright)
    \end{align*}
    copies of a pure bosonic Gaussian state~\(\psi\in\cG_m^0\), the measurement
    \(\mu_{n,m}\) outputs an estimate~\(\hat\psi\) such that
    \[
        \abs{\braket{\psi|\hat\psi}}^2
        \geq
        1-\eps^2
    \]
    with probability at least \(1-\delta\).
\end{cor}

\begin{proof}
    Set \(F^2\coloneqq \abs{\braket{\psi|\hat\psi}}^2\).  We assume
    \(n\geq 4m\), which is absorbed in the claimed sample complexity.
    Define
    \[
        s\coloneqq \floor*{\frac n4},
        \qquad
        r\coloneqq n-s .
    \]
    Then \(r\geq 2m+1\), so \cref{lem:moments pure boson measurement} applies
    with \(k=r\), and gives
    \[
        \Ex F^{-2s}
        =
        \Ex F^{2(r-n)}
        =
        \frac{d_{n,m}}{d_{r,m}}.
    \]
    We now bound this ratio using the explicit formula for \(d_{n,m}\):
    \begin{align*}
        \frac{d_{n,m}}{d_{r,m}}
        &=
        \prod_{1\leq i\leq j\leq m}
        \frac{n-(i+j)}{r-(i+j)}                                      \\
        &=
        \prod_{1\leq i\leq j\leq m}
        \left(
            1+
            \frac{s}{r-(i+j)}
        \right).
    \end{align*}
    Since \(i+j\leq 2m\) and \(n\geq 4m\), we have
    \[
        r-(i+j)
        \geq
        r-2m
        \geq
        \frac{3n}{4}-2m
        \geq
        \frac n4
        \geq
        s .
    \]
    Hence every factor in the product is at most \(2\).  Since there are
    \(m(m+1)/2\) factors,
    \[
        \frac{d_{n,m}}{d_{r,m}}
        \leq
        2^{m(m+1)/2}.
    \]
    Markov's inequality applied to \(F^{-2s}\) now yields
    \begin{align*}
        \Pr\mleft(F^2\leq 1-\eps^2\mright)
        &=
        \Pr\mleft(
            F^{-2s}
            \geq
            (1-\eps^2)^{-s}
        \mright)                                                        \\
        &\leq
        (1-\eps^2)^s \Ex F^{-2s}                                           \\
        &=
        (1-\eps^2)^s
        \frac{d_{n,m}}{d_{r,m}}                                         \\
        &\leq
        (1-\eps^2)^{\floor*{n/4}}
        2^{m(m+1)/2}                                                     \\
        &\leq
        \exp\mleft(
            -\eps^2\floor*{\frac n4}
            +
            \frac{m(m+1)}2\log 2
        \mright).
    \end{align*}
    Thus it suffices to choose $n
        =
        \bigO\mleft(
            \frac{m^2+\log(\delta^{-1})}{\eps^2}
        \mright)$, to make the failure probability at most \(\delta\).  This is as claimed.
\end{proof}

This sample complexity is optimal, as stated in \cref{thm:tomo boson lower bound}.

\subsection{Tomography of mixed Gaussian states and random purification}

We consider the natural identification $\cH^{\ot 2} \cong L^2(\RR^{2m})$, so $\Sp(4m)$ acts on $\cH^{\ot 2}$.
We also consider $H = \Sp(2m)$ as a subgroup of $G = \Sp(4m)$, acting only on the last $m$ modes, so on $\cH^{\ot 2}$ these act as $I \ot U_T$ for $T \in \Sp(2m)$.
Then there exists a unique measure $\mu$ that is invariant under $G$ on $G/H$ such that we can write
\begin{align}\label{eq:integral quotient}
     \int_G f(S) \dd S = \int_{G/H} \left( \int_H f(ST) \dd T \right) \dd \mu(SH)\,.
\end{align}
for any function $f \in L^1(G)$.
See e.g.\ Theorem 1.5.2 in \cite{deitmar2009principles}, and note that both $G$ and $H$ are unimodular, and $H$ is a closed subgroup of $G$.
We also need that, if the transpose is an antiautomorphism (up to irrelevant phases) for a unitary representation of an unimodular group $G$, then
\begin{align}\label{eq:transposetwirl}
\int_G U_gXU_g^{\dagger}\,\dd g =\int_G U_g^{\intercal}X\overline{U_g}\,\dd g,
\end{align}
as one can see from a change of variables and using that the measure induced by this change on the group is left and right invariant, and therefore equal to the Haar measure.

Analogous to the setting of \cref{thm:main technical}, we consider the KV decomposition for $\cH^{\ot n}$ as in \cref{thm:KV}, choose a product basis in that decomposition, and use the complex conjugate to decompose $\tilde \cH^{\ot n}$.
If $\Upsilon$ is the unitary implementing the decomposition, we (formally) have
\begin{align*}
    (\Upsilon \ot \bar\Upsilon) \ket{\Gamma}^{\ot n} = \bigoplus_\lambda \ket{\Gamma}_{\cV_\lambda \tilde \cV_\lambda} \ket{\Gamma}_{\cW_\lambda \tilde \cW_\lambda}.
\end{align*}
We will need the following lemma:
\begin{lem}\label{lem:gaussian-sym}
Let $\rho$ be a Gaussian state with zero mean, and let $\ket{\psi^\text{std}_\rho}_{AB}$ be its standard purification. Let $\Upsilon$, $\bar{\Upsilon}$ (as in \cref{eq:transpose}) be the unitary operators implementing KV decompositions for the representation of $\mathcal{H}^{\otimes n}$ and its transpose on $\tilde \cH^{\ot n}$.
Then we have
\begin{align}\label{eq:diagonal}
    \Upsilon\rho^{\ot n}\Upsilon^{\dagger} = \bigoplus_\lambda \rho_{\cV_\lambda} \ot \frac{I_{\cW_{\lambda}}}{\dim \cW_\lambda}\,,
\end{align}
and
\begin{align*}
(\Upsilon\otimes\bar{\Upsilon})\ket{\psi^{\text{std}}_{\rho}}^{\ot n} = \bigoplus_\lambda \ket{\psi_\lambda}_{\cV_\lambda \tilde \cV_\lambda} \ot \frac{\ket{\Gamma}_{\cW_{\lambda}\tilde\cW_\lambda}}{\sqrt{\dim \cW_\lambda}}\,,
\end{align*}
\end{lem}
\begin{proof}
If $\rho \in \S(\cH)$ is a Gaussian state with zero mean, $\rho^{\ot n}$ commutes with the action of $\O(n)$.
To see this, note that $\sigma = U_O \rho^{\ot n} U_O^\dagger$ is again a Gaussian state with zero mean, so it suffices to see it has the same covariance matrix. Indeed, if $\rho$ has covariance matrix $V$, then $\rho^{\ot n}$ has covariance matrix $V \ot I_n$, and $\sigma$ has covariance matrix $(I_{2m}\otimes O^{\intercal})(V\otimes I_n)(I_{2m}\otimes O)=V\otimes I_n$.
By Schur's lemma applied to the (finite-dimensional)  eigenspaces in the support of $\rho^{\otimes n}$, this implies that
\begin{align}\label{eq:diagonalrho}
    \rho^{\ot n} \cong \bigoplus_\lambda \rho_{\lambda} \ot \frac{I_{\cW_\lambda}}{\dim \cW_\lambda}.
\end{align}
Analogously to \cref{eq:max ent decomposition}, this yields
\begin{align}
(\Upsilon\otimes\bar{\Upsilon})\ket{\psi^{\text{std}}_{\rho}}^{\ot n} &=(\Upsilon\otimes\bar{\Upsilon})(\sqrt{\rho^{\otimes n}}\otimes I)\ket{\Gamma}\\
&=\bigoplus_\lambda (\sqrt{\rho_{\cV_\lambda}} \otimes I_{\cV_{\lambda}})\ket{\Gamma}_{\cV_{\lambda}\tilde{\cV}_{\lambda}}\otimes\frac{\ket{\Gamma}_{\cW_{\lambda}\tilde{\cW}_{\lambda}}}{\sqrt{\dim \cW_\lambda}}\,,
\end{align}
proving the statement.
\end{proof}

We consider the following quantum channel. Choose an arbitrary fixed state $\ket{v_\lambda}$ on each $\tilde \cV_\lambda$ (a mixed state would also work).
\begin{align}\label{eq:prepurification bosons}
    \cQ_{n,m}^{\text{Bose}}[X] = \Upsilon^{\dagger}\otimes \Upsilon^{\intercal}\left(\bigoplus_\lambda X_\lambda \ot \proj{v_\lambda} \ot \frac{\proj{\Gamma}_{\cW_\lambda \tilde \cW_\lambda}}{\dim \cW_{\lambda}}\right)\Upsilon\otimes \bar{\Upsilon}, \qquad X_{\lambda} = \tr_{\cW_\lambda}[\Pi_\lambda \Upsilon X\Upsilon^{\dagger} \Pi_\lambda].
\end{align}
Note that this has the same structure as a random purification channel, except that on $\tilde \cV_\lambda$ we have the state $\ket{v_\lambda}$ instead of a maximally mixed state (which would not be well-defined, since the $\cV_\lambda$ are infinite-dimensional).
While this channel \emph{does not have the interpretation of a random purification channel}, it will suffice for the sake of tomography.

\begin{proof}[Proof of \cref{thm:tomo boson}]
We begin with the case in which the unknown Gaussian state has zero first moment. The tomography algorithm proceeds as follows:
\begin{enumerate}[noitemsep]
\item First apply the channel $\cQ_{n,m}^{\text{Bose}}$ of \cref{eq:prepurification bosons} to $\sigma^{\ot n}$
\item Then perform the pure-state tomography protocol of \cref{cor:pure boson tomography} on the output, to obtain an estimate $\hat{\psi}$ of a $2m$-mode pure bosonic Gaussian state.
\item Return $\hat\sigma$, which is obtained by taking the partial trace of $\hat\psi$ over the purifying system.
\end{enumerate}

We now analyze the result of the last two steps, applied to an arbitrary state $\rho$ on $n$ copies of $\cH^{\ot 2}$. This yields a probability distribution $\dd P(\hat\sigma \vert \rho)$.
We claim that this distribution only depends on $(I \ot \mathcal{T})[\rho]$, where we apply the map $\mathcal{T}$ from \cref{lem:twirl symplectic} to the second system.
That is,
\begin{align*}
    \dd P(\hat\sigma \vert \rho) = \dd P(\hat\sigma \vert \rho') \quad \text{ if }\quad  (I \ot \mathcal{T})[\rho] = (I \ot \mathcal{T})[\rho'].
\end{align*}
We now prove this claim.
Denote by $\hat \sigma(U_S)$ the reduced state of $U_S \ket{0}$.
We denote by $G = \Sp(4m)$ the symplectic group on $2m$ modes, acting on $\cH^{\ot 2}$, and by $H$ the subgroup that only acts on the last $m$ modes, so $H \cong \Sp(2m)$ and acts as $I \ot U_S$ on $\cH^{\ot 2}$.
Then for an integrable function $f$ we have
\begin{align*}
    &\int f(\hat\sigma) \dd P(\hat\sigma \vert \rho) \\
    &\quad = d_{n,2m} \int_G \tr[U_S^{\ot n} \proj{0} U_S^{\dagger,\ot n} \rho] f(\hat \sigma(U_S)) \dd S \\
    &\quad = d_{n,2m} \int_{G/H} \left( \int_H \tr[((I \ot U_T) U_S)^{\ot n} \proj{0} (U_S^\dagger (I \ot U_T^\dagger))^{\ot n} \rho] f(\hat \sigma((I \ot U_T)U_S)) \dd T \right) \dd \mu(SH) \\
    &\quad = d_{n,2m} \int_{G/H} \left( \int_H \tr[U_S^{\ot n} \proj{0} U_S^{\dagger,\ot n} (I \ot U_T^\dagger)^{\ot n} \rho  (I \ot U_T)^{\ot n}] f(\hat \sigma(U_S)) \dd T \right) \dd \mu(SH) \\
    &\quad = d_{n,2m} \int_{G/H} \tr[U_S^{\ot n} \proj{0} U_S^{\dagger,\ot n} (I \ot \mathcal{T})[\rho]] f(\hat \sigma(U_S)) \dd \mu(SH),
\end{align*}
where we use \cref{eq:integral quotient}, the fact that $\hat \sigma((I \ot U_T)U_S) = \hat \sigma(U_S)$, and the definition of $\mathcal{T}$. Note that the integral over $(I \ot U_T^\dagger)^{\ot n} \rho  (I \ot U_T)^{\ot n}$ is evaluated against a pure state, so the integral is well-defined by \cref{lem:twirl symplectic} and equals $(I \ot \mathcal{T})[\rho]$. This proves the claim.

Now, consider the state
\begin{align*}
    \sigma^{\ot n} \cong \bigoplus_\lambda \rho_{\cV_\lambda} \ot \frac{I_{\cW_{\lambda}}}{\dim \cW_\lambda},
\end{align*}
and consider the standard purification $\ket{\psi^{\text{std}}}$ of $\sigma$ such that
\begin{align*}
    \ket{\psi^{\text{std}}}^{\ot n} \cong \bigoplus_\lambda \ket{\psi}_{\cV_\lambda \tilde \cV_\lambda} \ot \frac{\ket{\Gamma}_{\cW_{\lambda}\tilde\cW_\lambda}}{\sqrt{\dim \cW_\lambda}},
\end{align*}
where the $\ket{\psi}_{\cV_\lambda \tilde \cV_\lambda}$ purify $\rho_{\cV_\lambda}$, via~\cref{lem:gaussian-sym}.
Since the symplectic unitaries are closed under conjugate, and thus transpose (~\cref{lem:gaussian_purification_bos}), and transposition is an antiautomorphism of the unitary channels, we can apply~\cref{eq:transposetwirl}.
Then from \cref{lem:twirl symplectic} it follows directly that
\begin{align*}
    (I \ot \mathcal{T})[\cQ_{n,m}^{\text{Bose}}[\sigma^{\ot n}]] = (I \ot \mathcal{T})[\proj{\psi^{\text{std}}}^{\ot n}].
\end{align*}
This implies that the following procedure yields the same outcome distribution as our proposed tomography scheme:
\begin{enumerate}[noitemsep]
\item Assume we are given $n$ copies of a fixed purification $\ket{\psi}$ of $\sigma$.
\item Perform the pure-state tomography protocol of \cref{cor:pure boson tomography} on these purifications, to obtain an estimate $\hat{\psi}$ of a $2m$-mode pure bosonic Gaussian state.
\item Return $\hat\sigma$, which is obtained by tracing $\hat\psi$ over the purifying subsystem.
\end{enumerate}
This procedure directly inherits the sample complexity scaling from \cref{cor:pure boson tomography}, using that the standard purification is Gaussian by \cref{lem:gaussian_purification_bos}, and the fact that $F(\sigma,\hat \sigma)^2 \geq \abs{\braket{\psi | \hat \psi}}^2$.

This proves the statement in the case where the unknown Gaussian state is promised to have zero first moment. In the next paragraph, we show how to remove this assumption, thereby proving the theorem for arbitrary mixed Gaussian states.
\end{proof}

\subsection{Tomography of mixed Gaussian states with non-zero first moments}
\label{sec:bosonic-nonzero-first-moments}
We now explain how to remove the zero-mean assumption from the mixed Gaussian
tomography protocol developed above. The key idea is that first moments can be
separated from the covariance matrix by a physical operation, namely a balanced
beam splitter: starting from two identical copies, one output arm has the same
covariance matrix and zero first moment, while the other has the same covariance
matrix and a first moment amplified by a factor \(\sqrt2\). After this reduction,
the zero-mean tomography procedure developed above is used to estimate the
covariance matrix on one arm, while the first moment is estimated separately on
the other arm by a generaldyne measurement with the estimated covariance as
seed.

In this subsection we write \(\rho(\mu,V)\) for the \(m\)-mode bosonic Gaussian
state with first moment \(\mu\in\R^{2m}\) and covariance matrix \(V\).

\begin{lem}[Balanced beam splitter reduction]
\label{lem:balanced-bs-first-moments}
Let \(\rho(\mu,V)\) be an \(m\)-mode bosonic Gaussian state.  Let \(B\) be the
Gaussian unitary corresponding to the balanced beam splitter acting on two
\(m\)-mode systems, i.e. to the symplectic matrix
\[
    S_{\mathrm{bs}}
    =
    \frac1{\sqrt2}
    \begin{pmatrix}
        I & -I \\
        I & I
    \end{pmatrix},
\]
where \(I\) is the \(2m\times 2m\) identity matrix.  Then
\[
    B\bigl(\rho(\mu,V)\ot \rho(\mu,V)\bigr)B^\dagger
    =
    \rho(0,V)\ot \rho(\sqrt2\mu,V).
\]
\end{lem}

\begin{proof}
The input state is Gaussian with first moment \((\mu,\mu)\) and covariance
matrix $\begin{pmatrix}
        V & 0 \\
        0 & V
    \end{pmatrix}$. Since the beam splitter is a Gaussian unitary, the output is also a Gaussian
state. The output first moment is
\[
    S_{\mathrm{bs}}
    \binom{\mu}{\mu}
    =
    \binom{0}{\sqrt2\mu},
\]
while the output covariance matrix is
\[
    S_{\mathrm{bs}}
    \begin{pmatrix}
        V & 0 \\
        0 & V
    \end{pmatrix}
    S_{\mathrm{bs}}^{\mathsf T}
    =
    \begin{pmatrix}
        V & 0 \\
        0 & V
    \end{pmatrix}.
\]
The output is therefore a Gaussian state with block-diagonal covariance matrix,
and hence it factorizes as claimed.
\end{proof}

We will use the following standard form of generaldyne measurements; see, e.g.,
\cite{BUCCO}.

\begin{lem}[Generaldyne measurement]
\label{lem:generaldyne-measurement}
Let \(W\) be a valid \(2m\times 2m\) bosonic covariance matrix.  There exists a
POVM on \(\R^{2m}\), called the \emph{generaldyne measurement with seed \(W\)},
with the following property.  If the input is an \(m\)-mode Gaussian state
\(\rho(\mu,V)\), then the measurement outcome \(X\in\R^{2m}\) is a classical
Gaussian random variable with
\[
    \Ex X=\mu,
    \qquad
    \operatorname{Cov}(X)=\frac{V+W}{2}.
\]
Equivalently,
\[
    X\sim \mathcal N\!\left(\mu,\frac{V+W}{2}\right).
\]
\end{lem}
As a paradigmatic example of a generaldyne measurement, the usual heterodyne
measurement is the special case in which the seed is the vacuum covariance
matrix, \(W=I\).

We also record a slightly stronger consequence of the zero-first-moment tomography protocol from the previous section. Besides returning a Gaussian state with high fidelity to the unknown Gaussian state, the protocol also yields a multiplicative estimate of its covariance matrix.

\begin{lem}[Covariance estimation in multiplicative error]
\label{lem:bosonic-covariance-multiplicative}
Let \(0<\varepsilon,\delta<1\) and $m\in\mathbb{N}^+$. Let \(\rho=\rho(0,V)\) be an unknown \(m\)-mode bosonic Gaussian state with zero first moment.
There is a procedure (i.e.~the centered tomography algorithm outlined above) that, using
\[
    M
    =
    \mathcal O\!\left(
        \frac{m^2+\log(\delta^{-1})}{\varepsilon^2}
    \right)
\]
copies of \(\rho\), outputs a covariance matrix \(\widehat V\) such that, with
probability at least \(1-\delta\),
\[
    c_\varepsilon^{-1}V
    \preceq
    \widehat V
    \preceq
    c_\varepsilon V \,,
\]
where \[
    c_\varepsilon
    \coloneqq
    \frac{1+\sqrt{\varepsilon^2(2-\varepsilon^2)}}
         {1-\sqrt{\varepsilon^2(2-\varepsilon^2)}}
    =
    1+\mathcal O({\varepsilon}) .
\]
Moreover, if \(\widehat\rho=\rho(0,\widehat V)\), then $F(\rho,\widehat\rho)^2 \geq 1-\varepsilon^2$ with probability at least \(1-\delta\).
\end{lem}

\begin{proof}
Run the centered mixed-state tomography construction from the previous
subsection with accuracy parameter \(\varepsilon\) and failure probability \(\delta\). Although the
protocol itself does not prepare a purification of \(\rho\), its output
distribution coincides with the following auxiliary experiment: we have \(M\) copies of a Gaussian
purification \(\psi\) of \(\rho\); we perform the pure-state tomography protocol of
\cref{cor:pure boson tomography} on \(M\) copies of \(\psi\); we obtain a pure
Gaussian estimate \(\widehat\psi\), and then trace out the purifying system.

Thus, in order to analyze the covariance matrix returned by the actual centered
protocol, it suffices to analyze this auxiliary experiment.  Let \(W\) be the
covariance matrix of the fixed purification \(\psi\), and let \(W'\) be the
covariance matrix of the pure Gaussian estimate \(\widehat\psi\).  By
\cref{cor:pure boson tomography}, choosing
\[
    M
    =
    \mathcal{O}\!\left(
        \frac{m^2+\log(\delta^{-1})}{\varepsilon^2}
    \right)
\]
ensures that $F(\psi,\widehat\psi)^2
    \geq
    1-\varepsilon^2$ with probability at least \(1-\delta\). The fidelity formula for pure Gaussian states with zero first moments~\cite{BUCCO} gives
\[
    F(\psi,\widehat\psi)^2
    =
    \frac{1}{\sqrt{\det\left(\frac{W+W'}{2}\right)}} .
\]
Hence, on the above probability event,
\begin{equation}\label{eq_1a}
    \frac{1}{\sqrt{\det\left(\frac{W+W'}{2}\right)}}
    \geq
    1-\varepsilon^2 .
\end{equation}
Set
\[
    A \coloneqq W^{-1/2}W'W^{-1/2}.
\]
This is well defined because \(W\) is the covariance matrix of a pure Gaussian
state and is therefore real symmetric and positive definite.  We will use that
\(A\) is real symmetric positive definite and that its eigenvalues come in
reciprocal pairs.  The first property is immediate from the definition.  We now
justify the second one.

Since \(W\) is the covariance matrix of a pure Gaussian state, there exists a
symplectic matrix \(S\) such that $W=SS^{\mathsf T}$. By the Euler, or Bloch--Messiah, decomposition~\cite{BUCCO}, we may write $S=O_1ZO_2$, where \(O_1,O_2\) are orthogonal symplectic matrices and
\[
    Z=\operatorname{diag}(z_1,\ldots,z_m,z_1^{-1},\ldots,z_m^{-1}),
    \qquad z_1,\ldots,z_m>0 .
\]
Since \(O_2\) is orthogonal, this gives
\[
    W=O_1Z^2O_1^{\mathsf T},
    \qquad
    W^{-1/2}=O_1Z^{-1}O_1^{\mathsf T}.
\]
Thus \(W^{-1/2}\) is symplectic.  Moreover, \(W'\) is symplectic as well, since
it is also the covariance matrix of a pure Gaussian state.  Hence $A=W^{-1/2}W'W^{-1/2}$ is symplectic.  Since \(A\) is also real symmetric and positive definite, its eigenvalues are positive and come in reciprocal pairs.

Moreover, since \(W\) is symplectic, \(\det W=1\), and therefore
\[
    \det\left(\frac{W+W'}{2}\right)
    =
    \det\left(\frac{I+A}{2}\right).
\]
Let \(\lambda=\|A\|_\infty\).  Since the eigenvalues of \(A\) are positive and
come in reciprocal pairs, both \(\lambda\) and \(\lambda^{-1}\) are eigenvalues
of \(A\).  Every other reciprocal pair \((t,t^{-1})\) of eigenvalues of \(A\) contributes $\frac{(1+t)(1+t^{-1})}{4}\geq 1$ to \(\det((I+A)/2)\).  Therefore
\begin{equation}\label{eq_1b}
    \det\left(\frac{I+A}{2}\right)
    \geq
    \frac14(1+\lambda)(1+\lambda^{-1}) .
\end{equation}
Combining \eqref{eq_1a} and \eqref{eq_1b}, we obtain
\[
    \frac14(1+\lambda)(1+\lambda^{-1})
    \leq
    (1-\varepsilon)^{-2}.
\]
Solving this inequality gives
\[
    \lambda
    =
    \|A\|_\infty
    \leq
    \frac{1+\sqrt{\varepsilon^2(2-\varepsilon^2)}}
         {1-\sqrt{\varepsilon^2(2-\varepsilon^2)}}
    =
    c_\varepsilon .
\]
Since \(A\) is real symmetric and positive definite, the bound
\(\|A\|_\infty\leq c_\varepsilon\) is equivalent to \(A\preceq c_\varepsilon I\). Moreover, because the eigenvalues of \(A\) come in reciprocal pairs, the same
bound also implies the lower spectral bound \(c_\varepsilon^{-1}I
    \preceq
    A\). Hence
\[
    c_\varepsilon^{-1}I
    \preceq
    A
    \preceq
    c_\varepsilon I .
\]
Equivalently, by substituting \(A=W^{-1/2}W'W^{-1/2}\), we obtain
\[
    c_\varepsilon^{-1}W
    \preceq
    W'
    \preceq
    c_\varepsilon W .
\]
The covariance matrix \(\widehat V\) returned by the centered protocol has the
same distribution as the covariance matrix obtained by restricting \(W'\) to
the original \(m\)-mode subsystem in the auxiliary experiment.  Likewise, \(V\)
is the restriction of \(W\) to that subsystem.  Since taking principal
submatrices preserves the Loewner order, we obtain
\[
    c_\varepsilon^{-1}V
    \preceq
    \widehat V
    \preceq
    c_\varepsilon V
\]
with probability at least \(1-\delta\).

Finally, the reduced Gaussian state produced by the
actual protocol has the same distribution as \(\rho(0,\widehat V)\), which is the reduction of
\(\widehat\psi\) in the auxiliary experiment.  By monotonicity of fidelity
under the partial trace,
\[
    F\bigl(\rho(0,V),\rho(0,\widehat V)\bigr)^2
    \geq
    F(\psi,\widehat\psi)^2
    \geq
    1-\varepsilon^2
\]
with probability at least \(1-\delta\).  This concludes the proof.
\end{proof}

We now combine all the above preliminary results to conclude the proof of \cref{thm:tomo boson} in the case of mixed Gaussian states with non-zero first moments.
\begin{thm}[Tomography of mixed bosonic Gaussian states with non-zero first moments]
\label{thm:bosonic-mixed-nonzero-first-moments}
Let \(\rho=\rho(\mu,V)\) be an arbitrary \(m\)-mode bosonic Gaussian state.
For every \(0<\eps,\delta<1\), there is a tomography protocol which uses
\[
    \mathcal{O}\!\left(
        \frac{m^2+\log(\delta^{-1})}{\eps^2}
    \right)
\]
copies of \(\rho\) and outputs a Gaussian state $\widehat\rho=\rho(\widehat\mu,\widehat V)$ such that, with probability at least \(1-\delta\),
\[
    F(\rho,\widehat\rho)^2
    \geq
    1-\eps^2 .
\]
Moreover, on the same event, \(\widehat\mu\) and \(\widehat V\) satisfy
\[
    (1-f_\eps)V
    \preceq
    \widehat V
    \preceq
    (1+f_\eps)V,
\]
where \(f_\eps=\mathcal{O}(\eps)\), and
\[
    (\widehat\mu-\mu)^{\mathsf T}
    V^{-1}
    (\widehat\mu-\mu)=
    \mathcal O\!\left(
        \eps^2\,
        \frac{m+\log(\delta^{-1})}{m^2+\log(\delta^{-1})}
    \right)\,.
\]
\end{thm}
\begin{proof}
Choose  $2M=\mathcal O\!\left(\frac{m^2+\log(\delta^{-1})}{\eps^2}\right)$ input copies. We use them in pairs: applying the balanced beam splitter from
\cref{lem:balanced-bs-first-moments} to each pair transforms the \(2M\) copies
of \(\rho(\mu,V)\) into \(M\) copies of
\(\sigma\coloneqq \rho(0,V)\) and \(M\) copies of
\(\tau\coloneqq \rho(\sqrt2\mu,V)\).

First, apply \cref{lem:bosonic-covariance-multiplicative} to the \(M\) copies
of \(\sigma\), with accuracy parameter \(\eps\) and failure probability
\(\delta/2\).  With probability at least \(1-\delta/2\), this outputs a
covariance matrix \(\widehat V\) such that
\begin{equation}
\label{eq:nonzero-cov-good-event-fidelity}
    c_\eps^{-1}V
    \preceq
    \widehat V
    \preceq
    c_\eps V,
    \qquad
    F\bigl(\rho(0,V),\rho(0,\widehat V)\bigr)^2
    \geq
    1-\eps^2 .
\end{equation}
Since $c_\eps=1+\mathcal O(\eps)$, the Loewner estimate can be rewritten as
\[
    (1-f_\eps)V
    \preceq
    \widehat V
    \preceq
    (1+f_\eps)V
\]
for some \(f_\eps=\mathcal O(\eps)\).

We now estimate the first moment.  Condition on the event
\eqref{eq:nonzero-cov-good-event-fidelity}.  On each of the \(M\) copies of
\(\tau\), perform the generaldyne measurement with seed covariance
\(\widehat V\).  By \cref{lem:generaldyne-measurement}, this produces
independent samples \(X_1,\ldots,X_M\in\R^{2m}\) from the Gaussian distribution
with mean \(\sqrt2\mu\) and covariance matrix
\[
    \Sigma
    =
    \frac{V+\widehat V}{2}.
\]
Define
\[
    \overline X
    \coloneqq
    \frac1M\sum_{j=1}^M X_j.
\]
A standard concentration bound for the empirical mean of a Gaussian vector (see e.g.~\cite[Lemma~9]{bittel2025energy}) gives,
with conditional probability at least \(1-\delta/2\),
\[
    \left\|
        \Sigma^{-1/2}
        \left(
            \overline X-\sqrt2\mu
        \right)
    \right\|_2
    \leq
    \frac{\sqrt{2m}+\sqrt{2\log(4/\delta)}}{\sqrt M}.
\]
On the event \eqref{eq:nonzero-cov-good-event-fidelity}, we have
\[
    \Sigma
    =
    \frac{V+\widehat V}{2}
    \preceq
    \frac{1+c_\eps}{2}V .
\]
By setting
\[
    \widehat\mu
    \coloneqq
    \frac{\overline X}{\sqrt2},
\]
we obtain
\[
    \left\|
        V^{-1/2}
        \left(
            \widehat\mu-\mu
        \right)
    \right\|_2
    \leq
    \frac12\sqrt{1+c_\eps}\,
    \frac{\sqrt{2m}+\sqrt{2\log(4/\delta)}}{\sqrt M}.
\]
Since $M
    =
    \mathcal O\!\left(
        \frac{m^2+\log(\delta^{-1})}{\eps^2}
    \right)$ with the implicit constant chosen large enough, we obtain
\[
    (\widehat\mu-\mu)^{\mathsf T}
    V^{-1}
    (\widehat\mu-\mu)
    =
    \left\|
        V^{-1/2}(\widehat\mu-\mu)
    \right\|_2^2
    =
    \mathcal O\!\left(
        \eps^2\,
        \frac{m+\log(\delta^{-1})}{m^2+\log(\delta^{-1})}
    \right).
\]
Since \(m\geq 1\), this is in particular always \(\mathcal O(\eps^2)\). Hence
\[
    (\widehat\mu-\mu)^{\mathsf T}
    V^{-1}
    (\widehat\mu-\mu)
    =
    \mathcal O(\eps^2).
\]
By the union bound, the covariance estimate and the first-moment estimate hold
simultaneously with probability at least \(1-\delta\).

It remains to prove the fidelity guarantee.  By the Gaussian fidelity formula
for states with the same covariance matrix~\cite{Banchi2015},
\[
    F\bigl(\rho(\mu,V),\rho(\nu,V)\bigr)^2
    =
    \exp\!\left(
        -\frac12(\mu-\nu)^{\mathsf T}V^{-1}(\mu-\nu)
    \right).
\]
Using \(1-e^{-x}\leq x\), the first-moment estimate gives
\[
    1-
    F\bigl(\rho(\mu,V),\rho(\widehat\mu,V)\bigr)^2
     \leq
    \frac12
    (\widehat\mu-\mu)^{\mathsf T}
    V^{-1}
    (\widehat\mu-\mu)
    =
    \mathcal O(\eps^2).
\]
On the other hand, by \eqref{eq:nonzero-cov-good-event-fidelity} and invariance
of fidelity under displacement unitaries,
\[
    1-
    F\bigl(\rho(\widehat\mu,V),\rho(\widehat\mu,\widehat V)\bigr)^2
    =
    \mathcal O(\eps).
\]
By the triangle inequality for the purified distance
\[
    P(\omega,\zeta)
    \coloneqq
    \sqrt{1-F(\omega,\zeta)^2},
\]
we have
\begin{align*}
    P\bigl(\rho(\mu,V),\rho(\widehat\mu,\widehat V)\bigr)
    &\leq
    P\bigl(\rho(\mu,V),\rho(\widehat\mu,V)\bigr)
    +
    P\bigl(\rho(\widehat\mu,V),\rho(\widehat\mu,\widehat V)\bigr) \\
    &=
    \sqrt{
        1-F\bigl(\rho(\mu,V),\rho(\widehat\mu,V)\bigr)^2
    }
    +
    \sqrt{
        1-F\bigl(\rho(\widehat\mu,V),\rho(\widehat\mu,\widehat V)\bigr)^2
    } \\
    &=
    \mathcal O(\eps).
\end{align*}
Thus
\[
    1-
    F\bigl(\rho(\mu,V),\rho(\widehat\mu,\widehat V)\bigr)^2
    =
    \mathcal O(\eps^2).
\]
By increasing the implicit universal constant in the choice $M
    =
    \mathcal O\!\left(
        \frac{m^2+\log(\delta^{-1})}{\eps^2}
    \right)$, we can make the implicit constant in the bound above smaller than $1$. This proves all the claimed guarantees.
\end{proof}

\subsection{Sample complexity lower bound}
We will now show a tight sample complexity lower bound for estimating pure Gaussian states (which then also serves as a tight lower bound for mixed Gaussian states). The approach is the same as in the fermionic case in \cref{thm:tomo fermion lower bound}, following the strategy of \cite{harrow2013church,scharnhorst2025optimal}.
There is one subtlety: we cannot follow exactly the same argument since we cannot draw a uniformly random bosonic Gaussian state. The solution is to draw a random Gaussian state according to a probability distribution $Q$ defined by its overlap with the vacuum state.
To be precise,
\begin{align}\label{eq:define Q}
    \dd Q(\psi) = d_{2m+1,m}  \abs{\braket{\psi | 0}}^{4m+2}  \dd \psi
\end{align}
defines a probability measure, since
\begin{align*}
    \int_{\cG_m^0} \dd Q(\psi) &= d_{2m+1,m} \int_{\cG_m^0} \abs{\braket{\psi | 0}}^{4m+2}  \dd \psi \\
    &= d_{2m+1,m}\int_{\cG_m^0} \tr[ \proj{\psi}^{\ot (2m+1)} \proj{0}^{\ot (2m+1)} ]\\
    &= \tr[\Pi_{2m+1,m} \proj{0}^{\ot (2m+1)}] = 1
\end{align*}
using \cref{eq:projection random gaussian}.

Now consider an arbitrary POVM~$\mu$, with probability distribution $P_{\psi}$ for the measurement outcomes~given~by
\begin{align*}
    \dd P_{\psi}(\phi) = \tr[\dd \mu(\phi) \proj{\psi}^{\ot n}]
\end{align*}
when applied to $n$ copies of the state~$\psi$.
Given such a measurement, we now bound its average performance when the target state is chosen according to $Q$.

\begin{lem}\label{lem:moments fidelity boson}
    Suppose that $\psi$ is distributed according to $Q$ as in \cref{eq:define Q}, and $\hat{\psi}$ denotes the estimate from the measurement~$\mu$ on $n$ copies of $\psi$.
    Then,
    \begin{align*}
        \Ex_{\psi \sim Q} \, \Ex_{\hat\psi \sim P_{\psi}} \, \abs{\braket{\psi | \hat\psi}}^{2k} \leq \frac{d_{n + 2m + 1,m}}{d_{n+ 2m + 1 + k,m}},
    \end{align*}
    with $d_{n,m}$ given in \cref{lem:formal degree n copies}.
\end{lem}

\begin{proof}
    We compute
    \begin{align*}
        \Ex_{\psi \sim Q} \, \Ex_{\hat\psi \sim P_{\psi}} \, \abs{\braket{\psi | \hat\psi}}^{2k}
        &= d_{2m+1,m} \int_{\cG_m^0} \int_{\cG_m^0} \tr\bigl[ \dd \mu(\hat\psi) \proj{\psi}^{\ot n} \bigr] \, \, \abs{\braket{\psi | 0}}^{4m+2} \abs{\braket{\psi | \hat\psi}}^{2k} \dd \psi \\
        &= d_{2m+1,m}\int_{\cG_m^0} \int_{\cG_m^0} \tr\bigl[ (\dd \mu(\hat\psi) \ot \proj{0}^{\ot(2m+1)} \proj{\hat{\psi}}^{\ot k}) \proj{\psi}^{\ot(n+2m+1+k)} \bigr] \dd \psi \\
        &= \frac{d_{2m+1,m}}{d_{n+2m+1+k,m}} \int_{\cG_m^0} \tr\bigl[ (\dd \mu(\hat\psi) \ot \proj{0}^{\ot(2m+1)} \ot \proj{\hat{\psi}}^{\ot k}) \Pi_{n+2m+1+k,m} \bigr],
    \end{align*}
    where in the last equality we use \cref{eq:projection random gaussian}.
    We now use $\Pi_{n+2m+1+k,m} \leq \Pi_{n+2m+1,m} \ot I$ to upper bound
    \begin{align*}
        &\tr\bigl[ (\dd \mu(\hat\psi) \ot \proj{0}^{\ot(2m+1)} \ot \proj{\hat{\psi}}^{\ot k}) \Pi_{n+2m+1+k,m} \bigr] \\
        &\qquad \leq \tr\bigl[ (\dd \mu(\hat\psi) \ot \proj{0}^{\ot(2m+1)} \ot \proj{\hat{\psi}}^{\ot k}) (\Pi_{n+2m+1,m} \ot I) \bigr]\\
        &\qquad = \tr\bigl[ (\dd \mu(\hat\psi) \ot \proj{0}^{\ot(2m+1)}) \Pi_{n+2m+1,m}\bigr]
    \end{align*}
    The result now follows from
    \begin{align*}
    \Ex_{\psi \sim Q} \, \Ex_{\hat\psi \sim P_{\psi}} \, \abs{\braket{\psi | \hat\psi}}^{2k}
    &\leq \frac{d_{2m+1,m}}{d_{n+2m+1+k,m}} \int_{\cG_m^0} \tr\bigl[ (\dd \mu(\hat\psi) \ot \proj{0}^{\ot(2m+1)}) \Pi_{n+2m+1,m}\bigr] \\
    &= \frac{d_{2m+1,m}}{d_{n+2m+1+k,m}} \tr\bigl[ (I \ot \proj{0}^{\ot(2m+1)}) \Pi_{n+2m+1,m}\bigr] \\
    &= \frac{d_{n+2m+1,m}d_{2m+1,m}}{d_{n+2m+1+k,m}} \int_{\cG_m^0} \tr\bigl[ (I \ot \proj{0}^{\ot(2m+1)}) \proj{\psi}^{\ot(n+2m+1)} \bigr] \dd \psi\\
    &= \frac{d_{n+2m+1,m}d_{2m+1,m}}{d_{n+2m+1+k,m}} \int_{\cG_m^0} \tr\bigl[ \proj{0}^{\ot(2m+1)} \proj{\psi}^{\ot(2m+1)} \bigr] \dd \psi \\
    &= \frac{d_{n+2m+1,m}}{d_{n+2m+1+k,m}} \tr\bigl[ \proj{0}^{\ot(2m+1)} \Pi_{2m+1,m} \bigr] = \frac{d_{n+2m+1,m}}{d_{n+2m+1+k,m}}
    \end{align*}
    where we use that $\mu$ is a POVM, and apply \cref{eq:projection random gaussian} three times.
\end{proof}

We now use this to bound the required number of copies for tomography of Gaussian bosonic states, proving \cref{thm:tomo boson lower bound}.

\begin{proof}[Proof of \cref{thm:tomo boson lower bound}]
    Let $\mu$ be an $n$-copy measurement, for which we assume that for any $\psi \in \cG_m^0$, the resulting estimate $\hat{\psi}$ satisfies $\abs{\braket{\psi|\hat\psi}}^2 \geq 1 - \eps^2$ with probability at least 0.99.
        Together with the fact that $\abs{\braket{\psi | \hat\psi}}^{2k} \geq 0$, we have that for every $\psi \in \cG_m^0$,
    \begin{align*}
        \Ex_{\hat\psi \sim P_{\psi}} \, \abs{\braket{\psi | \hat\psi}}^{2k} \geq  0.99(1-\eps^2)^{k} .
    \end{align*}
    On the other hand, by \cref{lem:moments fidelity boson} and \cref{lem:formal degree n copies} we have
    \begin{align*}
        \Ex_{\psi \sim Q} \, \Ex_{\hat\psi \sim P_{\psi}} \, \abs{\braket{\psi | \hat\psi}}^{2k} &\leq \frac{d_{n + 2m+1,m}}{d_{n+2m+1+k,m}}
        = \prod_{1 \leq i \leq j \leq m} \frac{n + 2m + 1 - (i+j)}{n + 2m + 1 + k - (i+j)}\\
        &= \prod_{1 \leq i \leq j \leq m} \left(1 - \frac{k}{n + 2m + 1 + k - (i+j)}\right)
        \leq \left(1 - \frac{k}{n + 2m + 1 + k}\right)^{m(m+1)/2}.
    \end{align*}
    We conclude that for every $k$
    \begin{align*}
        0.99(1-\eps^2)^{k} \leq \Ex_{\psi \sim Q} \, \Ex_{\hat\psi \sim P_{\psi}} \, \abs{\braket{\psi | \hat\psi}}^{2k} \leq \left(1 - \frac{k}{n + 2m + 1 + k}\right)^{m(m+1)/2}.
    \end{align*}
    Using $1 + xy \leq (1+x)^y \leq e^{xy}$, which holds for $x \geq -1$, we deduce that
    \begin{align*}
        0.99(1-k\eps^2) \leq \exp\mleft( -\frac{km(m+1)}{2(n + 2m + 1 + k)}\mright).
    \end{align*}
    We now choose $k = \lfloor 1/(4\eps^2) \rfloor$, so $0.99(1-k\eps^2) \geq e^{-1/2}$ and taking logarithms we get
    \begin{align*}
        \frac{km(m+1)}{2(n + 2m + 1 + k)} \leq \frac12,
    \end{align*}
    which we can rewrite as
    \begin{equation*}
        n \geq k(m^2 + m - 1) - 2m - 1 = \Omega(m^2/\eps^2).
        \qedhere
    \end{equation*}
    This proves the result for constant probability of error $\delta$. The scaling of the lower bound with $\delta$ follows from the reduction to distinguishing two nearby states as in \cref{eq:distinguising two states}.
\end{proof}

\section{Random purification of bosonic gauge-invariant states}
We have shown that the covariant measurement, together with a variant of the random purification channel, can be used to derive sample-optimal tomography for bosonic Gaussian states.
However, the random purification channel itself is not well-defined in this setting.
We now prove \cref{cor:bosons}, which asserts the existence of a random purification channel for \emph{gauge-invariant} bosonic Gaussian states.
To that end, we formulate an appropriate infinite-dimensional variant of \cref{thm:main simplified,thm:main technical}, in the special case where we have a decomposition where each summand is finite-dimensional.
The reason that this is the case in contrast to general Gaussian states, is that the group of passive Gaussian unitaries is compact (and in fact it is a maximal compact subgroup of the group of all Gaussian unitaries).
Alternative derivations can be found in \cite{mergedwork1,mergedwork2}.

\subsection{The random purification channel for infinite dimensions}\label{sec:passive bosons}

We consider a slight generalization of \cref{thm:main technical} to the case where the Hilbert space is infinite dimensional, and the isotypical decomposition is an infinite direct sum, but all irreducible representation spaces and multiplicity spaces are finite dimensional. This is exactly the situation that occurs for passive bosonic Gaussian unitaries.

Throughout this subsection, let $\cH$ be a separable Hilbert space with a fixed orthonormal basis, and let $\tilde{\cH}$ be a copy of $\cH$ with the corresponding primed basis. If $\rho\in \trcl(\cH)$ is positive, its standard purification is the vector
\begin{align*}
\ket{\psi^\text{std}_\rho}
=
\sum_x \sqrt{\rho}\ket{x}\otimes \ket{\tilde{x}}
\in \cH\otimes \tilde{\cH},
\end{align*}
where the series converges in norm because $\sqrt{\rho}$ is Hilbert--Schmidt.

We now consider the following set-up.
We assume there exists a decomposition of the Hilbert space
\begin{align}\label{eq:decomposition infinite}
    \cH \cong \bigoplus_{\lambda \in \Lambda} \cL_\lambda \ot \cR_\lambda,
\end{align}
where~$\Lambda$ is a countable index set and $\cL_\lambda$ and $\cR_\lambda$ are finite dimensional Hilbert spaces.
We let
\begin{align}\label{eq:decomposition algebras infinite}
   \cA \cong \bigoplus_{\lambda \in \Lambda}  I_{\cL_\lambda} \ot \Lin(\cR_\lambda),
   \quad\text{and}\quad
   \cB \cong \bigoplus_{\lambda \in \Lambda} \Lin(\cL_\lambda) \ot I_{\cR_\lambda} ,
\end{align}
where the direct sum is the direct sum of von Neumann algebras, so $\cA$ consists of operators $(I_{\cL_\lambda} \ot A_\lambda)_{\lambda \in \Lambda}$ with $A_\lambda \in \Lin(\cR_\lambda)$ and $\sup_\lambda \norm{A_\lambda} < \infty$.
The algebras $\mathcal{A}$ and $\mathcal{B}$ are von Neumann algebras contained in the algebra of bounded operators $B(\cH)$. We recall that the commutant of a set of operators $M\subseteq \mathcal{B}(\mathcal{H})$ is the set of bounded operators that commute with $M$, denoted as $M'$. We also recall the double commutant theorem: $M''$ is a von Neumann algebra, and it is the smallest von Neumann algebra containing $M$.
We will denote by $P_\lambda$ the projections onto the components $\cL_\lambda \ot \cR_\lambda$. They are central projections in both $\cA$ and $\cB$. By restricting operators on $\mathcal{A}$ and $\mathcal{B}$ to fixed $\lambda$, it is clear that $\cB = \cA'$ and $\cA'' = \cA$.

We define
\begin{align}\label{eq:E infinite}
    \E_\cA[M] \cong \bigoplus_{\lambda \in \Lambda}  \frac{I_{\cL_\lambda}}{\dim \cL_\lambda} \ot \tr_{\cL_\lambda}[ P_\lambda M P_\lambda ]
    \qquad (M \in B(\cH)),
\end{align}
Alternatively, the orthogonal projection can be written as follows:
take any closed compact subgroup~$G \subseteq \U(\cH)$ that satisfies $G'' = \cB$ (or equivalently, by the double commutant theorem, satisfies that the weak operator closure of its span equals $\cB$).
Then the orthogonal projection is given by the Haar twirl over~$G$:
\begin{align}\label{eq:twirl infinite}
    \E_\cA[M] = \int_G g M g^\dagger \,\dd g.
\end{align}

For simplicity of exposition, we will assume that $\cA$, and hence $\cB$, are closed under transposes, in particular $\E_\cA[M] = \E_{\cA^{\tran}}[M]$.
The key point is now that since all the spaces $\cL_\lambda$ and $\cR_\lambda$ are finite-dimensional, maximally mixed states and maximally entangled states on them are well-defined, and hence we can construct the random purification channel in the same way as in the finite-dimensional case.

\begin{thm}[Random purification for infinite direct sums]\label{thm:infinite-random-purification}
Let $\cH$ and $\cA$ be as in \eqref{eq:decomposition infinite}, \eqref{eq:decomposition algebras infinite}.
Then there is a quantum channel
$\cP_\cA \colon \trcl(\cH) \to \trcl(\cH \ot \tilde{\cH})$
with the following properties:
\begin{enumerate}
\item \emph{Random purification of invariant input states:}
For any input state~$\rho \in \S(\cH)$ that commutes with~$\cA$, i.e., $\rho \in \cB$, we have
\begin{align}\label{eq:sym action infinite}
    \cP_\cA[\rho]
= \parens*{ \idCh \ot \E_{\cA} }[ \psi^\text{std}_\rho ]
= \int_G (I \ot g) \psi^\text{std}_\rho (I \ot g^\dagger) \,\dd g,
\end{align}
where $G \subseteq \U(\cH)$ is any closed compact subgroup with~$G'' = \cB$.
\item \emph{Explicit formulas:}
For all $\rho \in \S(\cH)$,
\begin{align}\label{eq:explicit infinite}
    \cP_\cA[\rho]
&\cong \bigoplus_{\lambda\in\Lambda} \tr_{\cR_\lambda}[P_\lambda \rho P_\lambda] \ot \frac{I_{\tilde{\cL}_\lambda}}{\dim \cL_\lambda} \ot \frac{\ket{\Gamma}_{\cR_\lambda\tilde{\cR}_\lambda}\bra{\Gamma}_{\cR_\lambda\tilde{\cR}_\lambda}}{\dim\cR_\lambda}  \\
&= \sum_k \sqrt{\cP_\cA[Q_k]} \parens*{ \rho \ot I } \sqrt{\cP_\cA[Q_k]},
\end{align}
where the $Q_k$ are any collection of finite-rank central projections satisfying $\sum_k Q_k = I$.
The sums converge in trace norm.
\end{enumerate}
\end{thm}

There is one minor difference, compared to \cref{thm:main technical}: in \cref{eq:explicit} an expression for the random purification channel is derived in terms of $\cP_{\cA}[I]$; this is not necessarily a bounded operator when $\cH$ is infinite dimensional, so we break this up in the operators $\cP_\cA[Q_k]$.

\begin{proof}
The proof is the same as that of \cref{thm:main technical}. As before, we identify $\cH$ and $\tilde\cH$ with their decompositions.
We define $\cP_\cA$ through
\begin{align*}
    \cP_\cA[\rho]
&= \bigoplus_{\lambda\in\Lambda}  \tr_{\cR_\lambda}[P_\lambda \rho P_\lambda] \ot \frac{I_{\tilde{\cL}_\lambda}}{\dim \cL_\lambda} \ot \tfrac1{\dim\cR_\lambda} \ket{\Gamma}_{\cR_\lambda\tilde{\cR}_\lambda}\bra{\Gamma}_{\cR_\lambda\tilde{\cR}_\lambda} \, .
\end{align*}
It is easy to see that $\cP_\cA$ defines a quantum channel, and that the direct sum converges in trace norm.
We now discuss its interpretation as a random purification channel. Let $\rho$ be a state in $\cB$, then by the structure of $\cB$ we may write (with convergence in trace norm)
\begin{align}\label{eq:rho-block-form}
\rho
\cong
\bigoplus_\lambda \rho_\lambda \ot I_{\cR_\lambda} ,
\qquad
\rho_\lambda = \frac{1}{\dim \cR_\lambda}\tr_{\cR_\lambda}
\mleft[ P_\lambda \rho P_\lambda \mright].
\end{align}
With respect to the bases induced by \eqref{eq:decomposition infinite}, the standard purification decomposes as
\begin{align}\label{eq:std-purification-infinite-direct-sum}
\ket{\psi^\text{std}_\rho}
=
\bigoplus_{\lambda\in\Lambda}
\ket{\psi^\text{std}_{\rho_\lambda}}_{\cL_\lambda\tilde \cL_\lambda} \ot \ket{\Gamma}_{\cR_\lambda\tilde{\cR}_\lambda} ,
\end{align}
where $\ket{\psi^\text{std}_{\rho_\lambda}}$ is defined by the same formula as above, now on the finite-dimensional space $\cL_\lambda$. The sum in \eqref{eq:std-purification-infinite-direct-sum} converges in norm because its squared norm is $\sum_\lambda \dim \cR_\lambda \tr[\rho_\lambda]=1$.

From \cref{eq:E infinite}, we see that
\begin{align*}
    \parens*{ \idCh \ot \E_{\cA} }\mleft[ \proj{\psi^\text{std}_\rho} \mright]
&= \parens*{ \idCh \ot \E_{\cA^{\tran} }}\mleft[ \proj{\psi^\text{std}_\rho} \mright] \\
&= \bigoplus_{\lambda\in\Lambda}  \tr_{\tilde{\cL}_\lambda}[\ket{\psi_\lambda}_{\cL_\lambda\tilde{\cL}_\lambda}\bra{\psi_\lambda}_{\cL_\lambda\tilde{\cL}_\lambda}] \ot \frac{I_{\tilde{\cL}_\lambda}}{\dim \cL_\lambda} \ot \ket{\Gamma}_{\cR_\lambda\tilde{\cR}_\lambda}\bra{\Gamma}_{\cR_\lambda\tilde{\cR}_\lambda}\\
&= \bigoplus_{\lambda\in\Lambda} \rho_{\lambda} \ot \frac{I_{\tilde{\cL}_\lambda}}{\dim \cL_\lambda} \ot \ket{\Gamma}_{\cR_\lambda\tilde{\cR}_\lambda}\bra{\Gamma}_{\cR_\lambda\tilde{\cR}_\lambda} \\
&= \bigoplus_{\lambda\in\Lambda} \tr_{\cR_\lambda}[P_\lambda \rho P_\lambda] \ot \frac{I_{\tilde{\cL}_\lambda}}{\dim \cL_\lambda} \ot  \frac{\ket{\Gamma}_{\cR_\lambda\tilde{\cR}_\lambda}\bra{\Gamma}_{\cR_\lambda\tilde{\cR}_\lambda}}{\dim\cR_\lambda} = \cP_\cA[\rho].
\end{align*}

Finally, consider a collection of projection operators $Q_k \in \cA \cap \cB$ with finite rank. These are finite sums of the $P_\lambda$.
Suppose that $\sum_k Q_k = I$.
The $Q_k$ are trace class, so $\cP_\cA[Q_k]$ is well-defined.
Suppose that $Q_k = \sum_{\lambda \in \Lambda_k} P_\lambda$ for some finite set $\Lambda_k$, then it follows directly from \cref{eq:explicit infinite} that
\begin{align*}
\sqrt{\cP_\cA[Q_k]} = \bigoplus_{\lambda \in \Lambda_k} I_{\cL_\lambda} \ot I_{\tilde{\cL}_\lambda}  \frac{\ket{\Gamma}_{\cR_\lambda\tilde{\cR}_\lambda}\bra{\Gamma}_{\cR_\lambda\tilde{\cR}_\lambda}}{\sqrt{\dim\cL_\lambda \dim\cR_\lambda}}
\end{align*}
and hence
\begin{align*}
 \sqrt{\cP_\cA[Q_k]} \parens*{ \rho \ot I } \sqrt{\cP_\cA[Q_k]} &= \bigoplus_{\lambda\in\Lambda_k}  \tr_{\cR_\lambda}[P_\lambda \rho P_\lambda] \ot \frac{I_{\tilde{\cL}_\lambda}}{\dim \cL_\lambda} \ot \tfrac1{\dim\cR_\lambda} \ket{\Gamma}_{\cR_\lambda\tilde{\cR}_\lambda}\bra{\Gamma}_{\cR_\lambda\tilde{\cR}_\lambda}
\end{align*}
from which it follows that
\begin{align*}
\cP_{\cA}[\rho] = \sum_k \sqrt{\cP_\cA[Q_k]} \parens*{ \rho \ot I } \sqrt{\cP_\cA[Q_k]}.
\end{align*}
\end{proof}

\subsection{Application to passive bosonic Gaussian states}

We now apply \cref{thm:infinite-random-purification} to passive bosonic Gaussian states. Let
\begin{align}\label{eq:identification L2 copies}
\cH_{m,n}
=
\cH_m^{\otimes n}
\cong
\L^2(\R^{mn})
\end{align}
be the Hilbert space of $n$ copies of an $m$-mode bosonic system. The passive group $\U(n)\times \U(m)$ acts on $\cH_{m,n}$ by number-preserving Gaussian unitaries. There are two commuting subgroups
\begin{align*}
G_1
\cong
{I_n\otimes \U(m)}
\cong
\U(m),
\qquad
G_2
\cong
{\U(n)\otimes I_m}
\cong
\U(n).
\end{align*}
In the Fock basis, the passive representation is closed under transpose, since $U_u\tran=U_{u\tran}$.
Note that $U_{I_n \ot u} \in G_1$ corresponds to the action of $U_u^{\ot n}$ under the identification in \cref{eq:identification L2 copies}.
Bosonic Howe duality states that the representation of $\U(n)\times \U(m)$ decomposes as
\begin{align}\label{eq:howe-passive-decomposition-for-rp}
\cH_{m,n}
\cong
\bigoplus_{N=0}^\infty
\bigoplus_{\lambda\vdash N}
\cV_\lambda
\otimes
\cW_\lambda,
\end{align}
where the sum is over partitions $\lambda$ of $N$ with $l(\lambda)\leq \min(m,n)$, and where $N$ is the particle number. The group $G_2\cong \U(n)$ acts irreducibly on $\cW_\lambda$, and the group $G_1\cong \U(m)$ acts irreducibly on $\cV_\lambda$.

We apply \cref{thm:infinite-random-purification} to the compact group $G_1$, where the irreducible representation spaces are $\cV_\lambda$, and the multiplicity spaces are $\cW_\lambda$. All of these spaces are finite dimensional.
Therefore \cref{thm:infinite-random-purification} gives the following corollary.

\begin{cor}[Passive bosonic random purification]\label{cor:passive-bosonic-random-purification}
There is a quantum channel
\begin{align*}
\cP^\mathrm{pass}_{n,m}
\colon
\trcl(\cH_m^{\otimes n})
\to
\trcl(\cH_m^{\otimes n}\otimes \cH_m^{\otimes n})
\end{align*}
such that, for every state $\rho_{A^n} \in \S(\cH_{m,n})$ that commutes with $U_{u\otimes \id_m}$ for all $u\in \U(n)$
we have
\begin{align}\label{eq:passive-rp-un-invariant-state}
\cP^\mathrm{pass}_{n,m}[\rho_{A^n}]
=
\int_{\U(m)}
(I_{A^n}\otimes U_v^{\otimes n})
\psi^\text{std}_{\rho_{A^n}}
(I_{A^n}\otimes U_v^{\dagger, \otimes n})
\dd v .
\end{align}
In particular, if $\sigma_A$ is an $m$-mode passive Gaussian state, then
\begin{align}\label{eq:passive-rp-iid-gaussian-state}
\cP^\mathrm{pass}_{n,m}[\sigma_A^{\otimes n}]
=
\int_{\U(m)}
\mleft(
(I_A\otimes U_v)
\psi^\text{std}_{\sigma_A}
(I_A\otimes U_v^\dagger)
\mright)^{\otimes n}
\dd v .
\end{align}
Thus the channel maps $n$ copies of $\sigma_A$ to $n$ copies of the same Haar-random passive Gaussian purification of $\sigma_A$, each having expected number of particles twice that of $\sigma_A$
\end{cor}
Note that \cref{cor:bosons} is a special case of the above result.

\begin{proof}
The first statement is exactly \cref{thm:infinite-random-purification} applied to $G_1$. With $\mathcal{A}$ and $\mathcal{B}$ as in \eqref{eq:decomposition algebras infinite}, by Schur's lemma applied to projections on finite direct sums of irreps, the commutant of $G_1$ is $\mathcal{A}$. It follows that $G_1''=\mathcal{B}$, therefore by the double commutant theorem, the strong operator topology closure of the span of $G_1$ is $\mathcal{B}$. In particular, $G_2$-invariant states belong to $\mathcal{B}$. 

It remains to check that $\sigma_A^{\otimes n}$ commutes with $G_2$, when $\sigma_A$ is passive Gaussian. It is sufficient to recall the identity $e^{-x}=\frac{1}{\pi}\int_{-\infty}^{\infty}\frac{1}{1+t^2}e^{-itx} dt$ for $x>0$. Then, for any gauge-invariant positive-definite Hamiltonian $H$, since $(e^{-itH})^{\ot n} \in G_1$ for all $t \in \R$, the above formula interpreted as convergence of Riemann sums, shows that $(e^{-H})^{\ot n}$ is in the weak operator topology closure of $G_1$, and thus in the double commutant $G_1''=\mathcal{B}$. 
Alternatively, one can explicitly check that $\sigma_A^{\ot n}$ commutes with the action of $G_2$, see~\cref{passive_gaussian} below.

Finally, with respect to the product Fock basis,
\begin{align*}
\psi^\text{std}_{\sigma_A^{\otimes n}}
=
\mleft(\psi^\text{std}_{\sigma_A}\mright)^{\otimes n}.
\end{align*}
Substituting this into \eqref{eq:passive-rp-un-invariant-state} gives \eqref{eq:passive-rp-iid-gaussian-state}. By \cref{lem:gaussian_purification_bos} the standard purification of a passive Gaussian state is Gaussian, with twice the expected number of particles, and passive Gaussian unitaries preserve Gaussianity and the total mean photon number. Therefore each state inside the integral in \eqref{eq:passive-rp-iid-gaussian-state} is a Gaussian purification of $\sigma_A$, with mean photon number equal to twice that of $\sigma_A$.
\end{proof}

We can also express this random purification channel differently. Let $Q_k$ denote the projection onto the $k$-particle subspace of $\cH_{m,n}$. This is a finite rank projector that commutes with $G_1$ and $G_2$. Additionally, $\sum_{k=0}^{\infty} Q_k = I$, so by \cref{thm:infinite-random-purification}, for any state $\rho$ such that $[\rho, G_2] = 0$ we can write
\begin{align*}
\cP^\mathrm{pass}_{n,m}(\rho_{A^n})
= \sum_k \sqrt{\cP^\mathrm{pass}_{n,m}[Q_k]} \parens*{ \rho \ot I } \sqrt{\cP^\mathrm{pass}_{n,m}[Q_k]}.
\end{align*}

We conclude with an alternative elementary proof that multi-copy passive Gaussian states are $\U(n)$-invariant.
\begin{lem}\label{passive_gaussian}
    Let $\rho^{\otimes n}$ be the $n$-fold tensor product of an $m$-mode passive Gaussian state $\rho$. Then, for any $u\in \U(n)$, it holds that
    \begin{align}\label{eq:U_o_rho}
    U_{u\otimes\id_m}\rho^{\otimes n}U^\dagger_{u\otimes\id_m}=\rho^{\otimes n}.
\end{align}
\end{lem}

\begin{proof}
    Since $\rho$ is a passive Gaussian state, it can be written as $\rho = U_{v}\tau U_{v}^\dagger$ for a unitary matrix $v\in \U(m)$, where $\tau = \tau_1 \ot \dots \tau_m$, where $\tau_i$ is a single-mode Gaussian state with Hamiltonian proportional to the number operator.
    Then, for $u\in \U(n)$
    \begin{align*}
        U_{u\otimes I_m}\rho^{\otimes n}U^\dagger_{u\otimes I_m} &= U_{u\otimes I_m}\left(U_{v}\tau U_{v}^\dagger\right)^{\otimes n}U^\dagger_{u\otimes I_m} \\
        &= U_{u\otimes I_m} U_{I \ot v} \tau ^{\otimes n} U_{I \ot v}^\dagger U^\dagger_{u\otimes I_m}.
    \end{align*}
    Now, we note that $U_{u\otimes I_m}$ and $U_{I \ot v}$ commute, and
    \begin{align*}
        U_{u\otimes I_m} \tau ^{\otimes n} U^\dagger_{u\otimes I_m} = \bigotimes_{i=1}^m U_u \tau_i^{\ot n} U_u^\dagger,
    \end{align*}
    but $\tau_i^{\ot n}$ is a Gibbs state for a Hamiltonian proportional to the number operator on $n$ modes, and thereby invariant under the number-preserving Gaussian unitary $U_u$.
    From this we conclude \cref{eq:U_o_rho} holds.
\end{proof}

\section{Haar measure of \texorpdfstring{$\mathrm{Sp}(2m,\RR)$}{Sp(2m)}}\label{sec:explicit haar}

An alternative way to compute the formal dimension $d_{n,m}$ in~\cref{eq:vacuum_formaldegree} is through explicit integration. While we recalled the general theory to argue that the formal degrees are bounded from below, an explicit integration approach is sufficient if we just want to compute the formal dimension of the vacuum representation. However, to give a formula for the Haar measure, we need to use a different root system than the one given in~\cref{eq:rootsys}. Let $\mathfrak{a}$ be a maximal subalgebra of the non-compact algebra $\mathfrak{p}$. The root system constructed from $\mathfrak{g}$ and $\mathfrak{a}$ (without complexification) is called \textit{restricted root system}, with root spaces
\[
\mathfrak{g}_{\mu}\coloneqq \{X\in\mathfrak{g}: [A,X]=\braket{\mu, A}X\, \, \forall A\in \mathfrak{b} \},
\]
and we denote $\beta_{\mu}$ the dimension of $\mathfrak{g}_{\mu}$. Let $\Sigma_+$ be a choice of a positive restricted root system.
For a reductive Lie group with a $G=KAK$ decomposition, where $K$ is a maximal compact subgroup and $A$ is the Lie group $e^{\mathfrak{a}}$, the Haar measure reads (Proposition 5.28 in~\cite{knapp2001representation}):
\begin{equation}
\dd\mu(g)= C \dd k_1 \dd k_2 \dd a \prod_{\mu\in\Sigma_+}|\sinh(\mu(a))|^{\beta_{\mu}},
\end{equation}
where $g=k_1e^ak_2$, and where $\dd k$ is the Haar measure of $K$ and $\dd a$ is Haar measure of $\mathfrak{a}$, $\beta_{\lambda}$ is the dimension of the restricted root space $\mathfrak{g}_{\lambda}$,  $\Sigma_{+}$ is the positive restricted root system and $C$ is some free constant.

Concretely for $\mathrm{Sp}(2m,\RR)$,
a maximal compact subgroup $K$ is $\mathrm{Sp}(2m,\RR)\cap \mathrm{O}(2m)\cong \U(m)$, and the non-compact (Cartan) subalgebra is the form $\{a \, : \, a=\diag(s_1,-s_1,\cdots s_{m}, -s_{m})\}$.
Let $\varepsilon_i\in \mathfrak{b}^{\C}{}' $ be defined from $\braket{\varepsilon_i,a}=s_i$, with the above notation for $a$. 
One can show that the roots are\footnote{note the formal similarity with the root system~\eqref{eq:rootsys}. However, the interpretation of the root spaces is different, with the new root spaces being associated with generators $q_i^2$, $p_i^2$, $q_ip_j$, $p_iq_j$, and only their linear combinations give phase shifters and beam splitters.}
\begin{align}\label{eq:rootsys2}
\pm\varepsilon_{i}\pm\varepsilon_j \,\, \text{ for } i\neq j \quad \text{ and }  \quad \pm 2\varepsilon_k\, \text{ for } k\in[m],
\end{align}
and the root spaces have multiplicity one. A valid choice of positive roots is
\begin{align*}
    \Sigma_+ = \{ \varepsilon_{i}\pm\varepsilon_j \; \text{ for } i < j \} \cup \{2\varepsilon_k,\; \text{ for } k\in[m]\}.
\end{align*}

Inserting the expressions of the roots, we obtain the following formula for the Haar measure at $S=u_1 e^{\diag(s_1,-s_1,\cdots s_{m}, -s_{m})} u_2$, with $s_1\geq s_2\geq\cdots\geq s_m$:
\begin{equation}
\dd\mu(S)=C \dd u_1 \dd u_2 \dd s_1  \cdots ds_{m}\mathbb{1}_{s_1\geq s_2\cdots \geq s_m}\prod_{i=1}^{m}\sinh(2s_i)\prod_{i<j}\abs{\cosh(2s_i)-\cosh(2s_j)}.
\end{equation}
This expression has also been derived in \cite{lupo2012invariant}.

We will now use this expression for the Haar measure to explicitly verify the formal degree.
We have that
\begin{equation}
d_{n,m}^{-1}= c\int_{\mathrm{Sp}(2m,\RR)} \dd\mu(S)  |\bra{0} U_S \ket{0}|^{2n},
\end{equation}
for some constant $c>0$. Clearly, $|\bra{0} U_S \ket{0}|^{2n}$, depends only on $a$, and it is equal to
\begin{align*}
|\bra{0} U_S \ket{0}|^{2n}&=\left(\frac{2^m}{\sqrt{\det (e^{2\diag(s_1,-s_1,\cdots s_{m}, -s_{m})}+I)}}\right)^{n}\\&=\prod_{i=1}^{m}\cosh^{-n}(s_i)=2^{mn/2}\prod_{i=1}^{m}({1+\cosh(2s_i)})^{-n/2}.
\end{align*}
Thus, the integral is

\begin{align}
&\int_{\mathrm{Sp}(2m,\RR)} \dd S  |\bra{0} U_S \ket{0}|^{2n}\\&=2^{mn/2}\int_{0}^{+\infty} \dd s_1  \cdots \dd s_{m}\delta_{s_1\geq s_2\cdots \geq s_m}\prod_{i=1}^{m}\sinh(2s_i)({1+\cosh(2s_i)})^{-n/2}\prod_{i<j}(\cosh(2s_i)-\cosh(2s_j))\\
&=\frac{2^{mn/2}}{2^m m!}\int_1^{+\infty} \dd t_1  \cdots \dd t_{m}\prod_{i=1}^{m}({1+t_i})^{-n/2}\prod_{i<j}|t_i-t_j|\\
&=\frac{2^{m(m-1)/2}}{m!}\int_{0}^{1} \dd x_1  \cdots \dd x_{m}\prod_{i=1}^{m}({1-x_i})^{n/2-m-1}\prod_{i<j}|x_i-x_j|,
\end{align}
where we used the change of variables $t_i=\cosh(2s_i)$ and then $t_{i}=\frac{1+x_{i}}{1-x_i}$. This is an instance of a Selberg integral~\cite{selberg1944remarks,forrester2008importance}, where for parameters $\alpha, \beta, \gamma$ satisfying $$\operatorname{Re}(\alpha) > 0, \, \operatorname{Re}(\beta) > 0, \, \operatorname{Re}(\gamma) > -\min\left\{ \frac{1}{m}, \frac{\operatorname{Re}(\alpha)}{m-1}, \frac{\operatorname{Re}(\beta)}{m-1} \right\},$$
the following integral converges
\begin{align*}
S_m(\alpha, \beta, \gamma) &= \int_0^1 \cdots \int_0^1 \prod_{i=1}^m x_i^{\alpha-1} (1-x_i)^{\beta-1} \prod_{1 \le i < j \le m} |x_i - x_j|^{2\gamma} \, \dd x_1 \cdots \dd x_m \\
&= \prod_{j=0}^{m-1} \frac{\Gamma(\alpha+j\gamma)\Gamma(\beta+j\gamma)\Gamma(1+(j+1)\gamma)}{\Gamma(\alpha+\beta+(m+j-1)\gamma)\Gamma(1+\gamma)}.
\end{align*}
In our case, $\alpha=1$, $\beta=n/2-m$, $\gamma=1/2$. We need $n>2m$ for the integral to be convergent. The final result can be written as
\begin{align*}
d_{n,m}^{-1}=c_{m}\prod_{j=0}^{m-1}\frac{\Gamma(\frac{n}{2}+\frac{j-2m}{2})}{\Gamma(\frac{n}{2}+\frac{j-m+1}{2})},
\end{align*}
where we hid terms independent of $n$ in $c_m$.
Now we observe that, for fixed $j$
\[
\prod_{1\leq i\leq j}(n-i-j)=\frac{(n-j-1)!}{(n-2j-1)!}.
\]
By using the duplication formula of Gamma functions, $\Gamma(z)=\frac{1}{\sqrt{\pi}}2^{z-1}\Gamma\left(\frac{z}{2}\right)\Gamma\left(\frac{z+1}{2}\right)$, we obtain

\[
\frac{(n-j-1)!}{(n-2j-1)!}=2^{j}\frac{\Gamma\left(\frac{n-j}{2}\right)\Gamma\left(\frac{n-j+1}{2}\right)}{\Gamma\left(\frac{n-2j}{2}\right)\Gamma\left(\frac{n-2j+1}{2}\right)}.
\]
By term-by-term cancellation in the full product over $j$, we can get the following simplification:
\[
\prod_{1\leq j\leq m}2^j\frac{\Gamma\left(\frac{n-j}{2}\right)\Gamma\left(\frac{n-j+1}{2}\right)}{\Gamma\left(\frac{n-2j}{2}\right)\Gamma\left(\frac{n-2j+1}{2}\right)}=2^{m(m+1)/2}\frac{\prod_{0\leq j\leq m-1}\Gamma\left(\frac{n-j}{2}\right)}{\prod_{m+1\leq j\leq 2m}\Gamma\left(\frac{n-j}{2}\right)}=2^{m(m+1)/2}\frac{\prod_{0\leq j\leq m-1}\Gamma\left(\frac{n+j-m+1}{2}\right)}{\prod_{0\leq j\leq m-1}\Gamma\left(\frac{n+j-2m}{2}\right)},
\]
establishing the desired correspondence.

\end{appendix}
\end{document}